\documentclass[thm-restate,cleveref]{lipics-v2021}
\usepackage[utf8]{inputenc}
\usepackage{amsmath}
\usepackage{amsfonts}
\usepackage{amsthm}
\usepackage{amsmath}
\usepackage{xspace}
\usepackage{caption}
\usepackage{imakeidx}
\usepackage{thmtools}
\usepackage{todonotes}
\usepackage{booktabs}
\usepackage{ragged2e}
\usepackage{xspace}
\usepackage{adjustbox}
\usepackage{todonotes}
\usepackage{colortbl}
\usepackage{tcolorbox}
\tcbuselibrary{breakable}
\definecolor{mapblue}{RGB}{225,241,252}
\definecolor{mapgreen}{RGB}{228,245,229}
\newtcolorbox{mapchange}[1]{breakable,colback=#1,colframe=#1,boxrule=0pt,
  sharp corners,left=2pt,right=2pt,top=1pt,bottom=1pt,boxsep=1pt}
\usepackage{graphicx}
\usepackage[nocompress]{cite}
\usepackage{tabularray}
\usepackage{silence}
\usepackage{multirow}
\usepackage{xcolor}
\usepackage{capt-of}
\usepackage[dvipsnames]{xcolor}
\usepackage{algorithm}
\usepackage{algpseudocode}

\newtheorem*{hypothesis}{Hypothesis}

\hideLIPIcs

\titlerunning{}
\title{
\texorpdfstring{
Tight Fréchet Bounds for $\lambda$-low-density curves}{Tight Fréchet Bounds for λ-low dense curves}}

\nolinenumbers

\author{Jacobus Conradi}{Department of Computer Science, University of Copenhagen, Denmark}{jaco@di.ku.dk}{https://orcid.org/0000-0002-8259-1187}{}

\author{Ivor van der Hoog}{IT University of Copenhagen, Denmark}{ivva@itu.dk}{
https://orcid.org/0009-0006-2624-0231}{}

\author{Frederikke Uldahl}{IT University of Copenhagen, Denmark}{freu@itu.dk}{https://orcid.org/0009-0003-2279-9239}{}

\author{Eva Rotenberg}{IT University of Copenhagen, Denmark}{erot@itu.dk}{
https://orcid.org/0000-0001-5853-7909}{}

\authorrunning{J. Conradi, I. van der Hoog, F. Uldahl, E. Rotenberg}

\Copyright{J. Conradi, I. van der Hoog, F. Uldahl, E. Rotenberg}

\funding{This research is supported by the VILLUM Foundation grant (VIL37507) ``Efficient Recomputations for Changeful Problems'', and Danmarks Frie Forskningsfond via the grant ``Dynamic Graphs: Distributed and Geometry'' with grant-ID 10.46540/4251-00004B.}

\acknowledgements{We thank Sampson Wong for pointing us to the implications of our results for simple paths to $\lambda$-low density curves.}


\newcommand{\fd}{\operatorname{FD}}                    
\newcommand{\FSD}{\Delta\text{-FSD}}    
\newcommand{\FSS}{\Delta\text{-FSS}}    
\newcommand{\N}{\mathrm{N}}             
\newcommand{\graphFS}{\mathrm{W}}
\newcommand{\graphFSDelta}{\graphFS_{\le\Delta}}

\newcommand{\eps}{\varepsilon}
\newcommand{\bR}{\mathbb{R}}

\newcommand{\band}{\mathcal{B}}

\newcommand{\Psimp}{P^\mu}              
\newcommand{\Qsimp}{Q^\mu}              
\newcommand{\Gsimp}{G^\mu}              

\newcommand{\ApproxDeciderRuntime}[1]{\lambda^{2/d}(n+|E|)^{2-2/d}\log(n+|E|)/(#1)^2}

\newcommand{\Decider}{\textbf{Decider} }
\newcommand{\ApproxDecider}{\textbf{ApproxDecider} }
\newcommand{\TwoApproxDecider}{\textbf{TwoCallsApproxDecider}\xspace}

\newcommand{\walk}[1][G]{\mathcal{W}(#1)}
\newcommand{\mapFD}[3]{\min \limits_{#3 \in \walk[#1]} \fd(#2, #3)}


\begin{document}

\maketitle

\begin{abstract}
We study computing the Fréchet distance between $\lambda$-low-density curves.
Out of all realistic curve assumptions, $\lambda$-low density is the most general: it demands that every ball of radius $r$ intersects at most $\lambda$ edges of length at least $r$. Crucially, grid graphs are $O(1)$-low density.

For the less general realism assumptions such as $c$-packedness and
$\kappa$-straightness, tight bounds for $(1+\eps)$-approximation
are known under the Orthogonal Vector Hypothesis (OVH). This contrasts $\lambda$-low density, which has only been studied under the assumption that $\lambda$ is constant, and which have no previously known lower bounds.
In this paper, we show an  $\widetilde O\left(
 \frac{\lambda^{2/d}n^{2-2/d}}{\varepsilon^{2-4/d}}
 \right)$ time algorithm to compute a $(1+\eps)$-approximation of the Fréchet distance under the $\ell_1$ or $\ell_\infty$ norm between two $\lambda$-low density, complexity $n$ curves in $\mathbb{R}^d$. Assuming OVH, we show that no algorithm can achieve a polynomial
improvement over our running time as a function of $\eps$, $n$,
and $\lambda$.
We also give tight bounds in $n$ and $\lambda$ for approximating
the Fréchet distance under any $\ell_p$ norm and for map matching.
\end{abstract}
\setcounter{page}{0}

\newpage
\section{Introduction}

The Fréchet distance is a classical distance measure between curves.
It can be described as the shortest leash length needed to connect a human walking along one curve and a dog walking along the other, with both moving forward without backtracking.
The Fréchet distance is used in various real-life applications, including 
signature verification\cite{791205}, 
the study of coastline retrievals\cite{mascret2006coastline}, map-matching \cite{1644335},
protein structures \cite{jiang2008protein},
air traffic \cite{bombelli2017strategic},
and trajectory analysis and clustering  \cite{buchin2011detecting, brakatsoulas2005map, buchin2017clustering,konzack2017visual,su2020survey, kenefic2014track, driemel2016clustering, agarwal2018subtrajectory, sung2012trajectory,zhou2008bag, gudmundsson2012motifs, buchin2020improved}.

\vspace{-0.2cm}
\subparagraph{Computing the Fréchet distance.}
 Alt and Godau\cite{alt1995computing} show a $O(nm \log(n + m))$-time algorithm for computing the Fréchet distance between two polygonal curves with, respectively, $n$ and $m$ vertices. Since then, there have been polylogarithmic-time improvements \cite{buchin2017four,doi:10.1137/1.9781611978322.173}. Bringmann's construction~\cite{DBLP:conf/focs/Bringmann14}, when sligly adapted, shows that under the Orthogonal Vector Hypothesis, no strongly subquadratic algorithm exists to compute a 1.001 approximation of the Fréchet distance between curves of complexity $n$ and $m$. 
Buchin, Ophelders, and Speckmann~\cite{buchin2019seth} rule out a $(3-\eps)$-approximation in time  $O((n)^{2-\delta})$ between two $n$-vertex curves in $\bR^1$ and Blank~\cite{blank2026fr} rules out a $(1+\sqrt{2}-\eps)$-approximation of the Fréchet distance between curves of complexity $n$ and $m$ in $\bR^2$ under the $\ell_2$-norm in time $O((nm)^{1-\delta})$. 

\vspace{-0.2cm}
\subparagraph{Restricted (realistic) curve classes.}
To obtain strongly subquadratic-time algorithms, one therefore has to do one of two things:
either, approximate the Fréchet distance by a considerable constant~\cite{blank2026fr, ChengHuangZhang2025,  LiuWang2026, VanDerHorstOphelders2026}, 
or, restrict the input.
A popular way to do so is to impose so-called \emph{realism assumptions} on the input---i.e., to exclude geometric behaviour considered unlikely to occur in practice and the four classical examples are: $\kappa$-straight and $\kappa$-bounded curves~\cite{alt2004comparison}, $c$-packed curves~\cite{DriemelHW12}, and $\lambda$-low-density curves (Figure~\ref{fig:realistic}).

\noindent
A few things are worth noting about these realism assumptions:
\begin{itemize}
    \item \emph{Hierarchy:}
    Every $\kappa$-straight curve is $2\kappa$-packed, while both $c$-packed and $\kappa$-bounded curves are $\lambda$-low-density with $\lambda$ in $O(c)$ and $O(\kappa^d)$, respectively.

    \item \emph{Near-linear bounds:} When $c$ or $\kappa$ are constant, there exist near-linear-time algorithms to $(1+\eps)$-approximate the Fréchet distance for $c$-packed, $\kappa$-straight, and $\kappa$-bounded curves~\cite{DriemelHW12}, but no such algorithm is known for $\lambda$-low-density curves.
    Driemel, Har-Peled, and Wenk~\cite{DriemelHW12} gave a $(1+\eps)$-approximation of the Fréchet distance between $n$-vertex $d$-dimensional curves in time
$\tilde{O}(\eps^{-2}n^{2-2/d})$,  assuming $\lambda$-low density \emph{and} $\lambda \in O(1)$.
    \item \emph{Tight bounds:} For $c$-packed curves, a $(1+\eps)$-approximation can be computed in $O(\frac{c}{\sqrt{\eps}}n)$ time, and this bound is tight under the Orthogonal Vectors Hypothesis (OVH)~\cite{bringmann17cpacked}.
    \item \emph{The theoretical standard:} in the theory literature, $c$-packedness is by far the most commonly used realism assumption for Fréchet-related problems such as distance computation, data structures, map matching, and clustering~\cite{DriemelHW12,bringmann17cpacked,bruning2022faster,conradi20231+,gudmundsson2023Cluster,gudmundsson:mapmatching,Hoog2024DataRealistic,Hoog2025Clustering,gudmundsson2024OneCurve,Conradi2026OneCurve}.
\end{itemize}\vspace{-0.5cm}

\begin{figure}[H]
    \centering
    \includegraphics[]{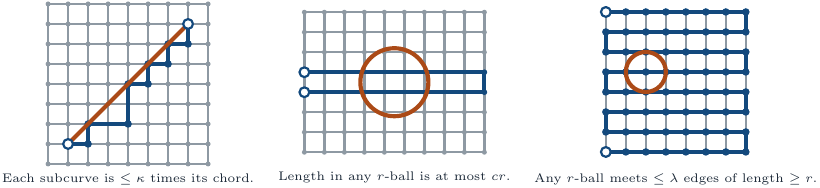}
    \caption{A curve that is either $\kappa$-straight, $c$-packed, or $\lambda$-low density for constant $\kappa$, $c$ and $\lambda$.}
    \label{fig:realistic}
\end{figure}

\newpage

\subparagraph{The case for $\lambda$-low density.}
The reason theory literature favours $c$-packedness over other realism assumptions, is likely the following strong and convenient theoretical property: 
\vspace{-0.1cm}
\begin{equation}
\tag{*}\label{eq:packed-simplification}
\text{
\hspace{-0.5cm}
If a curve } P \text{  is } c\text{-packed, then its } \mu\text{-simplification (Def.~\ref{def:simplification})
is } O(c)\text{-packed.~\cite[Lem~4.3]{DriemelHW12}}
\vspace{-0.2cm}
\end{equation}
Thus, one can simplify a $c$-packed curve $P$ to get a lower-complexity curve $P'$ where all edges of $P'$ have decent length, and use the $O(c)$-packed $P'$ in the remainder of the analysis. 
This property is convenient for theoretical algorithmic analysis, but there is a strong case for assuming $\lambda$-low density in practice. 
$\lambda$-low density considerably more general and practical: it restricts only the number of long edges intersecting a ball,
rather than the total length.
This distinction matters for road networks, as grid graphs (common in city architecture) have constant low density, but are $c$-packed for $c \in \Theta(\sqrt{n})$.

More practical and engineering oriented works, especially those involving real-world graphs, tend to favour assuming $\lambda$-low density~\cite{SCHWARZKOPF1996121,
DriemelHW12,Gudmundsson2024Lambda-Low-Density,
Buchin2024MapMatching,Gudmundsson2026Well-separated,
HarPeled2017ApproximationDensity,deBerg2012LowDensity}.
The authors from ~\cite{Buchin2024MapMatching} call $\lambda$-low density ``significantly more representative of real-world
networks'' and their claims are supported by experimental studies: 
road networks from four cities report small
low-density values~\cite{Chen2011ApproximateMapMatching},
while Gudmundsson, Huang, and Wong~
\cite{Gudmundsson2024Lambda-Low-Density}
find small low-density values across a range of real-world
trajectory data sets.
More recently, Diatzko, Lasseux, and Storandt~
\cite{Diatzko2026PackednessDensity}
experimentally compare the $\lambda$-low density versus $c$-packedness of inputs, and $\lambda$ is consistently much
smaller than $c$ on large road networks.

\subparagraph{Problem statement:}
This paper asks:
\emph{Can we, for $\lambda$-low-density curves, obtain the same
structural properties, efficient algorithms, and tight bounds
that we have for $c$-packed curves?}
This main question is split across four sub-questions, each with their own domain:

(1)  Curve simplification.
\emph{Can we establish an analogue of Property~\eqref{eq:packed-simplification} for $\lambda$-low-density?}

(2) Fréchet distance approximation. 
Thusfar, the analysis from~\cite{DriemelHW12} only studies $\lambda$-low density curves for when $\lambda \in O(1)$.
\emph{
Can we make the dependence on $\lambda$ explicit?}

(3) Lower bounds. Of all realistic input assumptions in our hierarchy, $\lambda$-low density is the only one for which no near-linear-time algorithm is known in dimensions $d>2$.
\emph{Can we establish conditional lower bounds showing that a superlinear dependence on $n$ is necessary?}

(4) Map matching.
Finally, can we extend our techniques beyond Fréchet distance computation?
Consider the map-matching problem (Figure~\ref{fig:mapmatching}):
given a trajectory $P$ and a geometric graph $G$, find a path in $G$ minimising its Fréchet distance to $P$.
A near-linear-time $(1+\eps)$-approximation is known when $G$ is $\lambda$-low-density and $P$ is $c$-packed~\cite{Chen2011ApproximateMapMatching}.
However, why should a recorded GPS trajectory be $c$-packed?
Indeed, assuming that only one curve is $c$-packed is a very strong algorithmic assumption as it already suffices for near-linear-time Fréchet distance approximation~\cite{gudmundsson2024OneCurve,Conradi2026OneCurve}.
Very recently, an ArXiV preprint shows that $c$-packedness of $P$ alone suffices~\cite{DeryckereGRW26}. We discuss this concurrent work in Appendix~\ref{app:concurrent}. Here, we consider the following open problem for the map matching problem: 
\emph{Can we obtain map-matching algorithms with tight bounds when both $G$ and $P$ are only $\lambda$-low-density?}
\begin{figure}[H]
    \centering
    \includegraphics[width=\linewidth]{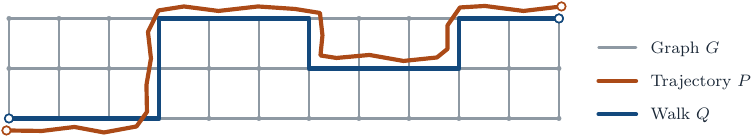}
    \caption{A trajectory $P$, a geometric graph $G$, and a walk $Q$ in $G$.}
    \label{fig:mapmatching}
\end{figure}

\newpage

\subparagraph{Contribution.}
We give an extensive overview of the properties of $\lambda$-low density curves in the context of Fréchet distance computation.
Table~\ref{tab:results} summarises all our results.

(1) Curve simplification.
Our key technical contribution is a tight simplification theorem.
We prove that the $\mu$-simplification (Definition~\ref{def:simplification}) of any $n$-vertex $\lambda$-low-density curve, for any $\mu > 0$, is
$O(\lambda^{1/d}n^{1-1/d})$-low-density.
This bound is tight: for any $(\lambda,n)$, we construct curves whose simplifications have low density
$\Omega(\lambda^{1/d}n^{1-1/d})$. 
See Appendix~\ref{app:simple_lambdaPrime}.

(2) Fréchet distance approximation.
In Appendix~\ref{app:driemel_details}, we revisit the analysis of $\lambda$-low-density curves from~\cite{DriemelHW12}.
We show how to extend their analysis to non-constant $\lambda$, yielding a $(1+\eps)$-approximation of the Fréchet distance between two $n$-vertex $\lambda$-low-density curves in $\tilde{O}(\lambda^2 n^{2-2/d}/\eps^2)$ time.
The main contribution of our paper are better upper bounds. We can compute a $(1+\eps)$-approximation for the Fréchet distance under the $\ell_1$ or $\ell_\infty$ metric in time $\widetilde O\left(
 \frac{\lambda^{2/d}n^{2-2/d}}{\varepsilon^{2-4/d}}
 \right)$.
 Under the $\ell_2$ metric, our algorithm takes $\tilde{O}\left(\frac{\lambda^{2/d}n^{2-2/d}}{\eps^{3/2-1/d}}\right)$ time.

(3) Lower bounds.
We prove that the running times from (2) are tight up to polylogarithmic factors under the OVH, except for the dependence on $\eps$ under the $\ell_2$ norm. 
In particular, a superlinear dependence on $n$ is unavoidable for $d>2$, even when $\lambda$ is constant.

(4) Map matching.
Finally,  we show that our techniques extend beyond Fréchet distance computation in Appendix~\ref{app:mapmatching}. 
We obtain tight bounds in terms of $n$ and $\lambda$ for map matching when both the input curve and the graph are $\lambda$-low-density, removing the $c$-packedness assumption on the input curve from previous work.

\begin{table}[h]
\centering
\small
\setlength{\tabcolsep}{3pt}
\begin{tabular}{@{}lllll@{}}
\toprule
\textbf{First input} & \textbf{Second input} & \textbf{$l_{\textbf{?}}$} & \textbf{Running time} & \textbf{Ref.} \\
\midrule

\multicolumn{5}{@{}l}{\textbf{Continuous and discrete Fréchet distance}} \\
\midrule

\rowcolor{black!7}
$P$: $O(1)$-low-dense
& $Q$: $O(1)$-low-dense
& Any
& $\tilde{O}\left(\frac{n^{2-2/d}}{\eps^2}\right)$
& \cite{DriemelHW12} \\

\rowcolor{black!7}
$P$: $\lambda$-low-dense
& $Q$: $\lambda$-low-dense
& Any
& $\tilde{O}\left(\frac{\lambda^2n^{2-2/d}}{\eps^2}\right)$
& App.~\ref{app:driemel_details} \\

$P$: $\lambda$-low-dense
& $Q$: $\lambda$-low-dense
& Any
& $\tilde{O}\left(\frac{\lambda^{2/d}n^{2-2/d}}{\eps^2}\right)$
& Thm. \ref{theorem:TwosidedLambdaLowDense} \\

$P$: $\lambda$-low-dense
& $Q$: $\lambda$-low-dense
& $\ell_1/\ell_\infty$
& $\nexists\,O\left(
    \left(
    \frac{\lambda^{2/d}n^{2-2/d}}{\eps^{2-4/d}}
    \right)^{1-\delta}
  \right)$
& Thm. \ref{thm:3up-approximation-hardness} \\

$P$: $\lambda$-low-dense
& $Q$: $\lambda$-low-dense
& $\ell_2$
& $\nexists\,O\left(
    \left(
    {\lambda^{2/d}n^{2-2/d}}
    \right)^{1-\delta}
  \right)$
& Thm. \ref{thm:euclidean-hardness} \\

$P$: $\lambda$-low-dense
& $Q$: $\lambda$-low-dense
& $\ell_1/\ell_\infty$
& $\tilde{O}\left(\frac{\lambda^{2/d}n^{2-2/d}}{\eps^{2-4/d}}\right)$
& Thm. \ref{thm:app-l1-improved} \\

$P$: $\lambda$-low-dense
& $Q$: $\lambda$-low-dense
& $\ell_2$
& $\tilde{O}\left(\frac{\lambda^{2/d}n^{2-2/d}}{\eps^{3/2-1/d}}\right)$
& Thm. \ref{thm:euclidean-hardness} \\

\midrule
\multicolumn{5}{@{}l}{\textbf{Map matching}} \\
\midrule

\rowcolor{black!7}
$G$: $\phi$-low-dense 
& $P$: $c$-packed
& $\ell_2$
& $\tilde{O}\left(
    \frac{\phi n}{\eps^d}+\frac{c|E|}{\eps}
  \right)$
& \cite{Chen2011ApproximateMapMatching} \\

\rowcolor{yellow!5}
$G$: arbitrary & $P$ : $c$-packed & $\ell_2$ & $\widetilde{O}\!\left(
    cn + \frac{n}{\eps^{2d}} + \frac{c|V|}{\eps^d}
    + \frac{c|E|}{\eps^3} + \frac{c^2|E|}{\eps^2}
\right)$
& \cite{DeryckereGRW26} \\

\rowcolor{yellow!5}
$G$: $\lambda$-low-dense
& $P$: $\lambda$-low-dense
& $\ell_2$
& $\tilde{O}\left(
        \frac{\lambda^2(n|V|)^{1-1/d}}{\eps^2}
        +\frac{\lambda^2(n+|V|)}{\eps^d}
  \right)$ expected
& \cite{DeryckereGRW26} \\

$G$: $\lambda$-low-dense
& $P$: $\lambda$-low-dense
& Any
& $\tilde{O}\left(
    \frac{\lambda^{2/d}(n+|E|)^{2-2/d}}{\eps^2}
  \right)\subsetneq \tilde{O}\left(
    \frac{\lambda^{2}(n+|V|)^{2-2/d}}{\eps^2}
  \right)$
& Thm.~\ref{thm:MapmatchingUpperbound} \\

$G$: $\lambda$-low-dense
& $P$: $\lambda$-low-dense
& $\ell_1/\ell_\infty$
& $\nexists\,O\left(
    \left(
    \frac{\lambda^{2/d}(n+|E|)^{2-2/d}}{\eps^{2-4/d}}
    \right)^{1-\delta}
  \right)$
& Thm.~\ref{thm:MapMatching3up-approximation-hardness} \\

\bottomrule
\end{tabular}
\caption{Our bounds for $(1+\eps)$-approximation.
Unless stated otherwise, both curves have $n$ vertices.
The lower bounds hold for every $\delta>0$, assuming the Orthogonal vector Hypothesis.
The orange rows are from a recent preprint.  }
\label{tab:results}
\end{table}

\newpage
\section{Preliminaries}

A collection of edges in $\mathbb{R}^d$ is $\lambda$-low-density if, for every $r>0$, every ball of radius $r$ intersects at most $\lambda$ edges of length at least $r$.
A $d$-dimensional polygonal curve $P$ is defined by a point sequence $(p_1, \ldots, p_n) \subset \mathbb{R}^d$ and forms a piecewise-linear function $P : [1,n] \to \mathbb{R}^d$ where for $t \in [i, i+1]$ it is defined as $P(t) = p_i + (t - i)(p_{i+1} - p_i)$. Let $P:[1,n]\to \mathbb{R}^d$ and $Q:[1,m]\to\mathbb{R}^d$ be two $d$-dimensional polygonal curves, where we assume $m \leq n$.
A pair of functions $(\alpha,\beta)$ is called a \textit{monotone traversal} if
$\alpha:[0,1]\to[1,n]$ and $\beta: [0,1] \to[1,m]$ are surjective and non-decreasing.
If $\mathcal{R}$ is the set of all monotone traversals then the \textit{continuous} Fréchet distance $\fd(P,Q)$ is defined as
$
\fd(P,Q)
:= \min_{(\alpha,\beta)\in\mathcal{R}}
\ \max_{t\in[0,1]}
\ \|P(\alpha(t))-Q(\beta(t))\| .
$

\subparagraph{$\Delta$-Free Space Diagram.}
For curves $P$ and $Q$ of complexity $n$ and $m$, we consider the \emph{parameter space} $[1, n] \times [1, m]$, For any $\Delta>0$, the $\Delta$-\emph{Free Space Diagram} ($\FSD$) is the subspace of the parameter space where for all $(s, t)$, the points $P(s)$ and $Q(t)$ are within distance $\Delta$ 
(formally, $\FSD(P, Q) := \{(s,t)\in [1, n] \times [1, m] \mid \|P(s)-Q(t)\| \le \Delta \}$.
We call the points in $\FSD(P, Q)$ the \emph{free space} and partition $\FSD(P, Q)$ by a unit grid of \emph{cells} where the cell $[i, i+1] \times [j, j+1]$ correspond to the $i$'th and $j$'th edge of $P$ and $Q$. 
The free space inside a cell is the intersection of the cell with an affine image of a disk and is thus convex and of constant complexity.
Alt and Godau~\cite{alt1995computing} showed that $\fd(P, Q) \leq \Delta$ if and only if there exists an $xy$-monotone path in $\FSD(P, Q)$ from $(1,1)$ to $(n,m)$.

\begin{figure}[h!]
    \centering
    \includegraphics[]{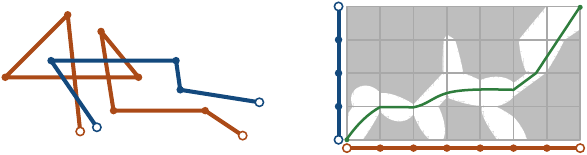}
    \caption{Curves P and Q and $\FSD(P, Q)$ with a $xy$-monotone path from $(1,1)$ to $(n,m)$. }
    \label{fig:FreespaceDiagram}
\end{figure}

\begin{definition}[Curve simplification]
\label{def:simplification}
    For any curve $P = (p_1, \ldots, p_n)$ and value $\mu > 0$, the \emph{(greedy) $\mu$-simplification}~\cite{DriemelHW12} is the maximal subsequence $S = (s_1,\ldots, s_m)$ of $P$ s.t.  $s_1=p_1$, $s_m = p_n$, and each~$s_j$ is the first vertex of $P$ after $s_{j-1}$ with distance $\geq \mu$ from $s_{j-1}$.
\end{definition}

\begin{observation}
\label{obs:simplified_distance}
For any polygonal curve $P$ and any $\mu \ge 0$,
$\fd(P, \Psimp) \le \mu.$
Moreover, if a ball $B$ of radius $\Delta$ contains an edge of length $\alpha$ of $\Psimp$ then the length of $P$ inside a  ball of radius $\Delta + \mu$ concentred at $B$ is at least $\alpha$ also. 
\end{observation}

\subparagraph{Approximating $\fd(P, Q)$.}
Our upper bounds use the framework for obtaining fast $(1+\eps)$-approximation from~\cite{DriemelHW12}. Given $\Delta \geq 0$ and $\varepsilon > 0$, an \ApproxDecider correctly outputs either $\fd(P,Q) \le \Delta(1+\varepsilon)$ or $\fd(P,Q) > \Delta$.  Note that the algorithm may return either outcome if $\fd(P,Q) \in (\Delta, (1+\varepsilon)\Delta]$.  
Their algorithm for input $(P, Q)$ and $\eps > 0$ computes a $(1+\eps)$-approximation in a time that depends on the  number of cells in $\FSD(\Psimp, \Qsimp)$ that are \emph{reachable}, where $\mu = \eps \Delta$ and $\Delta \geq 0$ is the choice of distance parameter that maximizes this complexity. 
\cite{DriemelHW12} then bounds the worst-case number of these reachable cells. Formally:

\begin{definition}
For two curves $P$ and $Q$ and $\Delta \geq 0$, their \emph{reachable free space}, is the set of
points in $\FSD(P, Q)$ where there exists an $xy$-monotone path from $(1, 1)$ to \emph{any} point in the cell that is contained in $\FSD(P, Q)$.\end{definition}

\begin{figure}
    \centering
    \includegraphics{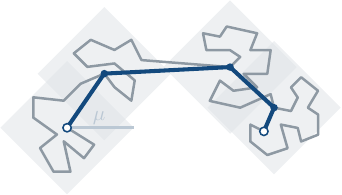}
    \caption{A curve $P$ (grey) and the curve simplification $\Psimp$ (blue) under the $\ell_1 $ norm.}
    \label{fig:Simplification}
\end{figure}

\subparagraph{Worst-case relative free-space complexity.}
Our key combinatorial object is the worst-case relative free-space complexity, which measures the maximum number of reachable cells in the free-space diagram after simplifying both curves (see Figure~\ref{fig:worst_case_complexity}).

\begin{definition}
    For two curves $P$ and $Q$ and $\Delta \geq 0$, we denote the complexity of the reachable free space, denoted by $N_{\Delta}(P, Q)$, as the total
    number of grid cells which have non-empty intersection with the reachable free space of $P$ and $Q$.     
\end{definition}

\begin{definition}[inspired by~\cite{DriemelHW12}]
\label{def:worstcase_complexity}
      For any pair of curves $(P, Q)$ and any $\eps > 0$, we define the worst-case relative free-space complexity as  
      $
      \N(\varepsilon, P, Q) := \max\limits_{\Delta \geq 0} \,  N_{\Delta}(P^{(\eps \Delta)}, Q^{(\eps \Delta)} )$.
\end{definition}

\begin{theorem}[$\Leftarrow$ Theorem 3.16 in \cite{DriemelHW12}] \label{theorem:ApproxAlgDriemel}
Given two polygonal curves $P$ and $Q$ with a total of $n$ vertices and a parameter $0<\varepsilon<1$, a $(1+\varepsilon)$-approximation of the continuous Fréchet distance between $P$ and $Q$ can be computed in time
$O\left(\N(\varepsilon,P,Q)+\N(1,P,Q)\log n\right)$.
\end{theorem}

\subparagraph{Lower bounds and OVH.}
There exist a variety of Fréchet distance lower bounds~\cite{bringmann2016approximability, bringmann17cpacked, DBLP:conf/focs/Bringmann14, vanderhoog_et_al:LIPIcs.ESA.2025.24, driemel_et_al:LIPIcs.SoCG.2022.36, buchin2019seth, blank2026fr}, which all can be conditioned on the \emph{Orthogonal Vector Hypothesis}.
Given two sets of vectors $U,W \subseteq \{0,1\}^b$ of $N$ and $M$ vectors, the Orthogonal Vector problem asks whether there exists vectors $u \in U$ and $w\in W$ such that $u$ is orthogonal to $w$, i.e., $u\cdot w=0$. If there exists a pair of orthogonal vectors, we will call it a YES instance. 
\begin{hypothesis}[Orthogonal Vector Hypothesis (OVH) \cite{DBLP:journals/tcs/Williams05}]
For all $\delta > 0$, there exist constants $\omega$ and $1 > \gamma > 0$ such that the Orthogonal Vector problem for vectors of dimension $b = \omega \log N$ and $M = N^\gamma$ cannot be solved in $O\big((MN)^{1-\delta}\big)$ time.
\end{hypothesis}

\noindent
Note that recently, the Exponential Time Hypothesis (ETH), and thus the Strong Expoenntial Time hypothesis (SETH), has been broken~\cite{openai2026eth}. However, OVH is incomparable to ETH and even though OVH implies SETH, the breaking of SETH does not break OVH.

\begin{figure}[H]
    \centering
    \includegraphics{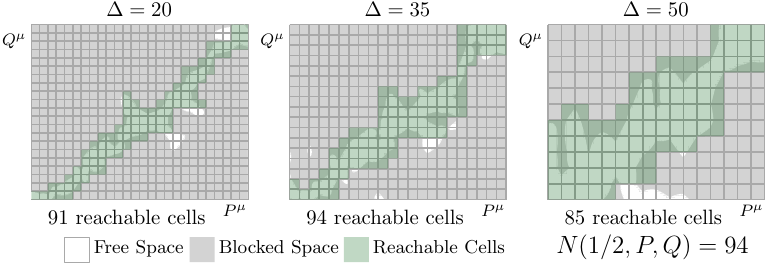}
    \caption{For $\eps=1/2$, three thresholds give different simplifications and different numbers of reachable cells. The maximum over all thresholds is $94$.}
    \label{fig:worst_case_complexity}
\end{figure}

\newpage
\section{\texorpdfstring{Lower bound for $\lambda$-low dense curves}{Lower bound for lambda-low dense curves}}\label{section:LowerboundForLowDense}

Let $P$ and $Q$ be two $d$-dimensional $\lambda$-low dense curves of complexity $n$ and $m$. 
We show lower bounds, conditioned on OVH, for computing their continuous or discrete Fréchet distance. 
A lower bound construction is \emph{balanced} if it assumes that $m = n$, and \emph{imbalanced} otherwise. 
We show lower bounds by creating what we call $\lambda$-paths:

\begin{definition}
    A $d$-dimensional curve $P$ is a $\lambda$-path if the curve is a walk on the $d$-dimensional unit grid graph where each edge is traversed at most $\lambda$ times. 
\end{definition}

\begin{observation}
    For any $\lambda$-path $P = (p_1, \ldots, p_n)$, we can subdivide all its edges into 3 edges of equal length to obtain an $O(\lambda)$-low dense curve of complexity $3n$.
\end{observation}

\noindent
 Thus, any lower bound for computing the Fréchet distance between a pair of $d$-dimensional $\lambda$-paths yields the same asymptotic lower bound for $\lambda$-low dense curves in $\bR^d$.
 To obtain a lower bound for $\lambda$-paths, we fix throughout the section one canonical $d$-dimensional unit grid graph. 
 We then  make use of one of two existing lower bounds:
\begin{itemize}
    \item A balanced lower bound for the continuous Fréchet distance in $\bR^1$, which states that a $(3-\eps)$-approximation for complexity $n$ curves in $O(n^{2-\delta})$ time \cite{buchin2019seth} refutes OVH.
    \item An imbalanced lower bound for the continuous Fréchet distance in $\bR^2$, which states that a $(3-\varepsilon)$-approximation algorithm for complexity $n$ and $m$ curves under the $\ell_1$-norm with running time $O((nm)^{1-\delta})$ \cite{blank2026fr} refutes OVH.\footnote{\cite{blank2026fr} rules out $(3-\varepsilon)$-approximations only for the $\ell_1$ and $\ell_\infty$ norm in $\bR^2$.}
\end{itemize}
These constructions live in $[-10,10]\subset\bR^1$ and $[-5,5]\times[-5,5]\subset\bR^2$, and decide whether the Fréchet distance is $1$ or $3$. 
Since the constructed edges have constant length, this yields a discrete Fréchet distance lower bound also. 
Throughout this section, we will assume the continuous Fréchet distance but similarly imply  discrete Fréchet distance lower bounds also.

\subparagraph{Applying these constructions.}
Prior lower bounds are not immediately applicable to our setting because the curves used in these reductions overlap arbitrarily and are thus not $\lambda$-low density. 
We resolve this by scaling up these constructions by some factor $C$.
If $C$ is sufficiently large then we can snap the original construction to the $1$-dimensional (or, $2$-dimensional) grid graph without creating a noticeable error in the Fréchet distance:

\begin{lemma}[1D hardness]\label{lem:1Dhardness}
   For all $C\geq 10$ and $N \ge  0$, there exists no algorithm that, given two 1D \emph{curves}, each consisting of $N$ $1$-\emph{paths} in $[-10C, 10C]$, can distinguish whether their discrete (or continuous) Fréchet distance is less than $1.1C$ or at least $2.9C$ in time $O(N^{2-\delta})$.
\end{lemma}

\begin{proof}
See also Figure~\ref{fig:1Dhardness}.
   We consider the balanced lower bound from \cite{buchin2019seth} and scale the construction by a factor $C$.
   This yields a problem instance in $[-10C, 10C]$ where the continuous Fréchet distance is at most $C$ in the YES instance and at least $3C$ in the NO instance.
    Next, we replace each edge of the curves by a path in the 1-dimensional grid graph. 
    This introduces an additive error of at most $1$ when considering the discrete Fréchet distance. By choosing $C \ge 10$ we bound this additive error by $1 \le 0.1C$. This yields two curves where their discrete (or continuous) Fréchet distance is at most $C + 1 \le C + 0.1C = 1.1C$ in the YES instance and at least $3C-1 \ge 3C-0.1C = 2.9C$ otherwise.
\end{proof}

\begin{figure}[H]
    \centering
    \includegraphics{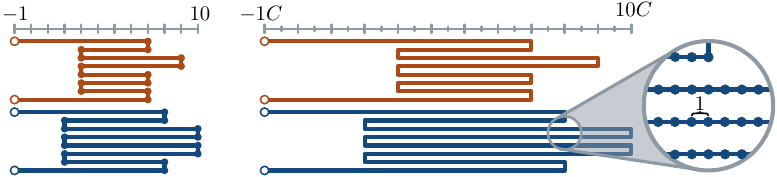}
    \caption{The left figure shows a subpart of two 1D curves constructed in \cite{buchin2019seth} (drawn in 2D). The right figure shows the same subpaths scaled by a factor $C$, where each edge is replaced by a $1$-dimensional path in the grid graph.}
    \label{fig:1Dhardness}
\end{figure}

\begin{lemma}[2D hardness]\label{lem:2Dhardness}
   For all $C\geq 20$, there exists no algorithm that, given  2D \emph{curves} consisting of $N$ and $M$ $1$-\emph{paths}, each of length $O(C)$,  in $[-5C, 5C] \times[-5C, 5C]$ respectively, can distinguish whether their discrete (or continuous) Fréchet distance is less than $1.1C$ or at least $2.9C$ in time $O((NM)^{1-\delta})$.
\end{lemma}
\begin{proof}
    We use the same scaling strategy as in Lemma \ref{lem:1Dhardness}, but invoke \cite{blank2026fr}, and note that discretising a constant-length edge adds a distortion of at most $2$ under any $\ell_p$ norm. Each discretised edge results in a path of length $O(C)$ and so the Lemma follows.
\end{proof}

In the above two lemmas, the resulting curves still have arbitrary overlap.
We obtain bounds for $\lambda$-paths in two steps:
first, we construct two $\lambda$-paths $P_*$ and $Q_*$ (which we refer to as \emph{origin paths}) in $\mathbb{R}^{d-1}$ (or $R^{d-2}$-dimensional for when we want an imbalanced lower bound) of complexity $N$ and $M$.
These have the property that any pair of points in $P_*$ and $Q_*$ are (roughly) equidistant.
We use the remaining $1$ (respectively $2$) dimensions to embed the 1D or 2D hardness construction. 


\subsection{Constructing the origin paths}

\begin{definition}
    For the $d$-dimensional grid graph we define the $r$-diagonals (\Cref{fig:D32}):
    \begin{align*}
        D_r^d&:=\{(x_1,\ldots,x_d)\in V \mid \text{$\forall i: x_i\geq 0$ and $\sum_{i=1}^dx_i=r$}\}\text{, and} \\
        -D_r^d&:=\{(x_1,\ldots,x_d)\in V\mid \text{$\forall i: x_i\leq 0$ and $\sum_{i=1}^dx_i=-r$}\}.
    \end{align*}
\end{definition}

\begin{figure}[b]
    \centering
    \includegraphics[]{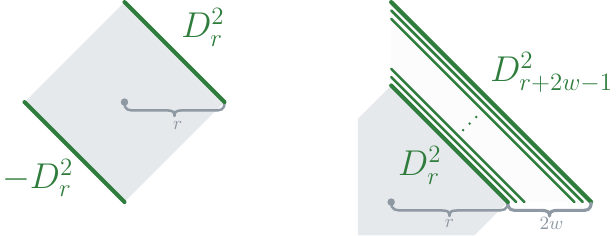}
    \caption{Left: Illustration of $D_r^2$ and $-D_r^2$. Right: Illustration of $\band^{2}(r,w)$.}
    \label{fig:D32}
\end{figure}
\noindent
These $r$-diagonals are defined such that for any pair $D_a^d$ and $- D_b^d$ the distance between any two points is uniquely determined by $a$ and $b$, and not by the specific points:

\begin{observation}\label{obs:surface-distance}
    For any $p\in D_a^d$ and any $q\in -D_b^d$ it holds that $d_G(p,q)=a+b$.
\end{observation}

\subparagraph{Embedding $\lambda$-paths on $r$-diagonals.}
For a fixed value $r$, the number of grid points on an $r$-diagonal $|D_r^d|$ is $\binom{r+d-1}{d-1}\geq\left(\frac{r}{d}\right)^{d-1}$.
In our lower bound, we wish to embed curves of complexity $N$ and $M$ as $\lambda$-paths on $r$-diagonals.
Suppose that we wish to do so with a single $r$-diagonal, then $r$ needs to be very large and as a consequence, $D_r^d$ and $-D_r^d$ will be very far apart and the Fréchet distance between our two embedded paths will be large.
To combat this, we use the fact that we have access to $d-1$ (or $d-2$) dimensions, and instead create a thick ``band'' of diagonals with $r \in [a, a+2w]$. More concretely:

\begin{definition}
    We define the \emph{band} $ \band^{d}(a,w):=\bigcup_{i=0}^{2w-1}D^d_{a+i}$ (see \Cref{fig:D32}, right).
    We equivalently define the band $-\band^d(a,w)$ via 
$-D^d_{a+i}$
\end{definition}

Pairs of points in $\band^d(a,v) \times -\band^d(b,w)$ are now ``almost equidistant'', that is, their distance is in $[a+b,a+b+2v+2w-2]$.
A complication of introducing bands, is that we need to be able to route a $\lambda$-path through the area $\band^{d}(a,w)$ which is achieved by the following Lemma:

\begin{lemma}\label{lem:origin-path}
    For all $a\geq 1$, $w\geq 1$ and $\lambda\geq 1$, there exist a $\lambda$-path with complexity $\Omega\left(\lambda w\left(\frac{a}{d}\right)^{d-1}\right)$ in the subgraph of the $d$-dimensional grid graph induced by $\band^d(a,w)$.
\end{lemma}

\begin{proof} 
     We observe that for a single $r$-diagonal,  $|D_r^d|=\binom{r+d-1}{d-1}\geq\left(\frac{r}{d}\right)^{d-1}$. 

    We prove the Lemma by induction over $d$, by proving the following statement:
     For every $d\geq 2$, every $a$, and every $s,t\in \{1,\ldots,d\}$, where $s\ne t$ the subgraph $G_a$ induced by $\band^d(a,1)=D^d_a\cup D^d_{a+1}$ contains a path from $(a+1) e_s$ to $a e_t$ that contains all vertices of $D_a^d$, where $e_i$ is the $i$th basis vector, that is $a e_1=(a,0,0,0,\ldots)\in G_a\subset G$. 

     We start with the base case, where we set $d=2$ and $s=1$ and $t=2$. Then $G_a$ contains the path $(a+1,0),(a,0),(a,1),(a-1,1),(a-1,2),...,(1,a),(0,a)$ ( \Cref{fig:LambdaPathOnDiagonals}, left).
     This is a simple path from $(a+1)e_1$ to $ae_2$ and contains all of $D_a^2$.

     Now consider some $d\geq 3$. By permuting the coordinates, it is enough to prove the claim for $s=1$ and $t=d$. For each $m=0,1,\ldots,a$ consider the slice $\Sigma_m=\{x\in G_a\mid x_d=m\}$. 
     Inside $\Sigma_m$, deleting the last coordinate identifies the slice with $D_{a-m}^{d-1}\cup D_{a-m+1}^{d-1}$ in $d-1$ dimensions. Now, let $i_0,i_1,\ldots i_{a}$ be the sequence that alternates between $1$ and $2$, starting at $1$.
     For each $m=0,1,\ldots,a$, we obtain via the induction hypothesis a path $P_m$ in $\Sigma_m$ from $((a+1)-m)e_{i_m}+me_d$ to $((a+1)-m-1)e_{i_{m+1}}+me_d$  containing every vertex of $D_{a}^d\cap\Sigma_m$. Further, the end of $P_m$ is adjacent to the start of $P_{m+1}$. So we may concatenate them to a simple path. Its first vertex is $(a+1)e_1$ and its last vertex is $ae_d$ (\Cref{fig:LambdaPathOnDiagonals}, right). As $\bigcup_m(\Sigma_m\cap D_a^d)=D_a^d$, the claim is also true for dimension $d$.

     \begin{figure}
         \centering
         \includegraphics{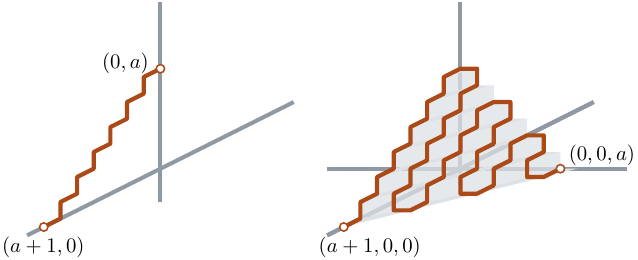}
         \caption{Left: a path from $(a+1)e_1$ to $ae_2$. Right: Path from $(a+1)e_1$ to $ae_d$ for $d=3$.}
         \label{fig:LambdaPathOnDiagonals}
     \end{figure}

     Overall, we obtain a path in $G_a$ of length $2|D_a^d|$. This path ends (choosing $s=1$ and $t=d$) in $ae_d$, and can be extended to end in $(a+1)e_d$ instead. Similarly, consider the path in $G_{a+2}$ with the same $s$ and $t$. We concatenate this path (traversed in reverse) to the first path. 
     This way, we obtain a path that is contained in $\band^d(a,2)$ of length at least $4|D_a^d|$. Repeating this process $w-2$ more times, we obtain a path in $\band^d(a,w)$ of length $w|D_a^d|=w\left(\frac{a}{d}\right)^{d-1}$. By walking back and forth $\lambda$ times we obtain the sought-after $\lambda$-path.
 \end{proof}

\subsection{Lower Bounds for exact algorithms}

We will first show, assuming the Orthogonal Vector Hypothesis, how we can exclude exact algorithms to compute the Fréchet distance between paths in $d$-dimensional grid graphs that run in strongly subquadratic time.
To this end, we apply either the known balanced construction for curves in $\bR^1$~\cite{buchin2019seth}, or the known imbalanced construction for curves in $\bR^2$~\cite{blank2026fr}.

\begin{theorem}\label{thm:exact-hardness}
    Unless the Orthogonal Vector Hypothesis fails, no algorithm can compute the continuous or discrete Fréchet distance in $\ell_1$ or $\ell_\infty$ between $d$-dimensional $1$-paths of complexity $n$ and $m$  in time
    $O(\min(n,m)^{2-\delta})$ if $d\geq 3$, or
        $O((nm)^{1-\delta})$ if $d\geq 4$, for $\delta > 0$.
\end{theorem}

\begin{proof}
    We note that we could prove the statement for $\lambda$-paths for any $\lambda \geq 1$ but since $\lambda$ does not show up in the running time (yet), the statement is strongest if $\lambda = 1$. 
        We first consider $d=3$.
        We apply Lemma~\ref{lem:1Dhardness}, picking $C = 10$, and a pair of one-dimensional curves $P$ and $Q$ of complexity $N$ such that there exists no algorithm to compute $\fd(P,Q)$ in time $O(N^{2-\delta})$ for $\delta > 0$ unless OVH fails.   
        These curves are restricted to the subgraph of the one-dimensional grid in $[-10C, 10C]$, and they are each the join of $N$ paths in $[-10C, 10C]$. 
        
     Next, we apply Lemma \ref{lem:origin-path} setting $a = N$, $w = 1$ and its $\lambda = 1$. If $N$ is sufficiently larger than the dimension $d = 3$, e.g. $N \geq 9$, then we obtain a  $2$-dimensional $1$-path $P_*$ in the band $\band^{d-1}(N, 1)$ of complexity at least $N$. 
     We similarly obtain a $2$-dimensional  $1$-path $Q_*$ in  $-\band^{d-1}(N, 1)$ of complexity at least $N$.

          \begin{figure}[!t]
    \centering
    \includegraphics{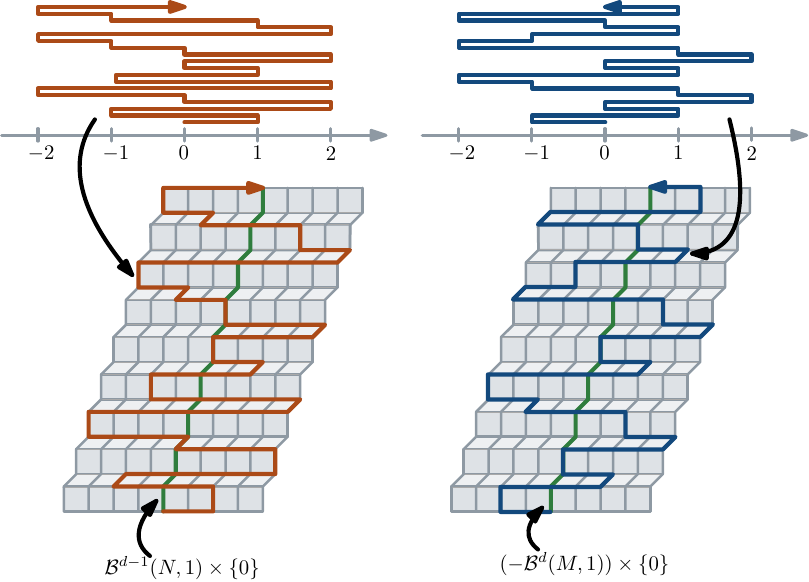}
    \caption{An illustration of the resulting paths $\hat{P}$ and $\hat{Q}$ in $\band^d(N,1)$ and $-\band^d(M,1)$.}
    \label{fig:lower-bound}
\end{figure}
     
     Our final step is to combine $P$ and $P_*$ to a $1$-path of complexity $\Theta(N)$ in $\mathbb{R}^3$. Recall that $P$ is the union of $N$ paths of complexity $O(C)$.  
      We start at the coordinate $(P_*[0],P[0])$. We then traverse the first path of $P$, while not changing the vertex of $P_*$ (see \Cref{fig:lower-bound}). Once we reach the end of the first path, we advance  to $P_*$'s next vertex, repeating this until we have traversed all paths in $P$. 
      The result is a path $\hat{P}$ in the $d$-dimensional grid graph of complexity $|P| \in O(N)$, as each path of $P$ has complexity at most $O(C)$ and $C$ is constant.
     Similarly, we obtain a path $\hat{Q}$ from $Q_*$ and $Q$ with complexity $O(N)$. 
    Any point $\hat{P}[i]$ in $\hat{P}$ has a corresponding point $P[i]$ in $P$ and any point $\hat{Q}[j]$ in $\hat{Q}$ has a corresponding $Q[j]$ in $Q$.

    By \Cref{obs:surface-distance} the distance between any point pair $(\hat{P}[i], \hat{Q}[j])$ of $\hat{P}$ and $\hat{Q}$, and thus also any pair that realises their  Fréchet distance, is in $[N+N+ \lVert P[i]-Q[j] \rVert,N+N+ 2 + \lVert P[i]-Q[j]\rVert]$. It follows that $\fd(\hat{P},\hat{Q})\in[2N+\fd(P,Q), \, 2N+\fd(P,Q)+2]$.
    
    Since $C \geq 10$, the Fréchet distance between $P$ and $Q$ is at most $1.1C$, then the Fréchet distance between $\hat{P}$ and $\hat{Q}$ is less than $N+N+2+1.1C\leq 2N+1.3C$.
    If instead the distance of $P$ and $Q$ is more than $2.9C$, then the distance between $\hat{P}$ and $\hat{Q}$ is more than $N+N+2.9C$. An exact algorithm can distinguish between these two. Thus, because the complexity $n$ of $\hat{P}$ and $\hat{Q}$ are $\Theta(N)$, \Cref{lem:1Dhardness} excludes  for any $\delta>0$ exact algorithms for computing the Fréchet distance between $\hat{P}$ and $\hat{Q}$ in time $O(n^{2-\delta})$.


     For $d \geq 4$, we instead 
     apply \Cref{lem:2Dhardness}, setting $C = 20$, and obtain two-dimensional curves $P$ and $Q$ of complexity $O(N)$ and $O(M)$. 
     We construct $P_*$ and $Q_*$ in the same manner, except that $Q_*$ lies in $- \band^{d-1}(M, 1)$.
     We combine $P$ and $P_*$ (respectively, $Q$ and $Q_*$) in the same way and obtain $4$-dimensional $1$-paths  $\hat{P}$ and $\hat{Q}$ where their Fréchet distance is less than $N+M+2+1.1C\leq N+M+1.3C$ if the  Fréchet distance between $P$ and $Q$ is at most $1.1C$.
     In contrast, if the  Fréchet distance between $P$ and $Q$ is at least $2.9C$ then the Fréchet distance between $\hat{P}$ and $\hat{Q}$ is at least $N+M+2.9C$. Thus, \Cref{lem:2Dhardness} excludes any exact algorithm for computing the Fréchet distance between $\hat{P}$ and $\hat{Q}$ in $O(nm^{1-\delta})$ time for $\delta>0$ as $n \in \Theta(N)$ and $m \in \Theta(M)$.
\end{proof}

\subsection{\texorpdfstring{Lower bounds for $(1+\eps)$-approximation algorithms}{Lower bounds for (1+epsilon)-approximation algorithms}}

Our lower bound in Theorem~\ref{thm:exact-hardness} excludes exact algorithms for all dimensions $d \geq 3$ (or $d \geq 4$ if we want an imbalanced bound). 
We strengthen this result by excluding $(1+\eps)$-approximations. 
We get conditional lower bounds that use the parameters $\lambda$, $\eps$, the dimension $d$, and the input sizes $n$ and $m$. 
The main body states our balanced lower bound:

\begin{theorem}\label{thm:3up-approximation-hardness}
    Unless the Orthogonal Vector Hypothesis fails, for any $d\geq 3$ and any $\eps\in \left(d^{-\frac{d-1}{d-2}} \, \left(\frac{\lambda}{n}\right)^{\frac{1}{d-2}},1\right]$, there is no algorithm to $(1+\eps)$-approximate the Fréchet distance between two size-$n$ $d$-dimensional $\lambda$-paths in time    $O\left(\left(\lambda^{2/d} \,  \frac{(n/d)^{2(1-1/d)}}{\eps^{2(1-2/d)}}\right)^{1-\delta}\right)$ with $\delta > 0$.
\end{theorem}
\begin{proof}
We start with a balanced instance of the Orthogonal Vector Problem.
Let $N$ be the number of $1$-paths in each curve of the construction
in Lemma~\ref{lem:1Dhardness}. Set
$
    a=\left\lceil d\left(\frac{N}{\lambda\eps}\right)^{1/(d-1)}\right\rceil$ and $    w=\lceil\eps a\rceil.$
We restrict $\eps$ to the range in the theorem because, for smaller
values of $\eps$, the expression inside the lower bound exceeds $n^2$.
As we show in Observation~\ref{obs:overkill}, computing a
$(1+\eps)$-approximation for these smaller values of $\eps$ still yields a lower bound.

    We then choose the value $C = 20 w$ and we apply Lemma \ref{lem:1Dhardness} and obtain two curves $P$ and $Q$, that each are the join of of $N$ paths in $[-10C, 10C]$. 
    Thus, these two curves have complexity $O(CN)$ and there exists no algorithm to distinguish between whether their Fréchet distance is less than $1.1C$ or at least $2.9C$ in time $O(N^{2-\delta})$ for all $\delta > 0$.

  Next, we apply the construction in \Cref{lem:origin-path} in
$d-1$ dimensions to obtain a $\lambda$-path in $\band^{d-1}(a,w)$
of length at least
$
    \lambda w\left(\frac{a}{d-1}\right)^{d-2}
    \geq \frac{\lambda\eps a^{d-1}}{(d-1)^{d-2}}
    \geq N$.    Denote this path by $P_*$. 
   We similarly obtain a $(d-1)$-dimensional  $\lambda$-path $Q_*$  of complexity at least $N$ in $- \band^{d-1}(a,w)$.

    From the pair $(P, P_*)$, we obtain a $d$-dimensional $\lambda$-path as follows: we start at $(P_*[0], P[0])$ and then traverse the first path of $P$, whilst not changing the vertex of $P_*$. Once we reach the end of the first path, we advance $P_*$ by one and repeat. 
    Applying this procedure to both $(P_*, P)$ and $(Q_*, Q)$ yields two $\lambda$-paths $\hat{P}$ and $\hat{Q}$ of complexity $O(C N)$ each. Note that each point in $\hat{P}[i]$ in $\hat{P}$ corresponds to the point $P[i]$ in $P$.  
    We then apply Observation~\ref{obs:surface-distance}.

    For any two points of $\hat P$ and $\hat Q$, the distance between
their points on $P_*$ and $Q_*$ lies between $2a$ and $2a+4w-2$.
It follows that if the Fréchet distance between $P$ and $Q$ is
at most $1.1C$, then the Fréchet distance between $\hat P$ and
$\hat Q$ is at most $2a+4w-2+1.1C\leq 2a+26w$, since $C=20w$.

If, otherwise, the Fréchet distance between $P$ and $Q$ is at
least $2.9C$, then the Fréchet distance between $\hat P$ and
$\hat Q$ is at least $2a+2.9C=2a+58w$. Since $\eps a\leq w$
and $\eps\leq1$, we have $(1+\eps)(2a+26w)\leq 2a+54w<2a+58w$.
Thus, a $(1+\eps)$-approximation distinguishes these two cases.
In particular, by Lemma~\ref{lem:1Dhardness}, there cannot exist
an $O(N^{2-\delta})$-time algorithm that computes such an
approximation for any $\delta>0$, unless OVH fails.
    
    Let us finally, for the sake of contradiction, suppose there exists a $(1+\eps)$-approximation algorithm with running time     $O\left(\left(\lambda^{2/d}\frac{(n/d)^{2(1-1/d)}}{\eps^{2(1-2/d)}}\right)^{1-\delta}\right)$. We now express this running time in terms of $N$. If $w\geq2$,
then $\eps a>1$, so $w=\lceil\eps a\rceil\leq 2\eps a$.
Thus, $C=20w=O\bigl(\eps(N/(\lambda\eps))^{1/(d-1)}\bigr)$.
Since $n=O(NC)$, substituting this bound on $C$ gives
$
    \lambda^{2/d}
    \frac{(n/d)^{2-2/d}}{\eps^{2-4/d}}=O(N^2)$.
If $w=1$, then $C=20$ and $n=O(N)$. In this case, the lower
bound on $\eps$ in the theorem gives
$
    \lambda^{2/d}
    \frac{(n/d)^{2-2/d}}{\eps^{2-4/d}}\leq n^2=O(N^2)$.
Thus, in both cases the assumed algorithm takes
$O(N^{2-2\delta})$ time.
Constructing $\hat P$ and $\hat Q$ takes $O(NC)$ time.
Since $d\geq3$, $\lambda\geq1$, and $\eps\leq1$, we have
$C=O(\sqrt N)$. Thus, we can construct the curves in
$O(N^{3/2})$ time and distinguish the two cases in
$O(N^{2-2\delta})$ time, contradicting OVH.
\end{proof}

\noindent
We briefly note that the lower bound on $\eps$ in Theorem~\ref{thm:3up-approximation-hardness} is tight:

\begin{observation}\label{obs:overkill}
Unless the Orthogonal Vector Hypothesis fails, 
 for any dimension $d\geq 3$ and any $\eps\in (0,1]$, there is no algorithm to $(1+\eps)$-approximate the Fréchet distance between two size-$n$ $d$-dimensional $\lambda$-paths  in time $O\left(\min\left(n^2,\lambda^{2/d} \,  \frac{(n/d)^{2(1-1/d)}}{\eps^{2(1-2/d)}}\right)^{1-\delta}\right)$ with $\delta > 0$.
\end{observation}
\begin{proof}
    Observe that for $\eps = d^{-\frac{d-1}{d-2}} \, (\frac{\lambda}{n})^{\frac{1}{d-2}}$, the lower bound from Theorem~\ref{thm:3up-approximation-hardness} excludes $O(n^{2-\delta})$-time algorithms.
    For any $\eps' < \eps$, a $(1+\eps')$-approximation that runs in $O(n^{2-\delta})$-time, is also a $(1+\eps)$-approximation that runs in $O(n^{2-\delta})$ time. 
\end{proof}

In Appendix~\ref{app:imbalanced}, we achieve an imbalanced bound by invoking the imbalanced versions of Lemma \ref{lem:2Dhardness} and \Cref{lem:origin-path} instead. In the plane, the map $(x,y)\mapsto(x-y,x+y)$ takes $\ell_1$
distances to $\ell_\infty$ distances. Thus, the hardness
construction also applies under the $\ell_\infty$ norm.

\begin{restatable}{theorem}{fourUpHardness}\label{thm:4up-approximation-hardness}
    Unless OVH fails, for any $d\geq 4$ and any $\eps\in \left(\frac{1}{d}^{\frac{d-2}{d-3}}\,\frac{\lambda}{n}^{\frac{1}{d-3}},1\right]$, there is no algorithm to $(1+\eps)$-approximate the Fréchet distance between two $\lambda$-paths of complexity $n$ and $m$, where $n \geq m$ in time 
    $O\left(\left(\lambda^{2/(d-1)} \, (m/d)  \, \frac{(n/d)^{1-2/(d-1)}}{\eps^{2(1-2/(d-1))}}\right)^{1-\delta}\right)$ with $\delta > 0$.
\end{restatable}


\newpage
\section{An improved upper bound}
\label{sec:improved_upper_bound}

Driemel, Har-peled and Wenk~\cite{DriemelHW12} showed an algorithm to $(1+\eps)$-approximate  the Fréchet distance between $\lambda$-low dense curves. 
In Appendix~\ref{app:driemel_details}, we remark two things about their upper bound.
Firstly, their paper assumes that $\lambda$ is constant, and a critical lemma of theirs does not survive the assumption that $\lambda$ is non-constant as written.
Secondly, one can adapt their proofs to work for non-constant $\lambda$ to obtain an $O\left(\frac{\lambda^2 n^{2(d-1)/d}}{\eps^2} \log n\right)$-time algorithm for when $P$ and $Q$ are $d$-dimensional, $\lambda$-low dense curves of complexity $n$ each. 

In the main body, we provide a tighter upper bound of $O\left(\frac{n^{2(d-1)/d}\lambda^{2/d}}{\varepsilon^2}\right)$ time, under \emph{any} $\ell_p$ metric. 
We achieve this by upper bounding worst-case relative free-space complexity $\N(\eps, P, Q)$ (see Definition~\ref{def:worstcase_complexity}) in what we argue is a simpler, and more geometric manner. 
Theorem~\ref{theorem:ApproxAlgDriemel} then yields an  $O\left(\frac{n^{2(d-1)/d}\lambda^{2/d}}{\varepsilon^2} \log n \right)$-time algorithm. 
Appendix~\ref{app:epsUnderl1norm} then gives an algorithm that also matches the lower bound in $\epsilon$ under the $\ell_1$ metric. 
Appendix~\ref{app:epsUnderl2norm} gives an improved (but not proven tight) dependency on $\eps$ under the $\ell_2$ metric.

\subparagraph{An important side remark.}
We briefly wish to note that our analysis can also be used to shown an useful geometric fact. 
In~\cite{DriemelHW12}, there are three types of ``realistic'' curves: $\kappa$-straight, $c$-packed, and $\lambda$-low dense.
Both $\kappa$-straight curves and $c$-packed curves $P$ have the following property: for any $\mu \geq 0$, the $\mu$-simplification $P^{\mu}$ is $O(\kappa)$-straight, or $O(c)$-packed. 
This intuitive property is used either implicitly or explicitly in~\cite{alt2004comparison, DriemelHW12, bringmann17cpacked, driemel_et_al:LIPIcs.SoCG.2022.36, vanderhoog_et_al:approximate, Conradi2026OneCurve, gudmundsson2024OneCurve, Conradi_2025_revisiting,Afshani2025Property, driemel2013jaywalking, Buchin2024MapMatching, Gudmundson2015Mapmatchingtrees, bruning2022FasterCovering, Conradi2024kshortcut, gudmundsson2023cpackedcluster, Gudmundsson2023mapmatchingrealistic, Conradi2024ANN, Chen2011approximate}. 
In particular, it allows for efficient upper bounds even when only one curve is $c$-packed~\cite{Conradi2026OneCurve, gudmundsson2024OneCurve}. 
We show in Appendix~\ref{app:simple_lambdaPrime} that our analysis  can also be used to make an analogous statement for $\lambda$-low dense curves:

\begin{restatable}{theorem}{simpleIsLowDense}
        \label{theorem:SimplIsLowDense}
    Let $P \subset \mathbb{R}^d$ be a $\lambda$-low dense curve of complexity $n$ and $\mu \geq 0$.     
    Then $\Psimp$ is $\lambda'$-low dense for $\lambda' \in  O(n^{1-1/d}\lambda^{1/d})$ and $\exists P\subset \mathbb{R}^d $ where $\Psimp$ is $\Omega(n^{1-1/d}\lambda^{1/d})$-low dense.
\end{restatable}

Theorem~\ref{theorem:SimplIsLowDense} gives a general structural property whose dependence on $n$ and $\lambda$ equals that of our lower bounds.
For computing the continuous Fréchet distance, however, we rather work with the following more technical statement.
It resembles to Lemma~4.9 of~\cite{DriemelHW12}, but it has an explicit and better dependency on $\lambda$, and a more geometric proof:


\begin{lemma}\label{lemma:VerticesIn3C}
Let $1 \le p \le \infty$, and let $P \subset \mathbb{R}^d$ be a $\lambda$-low-density set of edges under the $\ell_p$ norm.
Let $C$ be a cube of side length $r$. Let $3C$ denote the cube concentric with $C$ whose side length is $3r$. Let the total length of $P$ inside $C$ be $\alpha$, where $\alpha \ge c\lambda r$ for a sufficiently large constant $c$. If $d \in O(1)$ and $m$ is the number of vertices of $P$ contained in $3C$, 
then
$
m \in \Omega\left(
\frac{(\alpha/r)^{d/(d-1)}}{\lambda^{1/(d-1)}}
\right).$
\end{lemma}

\begin{proof}
We first conceptually construct an infinite quadtree on $C$. We construct levels $j \in \{0, 1, \ldots, \infty \}$ where at level $j$, the cube $C$ is partitioned into $2^{dj}$ cubes, each having side length
$s_j=\frac{r}{2^j}$.
Next, we partition the edges of $P$ based on their length. Let $E_{0}$ denote the set of edges whose length is at least $r$ that intersect $C$, and for each integer $j>0$, define
$ E_j= \left\{ e \text{ of }  P \mid \frac{r}{2^{j}} \le |e| < \frac{r}{2^{j-1}}  \text{ and } e \cap C \neq \emptyset \right\}. $
Thus, the set $E_j$ contains precisely those edges whose length is within a factor 2 of the side length of cubes at level $j$.

The edges in $E_{0}$ require separate treatment. There are only $O(2^d \lambda)$ such edges intersecting $C$, and their total length inside $C$ is at most $\sqrt{d}  2^d \lambda r \in O(\lambda r)$, since $d \in O(1)$. As we assume $\alpha\ge c\lambda r$ for some sufficiently large $c$, it follows that after removing these edges, $\alpha \geq c' \lambda r$ for some positive constant $c'$.
We do not wish to carry the variable $c'$ throughout the proof, and instead abuse notation by henceforth assuming that $\alpha \geq c \lambda r$ and that $E_0$ is empty.

We now bound the number of edges in each class $E_j$. Consider the cubes at level $j$. 
Every hypercube of side length $s_j$ of can be covered by $2^d$ disk of radius $s_j$. Every edge in $E_j$ has length at least $s_j$. 
Consequently, every level-$j$ cube can intersect only $O(2^d\lambda)$ edges of $E_j$, by the $\lambda$-low-density property. Since there are $2^{dj}$ cubes, we obtain $|E_j| \le 2^d\lambda2^{dj}$.
On the other hand, every edge in $E_j$ has both endpoints inside $3C$, so there can be at most $m$ such edges. Combining the two bounds yields $ |E_j| \le \min\{m,2^d\lambda2^{dj}\}$.

We choose an integer $J$ satisfying $m \in \Theta(2^d\lambda2^{Jd})$. The integer $J$ serves as a transition point: $|E_j| \in O(2^d\lambda2^{dj})$ for $j \leq J$ and $|E_j| \in O(m)$ otherwise. 
Let $L_j$ denote the total length contributed by edges in $E_j$ inside $3C$. Since every edge in $E_j$ has length at most $r/2^{(j-1)}$,
$ L_j \le \frac{r}{2^{j-1}}|E_j| \le \frac{r}{2^{j-1}} \min\{m,2^d\lambda2^{dj}\}.$ 
Now, split the total curve length into the contribution from coarse edges,
$ L_{\mathrm{coarse}} = \sum\limits_{j=1}^J L_j, $
and the contribution from fine edges,
 $ L_{\mathrm{fine}} = \sum\limits_{j = J}^\infty L_j. $
Since 
$ \alpha \leq L_{\mathrm{coarse}} + L_{\mathrm{fine}}, $
at least one of these quantities must be at least $\alpha/2$.

Suppose first that
$L_{\mathrm{coarse}}\ge\alpha/2$.
The total length $L_j$ contributed by edges in $E_j$ inside $3C$ is  upper bounded by 
$L_j \le \frac{r}{2^{j-1}} 2^d\lambda2^{dj} = 2^{(d+1)}r\lambda2^{(d-1)j}   $.
Summing over all levels $j \leq J$, 
$ L_{\mathrm{coarse}} \le \sum\limits_{j=1}^J 2^{(d+1)}r\lambda2^{j(d-1)}. $
Since this is a geometric series with increasing terms, it is dominated by its final term. Hence,
$ L_{\mathrm{coarse}} \in O( 2^{(d+1)}r\lambda2^{J(d-1)}) . $
Using 
$ 2^J = \Theta\!\left( (\frac{m}{2^d\lambda})^{1/d} \right), $
we obtain
$ L_{\mathrm{coarse}} \in O( 2^{(d+1)}r\lambda (\frac{m}{\lambda2^d})^{(d-1)/d}) \in O(r \lambda^{1/d} {m}^{(d-1)/d})$

Now suppose instead that
$L_{\mathrm{fine}}\ge\alpha/2$.
We can also upper bound $L_j$ by $ L_j \le \frac{rm}{2^{j-1}}. $
Summing over all fine scales gives
 $ L_{\mathrm{fine}} \le \sum\limits_{j=J}^\infty \frac{rm}{2^{j-1}}. $
This geometric series has decreasing terms and is therefore dominated by its first term. Consequently, $ L_{\mathrm{fine}}\in O(\frac{rm}{2^{J-1}}). $
Substituting 
$ 2^J \in \Theta\!\left( (\frac{m}{\lambda})^{1/d} \right), $
yields $ L_{\mathrm{fine}} \in O\!\left( r \lambda^{1/d} {m}^{(d-1)/d} \right). $

In both cases, 
$ \alpha \in O\!\left(r m^{(d-1)/d} \lambda^{1/d} \right). $
Rearranging gives
$ m^{(d-1)/d} = \Omega\!\left( \frac{\alpha} {r\lambda^{1/d}} \right), $
and raising both sides to the power $d/(d-1)$ yields
$ m = \Omega\!\left( \frac{(\alpha/r)^{d/(d-1)}} {\lambda^{1/(d-1)}} \right), $
as claimed.
\end{proof}

The free space complexity is a reflection of the number of vertex-edge pairs of $P$ and $Q$ within distance $\Delta$. If a vertex $p$ of $P$ is included in many vertex-edge pairs, there must be a substantial length of $Q$ inside a region of distance $O(\Delta)$ from $p$. We know from \Cref{lemma:VerticesIn3C} that a low-density set of edges cannot accumulate substantial length in a region without having many nearby vertices. We will use this inside to form a charging argument that distributes the charge of each vertex-edge pair among the vertices in that region.
The charging argument is independent of the choice of $\ell_p$ norm, as it only relies on the low-density property and the fact that constant-size expansions of grid cells intersect $O(1)$ grid cells.
We note for the reader that, just as in~\cite{DriemelHW12}, we require that $|P| = n = m = |Q|$.

\begin{theorem}\label{theorem:TwosidedLambdaLowDense}
    For any $\ell_p$ norm, $1 \le p \le \infty$ and any two size-$n$ $\lambda$-low-density curves $P$ and $Q$ in $\mathbb{R}^d$, and $0 < \varepsilon < 1$, we have $\N(\varepsilon, P, Q) \in O\left(\frac{\lambda^{2/d} n^{2(d-1) / d}}{\varepsilon^2}\right)$. Thus, there exist an algorithm which $(1+\eps)$-approximate $\fd(P,Q)$ in time $O\left(\frac{\lambda^{2/d} n^{2(d-1) / d}}{\varepsilon^2}+ \left(\lambda^{2/d} n^{2(d-1) / d}\right)\log n\right)$. 
\end{theorem} 

\begin{proof}
The proof strategy resembles that of Lemma 4.12 in \cite{DriemelHW12}. We fix an arbitrary $\ell_p$ norm, $1 \le p \le \infty$, throughout the proof. All distances and curve lengths are measured with respect to this norm.
Let $\Delta \ge 0$ be an arbitrary radius, and let $\Psimp$ and $\Qsimp$ be the $\mu$-simplified curves of $P$ and $Q$ for $\mu = \varepsilon\Delta$. Then, we need to prove that $\N_{\le\Delta}(\Psimp, \Qsimp) \in O\left(\frac{\lambda^{2/d} n^{2(d-1) / d}}{\varepsilon^2}\right)$.
    
To this end, it suffices to bound the number of vertex-edge pairs $(p,q)$, where $p$ is a vertex of $\Psimp$ and $q$ is an edge of $\Qsimp$ (and the other way around) where the distance between $p$ and $q$ is at most $\Delta$. Consider a  cell in $\FSD(\Psimp, \Qsimp)$ that is reachable from $(1, 1)$: the $xy$-monotone path must enter the cell via one of its facets and thus the
total number of such pairs $(p, q)$ bounds the total complexity of $\N_{\le \Delta}(\Psimp, \Qsimp)$.
We choose $M \in O\left(\frac{\lambda^{2/d} n^{(d-2) / d}}{\varepsilon^2}\right)$.

Consider a $d$-dimensional grid $\mathcal{G}$ of side length $\Delta$. 
For a grid cell $R$ where $R$ contains at least one vertex of $\Psimp$, we select the vertex $p_R \in R$ which is part of the largest number of edge pairs $(p_R, q)$ for $q$ an edge of $\Qsimp$ and we denote by $d_R$ the number of vertex-edge pairs that it is a part of.
The total number of vertex-edge pairs $(p, q)$ where $p$ is a vertex of $\Psimp$ inside $R$ is upper bounded by $U_R=|\Psimp\sqcap R| \cdot d_R$, where $|\Psimp\sqcap R|$ denotes the number of vertices of $\Psimp$ inside $R$. We make a case distinction, comparing $d_R$ to $M$:

\begin{figure}
    \centering
    \includegraphics[]{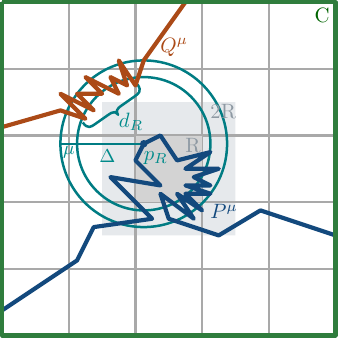}
    \caption{Illustration of the grid $\mathcal{G}$, the cell $R$, the vertex $p_R$, and the vertices of $\Qsimp$, that contributes to the count $d_R$. Showcases that $\| \Qsimp \sqcap B(p_R,\Delta+ \mu) \rvert \ge d_R\mu$ or $\| \Psimp \sqcap 2R \| \ge |\Psimp \sqcap R|\mu$.}
    \label{fig:Lemma412}
\end{figure}

If $d_R \le M$ then $U_r \le |\Psimp\sqcap R| \cdot M$, and we charge $M$ units to each vertex of $P$ inside $R$.

If $d_R > M$ we consider the two sub cases: $|\Psimp\sqcap R| \le d_R$ or $|\Psimp\sqcap R| > d_R$.

\subparagraph{Subcase 1: $|\Psimp\sqcap R| \le d_R$.}
Denote by $C = 5R$ the cube of side length $5\Delta$ centred at $R$.
Since all edges of $\Qsimp$ have length at least $\mu = \eps \Delta$, the length of $\Qsimp$ inside a ball of radius $\Delta + \mu$ around $p_R$ is at least $d_R \mu$. 
By Observation~\ref{obs:simplified_distance}, the length of $Q$ inside $C$ is at least $d_R \mu > M\mu = \frac{\lambda^{2/d} n^{1-2/d}}{\varepsilon^2} \Delta \varepsilon \ge \lambda \Delta \varepsilon^{-1}$.
For small enough $\eps$ (even when $\eps$ is constant), the length of $Q$ inside $C$ is at least $c \lambda \Delta$ for some sufficiently large constant $c$, which is the assumption needed to invoke \Cref{lemma:VerticesIn3C}.
Let $3C$ denote the cube co-centric with $C$ whose side length is $15 \Delta$.
By invoking \Cref{lemma:VerticesIn3C}, $Q$ must have at least $\Omega\left(\frac{(d_R\mu/\Delta)^{d/(d-1)}}{\lambda^{1/(d-1)}}\right) = \Omega\left(\left(d_R\frac{\varepsilon}{\lambda^{1/d}}\right)^{d/(d-1)}\right)$
vertices inside $3C$. There is some constant $\beta$ s.t.:

\begin{align*}
& \beta \left(\frac{\varepsilon}{\lambda^{\frac{1}{d}}} d_R\right)^{\frac{d} {(d-1)}}
\leq|Q \sqcap 3C| \Longrightarrow 
d_R \leq \frac{\lambda^{\frac{1}{d}}}{\varepsilon}\left(\frac{|Q \sqcap 3C|}{\beta}\right)^{\frac{(d-1)}{d}} \Longrightarrow d_R^2 \leq \frac{\lambda^{\frac{2}{d}}}{\varepsilon^2}\left(\frac{|Q \sqcap 3C|}{\beta}\right)^{2-\frac{2}{d}} \\
& \Longrightarrow d_R^2 \leq|Q \sqcap 3C| \frac{\lambda^{\frac{2}{d}}}{\beta \varepsilon^2}\left(\frac{|Q \sqcap 3C|}{\beta}\right)^{1-\frac{2}{d}} \leq \frac{\lambda^{\frac{2}{d}}}{\beta \varepsilon^2}\left(\frac{n}{\beta}\right)^{1-\frac{2}{d}}|Q \sqcap 3C| \leq M|Q \sqcap 3C| \tag{1}
\end{align*}
by picking $M$ to be sufficiently large. As we assumed $|\Psimp\sqcap R| \le d_R$, then $U_R=|\Psimp\sqcap R|d_R \le d_R^2 \le M |Q \sqcap 3C|$. Hence, we charge $M$ units to each vertex of $Q$ inside the cube $3C$.

\subparagraph{Subcase 2: $|\Psimp\sqcap R| > d_R$.}
Since all edges of $\Psimp$ have length at least $\mu = \Delta\eps$, the length of $\Psimp$ inside the cube of side length $2\Delta$ centred at $R$ is $|\Psimp\sqcap R|\mu$. By \Cref{obs:simplified_distance}, the length of $P$ inside $C$ is at least $|\Psimp\sqcap R|\mu > d_R \mu > M\mu = \frac{\lambda^{2/d} n^{1-2/d}}{\varepsilon^2} \Delta \varepsilon \ge \lambda \Delta \varepsilon^{-1}$.
For small enough $\eps$ (even when $\eps$ is constant), the length of $P$ inside $C$ is at least $c \lambda \Delta$ for some sufficiently large constant $c$, which is the assumption needed to invoke \Cref{lemma:VerticesIn3C}.

By \Cref{lemma:VerticesIn3C}, we have that $P$ must have at least $\Omega\left(\left(|\Psimp\sqcap R|\frac{\varepsilon}{\lambda^{1/d}}\right)^{d/(d-1)}\right)$ vertices inside $3C$. 
Via the exact same math as as (1), this implies  $|\Psimp\sqcap R|^2 \le M|P \sqcap 3C|$.
As $|\Psimp\sqcap R| > d_R$, then $U_R=|\Psimp\sqcap R|d_R \le |\Psimp\sqcap R|^2 \le M |P \sqcap 3C|$. Hence, we charge $M$ units to each vertex of $P$ inside the cube  $3C$.

Since $3C = 15R$ intersects a constant number of grid cells, no vertex gets charged more than a constant number of times by the above scheme. Thus, every vertex of either curve, gets charges $O(M)$ units overall, and the total number of vertex-edge pairs present in $R_{\le \Delta}(\Psimp,\Qsimp)$ is $O(nM)$ as claimed. Thus $\N(P,Q,\eps) \in O(nM) = O(\frac{\lambda^{2/d} n^{2(d-1) / d}}{\varepsilon^2})$. By \Cref{theorem:ApproxAlgDriemel}, this yields an algorithm which $(1+\eps)$-approximate $\fd(P,Q)$ in time $O\left(\frac{\lambda^{2/d} n^{2(d-1) / d}}{\varepsilon^2}+ \left(\lambda^{2/d} n^{2(d-1) / d}\right)\log n\right)$.
\end{proof}

\section{Conclusion}
We establish tight bounds for approximating the Fréchet distance between $\lambda$-low-density curves under the $\ell_1$ and $\ell_\infty$ norms, assuming OVH. For general $\ell_p$ norms and map matching, we obtain tight bounds in $n$ and $\lambda$, leaving the precise dependence on $\eps$ as an open problem.

\bibliographystyle{plainurl}
\bibliography{refs}

\subparagraph{Ai Statement.}
All lemmas, theorems, and results were obtained without the use of artificial intelligence. The paper is based on a master's thesis defended November 14$^{\text{th}}$ 2025. These facts can be verified by this paper's appearances on ArXiV and EUROCG (which we cannot link to due to double blindness).
In the final writing stages, large language models were used for proofreading and minor linguistic improvements. 
 We also used ChatGPT, Claude and Google DeepReview to check the proofs for errors, and all three certified our results as correct.

\appendix

\section{
\texorpdfstring{
Improved upper bound in $\eps$ under $l_1$ norm}{Improved upper bound in epsilon under l1 norm}}
\label{app:epsUnderl1norm}

We show that, for fixed constant $d\geq 4$, we can compute a $(1+\eps)$-approximation of the (discrete or continuous) Fréchet distance in time $\widetilde O\left(
 \frac{\lambda^{2/d}n^{2-2/d}}{\varepsilon^{2-4/d}}
 \right)$.
 The tilde hides logarithmic factors in $n$ and $1/\eps$. The hidden constants may depend on $d$.
 This improves our dependency on $\eps$ compared to the main body from $\eps^{-2}$ to $\varepsilon^{-(2-4/d)}$, matching the lower bound. 
In this section, all distances, edge lengths, and balls are measured in the $\ell_1$ norm. 
We consider two $\lambda$-low-density curves $P,Q\subset\mathbb R^d$ of at most $n$ vertices each, with $1\leq\lambda\leq n$. 
We use $\mathrm{FD}$ and $\mathrm{FDD}$ for the continuous and discrete Fréchet distance, respectively. 

\begin{definition}[Arclength parametrization]
For any curve $P$ we denote by $\ell(P)$ its length, i.e., the sum of the lengths of its edges. We view any curve $P$ as a map $P : [0, \ell(P)] \rightarrow \mathbb{R}^d$.    
\end{definition}

We deviate from the classical $\mu$-simplification (Definition~\ref{def:simplification}) for curves, and instead make a simplification based on arclength.
Fix some simplification value $r > 0$.
We partition the arclength parameter domain \([0,\ell(P)]\) into intervals, which we call $r$-bins.
A $r$-\emph{bin} is the closed interval $[k \cdot r, (k+1)\cdot r] \subset [0, \ell(P)]$ for some integer $k \geq 0$.
We define the \emph{exceptional bin} as the remaining interval $[k \cdot r, \ell(P)]$, where $k = \lfloor \ell(P)/r \rfloor$.
A bin is \emph{occupied} if it contains at least one vertex: i.e., there exists a vertex $p = P[x]$ where $x$ lies in the bin.

\begin{definition}[arclength simplification]
For a curve $P$ and $\mu > 0$, the \emph{arclength simplification} $\widehat P^\mu$ is defined by the sequence obtained by iterating over all $\mu$-bins in increasing order, adding the first and last vertex of $P$ in each occupied bin (we do not add consecutive duplicates).
\end{definition}

\begin{figure}[h!]
    \centering
    \includegraphics{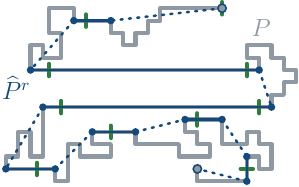}
    \caption{Curve $P$ and the \emph{arclength simplification} $\widehat P^r$. Green ticks along $P$ marks $P[rk]$ for integers $k\ge0$. The first and last vertex of each bin is marked as a blue vertex. The path $\widehat{P}^r$ is shown as alternating solid and dashed blue edges. Each solid blue edge coincide with a edge of $P$.}
    \label{fig:archlengthsimplification}
\end{figure}

We parameterize $\widehat P^\mu$ by its own arclength also. We observe the following:

\begin{lemma}\label{lem:app-l1-simplification}
The arclength-bin simplification $\widehat P^\mu$ can be constructed in $O(n)$ time and:
\begin{enumerate}
    \item \label{item:app-l1-discrete}
    $\mathrm{FD}(P,\widehat P^\mu)
\leq \mathrm{FDD}(P,\widehat P^\mu)\leq\mu$.
    \item \label{item:app-l1-length}
    For every integer $h\geq 0$, consider a $r$-bin $[a,b]$ of $P$ with $r=2^h\mu$. Then $P[a]$ and $P[b]$ lie on $\widehat P^\mu$, and the subcurve of $\widehat P^\mu$ between these two points has length at most $r$.
    \item \label{item:app-l1-count}
    Each such subcurve from $P[a]$ to $P[b]$ contains $O(1+r/\mu)$ vertices of $\widehat P^\mu$.
\end{enumerate}
\end{lemma}

\begin{proof}
We construct $\widehat P^\mu$ by simply iterating over the vertices of $P$, keeping track of the accumulated length. This identifies all occupied bins in linear time, and thus $\widehat P^\mu$.

For~\ref{item:app-l1-discrete}, we construct a monotone traversal as follows:
we advance along $P$, keeping $\widehat P^\mu$ at its current vertex until $P$ reaches the next vertex of $\widehat P^\mu$, at which point we advance both curves to that vertex.
Every vertex of $P$ that is not in $\widehat P^\mu$ lies in the same bin as the vertex that we remain standing on, so their distance is at most $\mu$.
Thus, $\mathrm{FDD}(P,\widehat P^\mu)\leq\mu$.
Since the continuous Fréchet distance is at most the discrete Fréchet distance, the claim follows.

For~\ref{item:app-l1-length}, consider a $r$-bin $[a,b]$ of $P$.
We first show that $P[a]$ and $P[b]$ lie on $\widehat P^\mu$.
Since $r=2^h\mu$, both $a$ and $b$ are boundaries of $\mu$-bins, or endpoints of $P$.
Consider either boundary.
If it is a vertex of $P$, it is retained.
Otherwise, it lies strictly inside an original edge $e$ of $P$.
The endpoints of $e$ are the last and first vertex of their respective $\mu$-bins, so both are retained and the edge is preserved (\Cref{fig:archlengthsimplification}). 
Thus, $P[a]$ and $P[b]$ lie on $\widehat P^\mu$.
The subcurve of $\widehat P^\mu$ between these two points replaces parts of the original bin by straight segments between the same endpoints.
And each segment is no longer than the part of $P$ it replaces.

Finally, we prove~\ref{item:app-l1-count}. Each $r$-bin intersects $O(1+r/\mu)$
$\mu$-bins, and each $\mu$-bin contributes at most two vertices to $\widehat P^\mu$.
\end{proof}

\begin{lemma}\label{lem:app-l1-approximation}
For two curves $P$ and $Q$, a threshold $\Delta>0$, and $0<\eps\leq 1$, set $\mu=\eps\Delta/16$ and $D=\Delta+4\mu$.
An exact decider for the continuous (or discrete) Fréchet distance between $\widehat P^\mu$ and $\widehat Q^\mu$ at threshold $D$ gives an \ApproxDecider for $P$ and $Q$ at threshold $\Delta$.
\end{lemma}

\begin{proof}
Our argument is identical for the continuous and discrete Fréchet distance. 
By item~\ref{item:app-l1-discrete} of Lemma~\ref{lem:app-l1-simplification} and the triangle inequality, if $\mathrm{FD}(P,Q)\leq\Delta$, then $
\mathrm{FD}(\widehat P^\mu,\widehat Q^\mu)
\leq\Delta+2\mu\leq D$. 
Thus, the decider returns \textsc{No} only when $\mathrm{FD}(P,Q)>\Delta$.

Conversely, if the decider returns \textsc{Yes}, then $
\mathrm{FD}(P,Q)
\leq D+2\mu=\Delta+6\mu\leq(1+\eps)\Delta$. 
\end{proof}

Finally, we define a very important value $L$ which we will use throughout our construction: 

\begin{definition}
    For a threshold $\Delta>0$, and $0<\eps\leq 1$, fix  $\mu$ and $D$ as in Lemma~\ref{lem:app-l1-approximation}.
    There exists an integer $h \geq 0$ with $h \in O(\log (2/\eps))$ such that $2^h\mu\in[D/2,D]$.
    We define $L := 2^h \mu$.
\end{definition}

\subsection{$\mu$-bins, $L$-bins, and $L$-Rectangles}
\label{app:ld-l1-roots}

For fixed $\eps$ and $\Delta$, we set $\mu = \eps \Delta/16$ and create $\widehat P^\mu$.
    Next, we split $P$ into \emph{pieces}:
    First, split the arclength parameter domain \([0,\ell(P)]\) into $L$-bins.
    Then merge each maximal consecutive run of unoccupied $L$-bins into one interval which we call a \emph{long piece}.
    For a long piece $[a,b]$, $P[a,b]$ lies inside an original edge $e$ of $P$ and has length $b-a$.
    Moreover, by Lemma~\ref{lem:app-l1-simplification}, $P[a]$ and $P[b]$ lie on $\widehat P^\mu$ and so the sub-edge of $e$ is a sub-edge of $\widehat P^\mu$ also.
    The remaining $L$-bins are called \emph{regular pieces}; and thus we have split $P$ into $O(n)$ \emph{pieces}. 
    We can compute these pieces in $O(n)$ time via a single linear scan.
    By  Lemma~\ref{lem:app-l1-simplification}, each piece $[a, b]$ has a corresponding subcurve on $\widehat P^\mu$ from $P[a]$ to $P[b]$, as $P[a]$ and $P[b]$ lie on $\widehat P^\mu$. 

Now consider the $D$-FSD($\widehat P^\mu$,$\widehat Q^\mu$).
Each piece $[a,b]$ corresponds to a subcurve of $\widehat P^\mu$ (or $\widehat Q^\mu$).
Thus, our split into pieces of $P$ and $Q$ creates rectangles in this free-space diagram between $\widehat P^\mu$ and $\widehat Q^\mu$ that we call \emph{$L$-rectangles} (\Cref{fig:Lrectangles}).
Each $L$-rectangle corresponds to a piece of $P$ and a piece of $Q$, and we distinguish $L$-rectangles created by two regular pieces (regular--regular), a regular and a long piece (regular--long), and two long pieces (long--long). An $L$-rectangle is \emph{nonempty} if its $D$-free space is nonempty.
An $L$-rectangle is \emph{reachable} if there exists a monotone path from $(0, 0)$ to the bottom or left facet of the rectangle in $D$-FSD($\widehat P^\mu$,$\widehat Q^\mu$).

\begin{figure}
    \centering
    \includegraphics{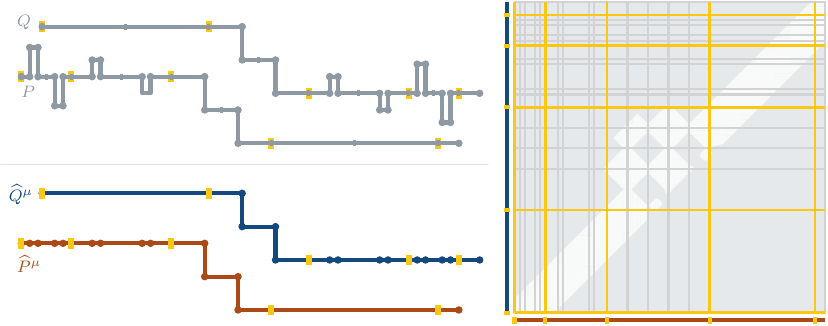}
    \caption{Left: The curves $P$ and $Q$ (grey) and the simplifications $\widehat P^\mu$ and $\widehat Q^\mu$ (red and blue). The points $P[k\cdot \mu]$ are marked as grey ticks and $P[k\cdot L]$ are marked as yellow ticks, for $L = 2^1 \mu$.  
    Right: $D$-FSD($\widehat P^\mu$,$\widehat Q^\mu$), parametrized by arch length, split into $L$-rectangles. The length of $\widehat P^\mu$ between $P[k\cdot L]$ and $P[(k+1)\cdot L]$ varies, so the rectangles have different side lengths. 
    }
    \label{fig:Lrectangles}
\end{figure}

\subsection{Algorithm}
\label{app:ld-l1-algorithm}

We describe our \ApproxDecider algorithm for the continuous Fréchet distance.
Given $\Delta$ and $\eps$, we compute $\mu$, $D$, $L$, $\widehat P^\mu$ and $\widehat Q^\mu$.
If the initial points of $\widehat P^\mu$ and $\widehat Q^\mu$ have distance greater than $D$, we return \textsc{No}.
Otherwise, we mark $(0,0)$ reachable and start at the $L$-rectangle that contains it.
From there, we process a diagonal sweepline of reachable $L$-rectangles by solving the following local problem:

The input of a local problem is an rectangle (not necessarily an $L$-rectangle) $I\times J \subset [0, \ell(\widehat P^\mu)] \times [0, \ell(\widehat Q^\mu)] $.
We additionally require as input all maximal intervals of reachable points on its bottom and left facets.
The desired output is all maximal intervals of reachable points on its top and right facets.
We solve the local problem differently, depending on the type of rectangle. 
The free space on one boundary of one free-space cell is an interval.
Its reachable subset is a suffix of that interval.
The total reachable boundary of rectangle that spans $k$ vertices of $\widehat P^\mu$ and $\ell$ vertices of $\widehat Q^\mu$ consequently has complexity $O(k+\ell)$.

Some rectangles can be solved in near-linear time.
This is the case for an $L$-rectangle with a long piece, or for a
rectangle whose subcurves $U$ of $\hat P^\mu$ and $V$ of $\hat Q^\mu$
are \emph{coordinate separated}.
We say that $U$ and $V$ are \emph{coordinate separated} if, for every
coordinate $j\in\{1,\ldots,d\}$, there is a real number $c_j$ such that
the $j$th coordinate of every point of $U$ is at most $c_j$ and the
$j$th coordinate of every point of $V$ is at least $c_j$, or vice versa.

\begin{lemma}
If an $L$-rectangle is defined by at least one long piece, its local problem can be solved in $O(k+\ell)$ time, where $k$ and $\ell$ are the numbers of vertices of its two subcurves.
\end{lemma}

\begin{proof}
The subcurve corresponding to a long piece is a single segment.
Thus, the $L$-rectangle consists of a single row or column of
$O(k+\ell)$ free-space cells.
We process these cells in order, computing the reachable intervals
on their boundaries in constant time per cell.
\end{proof}

\begin{lemma}\label{lem:app-l1-separated}
Let $I\times J$ be any rectangle in
$[0,\ell(\hat P^\mu)]\times[0,\ell(\hat Q^\mu)]$.
If its two subcurves are coordinate separated, its local problem can be solved exactly in $\widetilde O(k+\ell)$ time, where $k$
and $\ell$ are the numbers of vertices of its two subcurves.
\end{lemma}

\begin{proof}
Let $U$ and $V$ be the subcurves corresponding to $I$ and $J$.
We construct two separated one-dimensional curves $A$ and $B$
with the same free space as $U$ and $V$.
For each coordinate $j$, choose $\sigma_j\in\{-1,1\}$ such tha
$\sigma_j(u_j-v_j)\geq 0$ for all $u\in U$ and $v\in V$.
Then $\|u-v\|_1=\sigma\cdot u-\sigma\cdot v$.
We then set $A(s)=\sigma\cdot U(s)$ and $B(t)=\sigma\cdot V(t)$ (\Cref{fig:1Dprojection}). 
Thus, $A(s)\geq B(t)$ and
$|A(s)-B(t)|=\|U(s)-V(t)\|_1$ for every $(s,t)\in I\times J$.

Free-space reachability between separated one-dimensional curves is
well studied. Given $O(k+\ell)$ entrances and $O(k+\ell)$ potential
exits, \cite[Theorem~33]{vanderHorstEtAl2026Geodesic} reports which
potential exits are reachable from at least one entrance in
$O((k+\ell)\log(k+\ell))$ time.
This immediately solves the discrete local problem. Indeed, the continuous and
discrete Fréchet distances agree on separated one-dimensional
subcurves~\cite[Lemma~4.3]{BK14}. We use the reachable vertex pairs
on the bottom and left facets as entrances. As potential exits, we
use the vertex pairs on the top and right facets whose distance is
at most $D$.

To apply \cite[Theorem~33]{vanderHorstEtAl2026Geodesic} in the
continuous case, we need a set of entrances and a set of potential
exits. For each incoming reachable interval, we use its first point
as an entrance. Every other point of the interval is reachable from
it along the same side. We subdivide $A$ and $B$ at the entrance
parameters, adding $O(k+\ell)$ vertices.

On each outgoing cell side, we need only find the first reachable
point, since every free point after it is reachable along the same
side. Consider the top side of a cell. If the left side contains a
reachable point, every free point on the top side is reachable.
Otherwise, we can reach exactly the free points on the top side
that lie at or to the right of the first reachable point on the
bottom side.
Thus, as we move up a column of cells, the first reachable parameter
is either passed up from below or set to the start of the free
interval on the current top side. The candidates are therefore
edge endpoints of $A$, which include the entrances after
subdivision, and points satisfying $A(s)=D+B(t)$ for a vertex
parameter $t$ of $B$. The argument for the right side is symmetric.

We sort the values $D+B(t)$ over all vertex parameters $t$ of $B$.
For a free interval $[a,b]$ on a nonconstant edge of $A$, we use
the sorted values between $A(a)$ and $A(b)$. Each such value gives
one candidate point in $[a,b]$. We read these values in increasing
order if $A$ increases along the edge, and in decreasing order
otherwise. We also include the endpoints of the free interval.
On a constant edge, only the endpoints are needed. We proceed
symmetrically on the right side.
We binary search for the first reachable candidate on each side,
using \cite[Theorem~33]{vanderHorstEtAl2026Geodesic} to test one
candidate per side in each round. If no candidate is reachable,
the reachable interval is empty. 
For each edge of $A$, we binary search for the first point on that
edge that forms a reachable pair with the last point of $B$.
We do the same for each edge of $B$, pairing its points with the
last point of $A$. There are $O(k+\ell)$ searches, each with
$O(k+\ell)$ candidates. In each round, we test one candidate per
search using \cite[Theorem~33]{vanderHorstEtAl2026Geodesic}.
This halves the number of remaining candidates in every search.
If no candidate is reachable, the reachable interval is empty.
Thus, all searches finish after $O(\log(k+\ell))$ rounds,
taking $O((k+\ell)\log^2(k+\ell))$ time in total.
\begin{figure}
    \centering
    \includegraphics{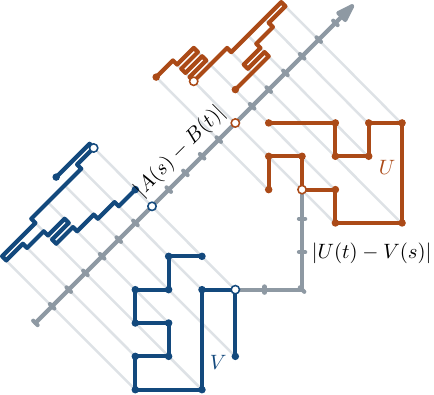}
    \caption{Projection of coordinate-separated curves $U$ and $V$ into one dimension curves $A(s)=\mu\cdot U(s)$ and $B(t)=\mu\cdot V(s)$ for $\mu = \begin{pmatrix}
        1 \\ 1
    \end{pmatrix}$.}
    \label{fig:1Dprojection}
\end{figure}

\end{proof}

\subsubsection{Solving the local problem for regular--regular $L$-rectangles}

We solve the local problem for a regular--regular $L$-rectangle
recursively. 
We first check whether we can compute the reachable
intervals on its top and right facets directly. If none of the
cases below applies, we split the rectangle and solve the local
problem for each of the resulting subrectangles. 

In a recursive call, let $[a,a']$ and $[b,b']$ be the two $r$-bins
of $P$ and $Q$ that define the rectangle, where $r=2^h\mu \le L$ for some integer $h$. In the first call, $r=L$.
By Lemma~\ref{lem:app-l1-simplification}, $P[a]$ and $P[a']$ lie on
$\hat P^\mu$. We denote by $U$ the subcurve of $\hat P^\mu$ between
these two points. We define $V$ analogously for the $r$-bin $[b,b']$
of $Q$.
By item~\ref{item:app-l1-length} of
Lemma~\ref{lem:app-l1-simplification}, both subcurves have length
at most $r$. Thus, for all $u\in U$ and $v\in V$, $
    \|P[a]-Q[b]\|_1-2r\leq\|u-v\|_1\leq\|P[a]-Q[b]\|_1+2r.$

\noindent
Our recursive procedure stops at any of the following cases:
\begin{enumerate}
    \item If $\|P[a]-Q[b]\|_1+2r\leq D$ then the entire rectangle is free.
    \item If $\|P[a]-Q[b]\|_1-2r>D$ then the entire rectangle has no free space.
    \item If $U$ and $V$ are coordinate separated, we apply
    Lemma~\ref{lem:app-l1-separated}.
\item If either subcurve has no vertex between its endpoints,
it is a single segment. The rectangle then consists of one row
or one column of free-space cells.
    \item If $r=\mu$, each subcurve has $O(1)$ vertices by
    item~\ref{item:app-l1-count} of
    Lemma~\ref{lem:app-l1-simplification}.
\end{enumerate}

Thus, in each of the above cases, we solve the local problem in
$\widetilde O(k'+\ell')$ time, where $k'$ and $\ell'$ are the numbers
of vertices of $U$ and $V$, respectively.
By item~\ref{item:app-l1-count} of Lemma~\ref{lem:app-l1-simplification},
we have $k'+\ell'=O(1+r/\mu)$. A terminating call on two $r$-bins
therefore takes $\widetilde O(1+r/\mu)$ time.

Otherwise, we split the two $r$-bins into $r/2$-bins. This splits
the rectangle into at most four subrectangles.
We start with the bottom-left
subrectangle. Its output gives the remaining input for the local
problems of the top-left and bottom-right subrectangles. After
solving both, we have the input for the top-right subrectangle,
which we solve last.

\subsection{Analysis of the algorithm}
\label{app:ld-l1-analysis}

To analyse the total time spent by our diagonal sweepline algorithm, we first consider the time spend on $L$-rectangles involving at least one long piece:

\begin{lemma}\label{lem:app-l1-root-count}
The diagonal sweepline spends $\widetilde O(\lambda nL/\mu)$
time in total on $L$-rectangles involving at least one long piece.
\end{lemma}

\begin{proof}
Every $L$-rectangle processed by the sweepline is nonempty.
We count the nonempty regular--long and long--long $L$-rectangles
separately.

Fix a regular piece defined by an $L$-bin $[a,a']$ of $P$. Let $U$ be the subcurve of $\widehat P^\mu$ from $P[a]$ to $P[a']$.
$U$ starts at $P[a]$ and has length
at most $L$. If a long piece of $Q$ forms a nonempty $L$-rectangle
with it, some point of the long piece is within distance $O(D)$ of $U$. Thus, the original edge of $Q$ containing the long
piece has length at least $L$ and intersects the ball of radius
$L+D$ around $P[a]$ (\Cref{fig:LongPieceIntersection}). Since $L+D = \Theta(L)$, we can cover this ball
by $O(1)$ balls of radius $L$. By low density, there are
$O(\lambda)$ such edges, and each contains at most one long piece.
There are $O(n)$ regular pieces. Each has $O(1+L/\mu)$ vertices
on its simplified subcurve by item~\ref{item:app-l1-count} of
Lemma~\ref{lem:app-l1-simplification}. By the preceding lemma,
the total time on regular--long $L$-rectangles is therefore
$O(\lambda nL/\mu)$.

For each nonempty long--long $L$-rectangle defined by the bin $[a,a']$ of $P$ and $[b,b']$ of $Q$, charge it to the
shorter piece. Assume $s = a'-a \le b'-b$. The
edge of $Q$ containing $Q[b,b']$ has length at least $s$ and
intersects the ball of radius $s+D$ around $P[a]$. Since
$s\geq L$, we can cover this ball by $O(1)$ balls
of radius $s$. Thus, each of the $O(n)$ long pieces is charged
$O(\lambda)$ times. Each such $L$-rectangle consists of one cell,
so the total time on long--long $L$-rectangles is $O(\lambda n)$.

\begin{figure}
    \centering
\includegraphics{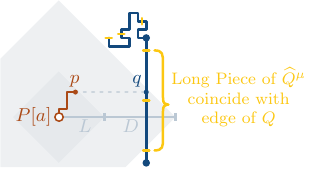}
    \caption{Two points $p\in \widehat P^\mu$ and $q\in \widehat Q^\mu$ within distance $D$ of each other. Ball of radius $L+D$ around $P[a]$. }
    \label{fig:LongPieceIntersection}
\end{figure}

\end{proof}

What remains is to bound the recursion of our algorithm.
Suppose that we subdivide a rectangle defined by two $r$-bins
$[a,a']$ of $P$ and $[b,b']$ of $Q$: which yields two curves $U$ and $V$ on $\hat P^\mu$ and $\hat Q^\mu$ respectively. Since neither of the first
two stopping cases applies, we have
$\bigl|\|P[a]-Q[b]\|_1-D\bigr|\leq 2r$.
Since $U$ and $V$ are not coordinate separated, their ranges
overlap in some coordinate $j$. Both subcurves have length at
most $r$, so the $j$th coordinates of $P[a]$ and $Q[b]$ differ
by at most $2r$.
We use \Cref{lem:app-l1-length,lem:app-l1-bins,lem:app-l1-pairs}
to bound the number of bin pairs satisfying these two conditions.

\begin{lemma}\label{lem:app-l1-length}
Let $\mathcal E$ be an $O(\lambda)$-low-density set of edges,
each of length at most $2r$. Consider a cube $C$ of side length
$O(r)$, and let $m_C$ be the number of edges of $\mathcal E$
that intersect $C$. The sum of the full lengths of these edges
is $O\bigl(r\lambda^{1/d}m_C^{1-1/d}\bigr)$.
\end{lemma}

\begin{proof}
If $m_C\leq\lambda$, the sum of the lengths is at most
$2rm_C\leq 2r\lambda^{1/d}m_C^{1-1/d}$. So suppose that
$m_C>\lambda$.

Let $C^*$ be the cube concentric with $C$ whose side length $h_0$
is $2r$ larger than that of $C$. Every edge intersecting $C$ has
length at most $2r$, so its midpoint lies in $C^*$. We construct the infinite $d$-dimensional quadtree
rooted at $C^*$, obtained by recursively halving every side of each hypercube. At level $j$,
there are $2^{jd}$ cubes of side length $h_j=h_0\,2^{-j}$.

For every $j\geq 0$, consider the edges intersecting $C$ whose
length is in $(h_j/2,h_j]$. Assign each such edge to a level-$j$
cube containing its midpoint. We can cover each cube by $O(1)$
balls of radius $h_j/2$. Every edge assigned to the cube intersects
one of these balls and has length greater than its radius.
By low density, at most $O(\lambda)$ edges are assigned to each
cube. Thus, there are at most $\min \{ m_C, O(\lambda 2^{jd}) \}$ edges at level
$j$.

Choose $J$ such that $2^{Jd}\leq m_C/\lambda<2^{(J+1)d}$.
We bound the number of edges assigned to level $j$ by
$O(\lambda 2^{jd})$ for $j\leq J$ and by $m_C$ for $j>J$.
Each edge assigned to level $j$ has length at most
$h_j=h_0\,2^{-j}$. Multiplying the number of edges at each level
by this length bound and using $h_0=O(r)$, we obtain that the sum
of the full edge lengths is at most
\[
    O\left(
        r\lambda\sum_{j=0}^{J}2^{j(d-1)}
        +rm_C\sum_{j>J}2^{-j}
    \right)
    =O\left(r\lambda 2^{J(d-1)}+rm_C2^{-J}\right).
\]
Since $2^J=\Theta((m_C/\lambda)^{1/d})$, both terms are
$O\bigl(r\lambda^{1/d}m_C^{1-1/d}\bigr)$.
\end{proof}

To apply Lemma~\ref{lem:app-l1-length} to occupied $r$-bins, we
replace long edges of $P$  as follows.

\begin{definition}
For a curve $P$ and $r>0$, consider each occupied $r$-bin $[a,b]$.
Cut $P$ at $P[a]$ and $P[b]$ and add the edges of the subcurve
between these points to $\mathcal E_r^P$.
\end{definition}

An occupied $r$-bin adds a segment of an edge to $\mathcal E_r^P$
only if it contains an endpoint of that edge. At each endpoint,
only one bin extends into the edge. Thus, each edge of $P$ adds
at most two segments to $\mathcal E_r^P$.
The set $\mathcal E_r^P$ therefore has at most $2n$ segments,
each of length at most $r$, and is $2\lambda$-low-density.

\begin{figure}
    \centering
    \includegraphics{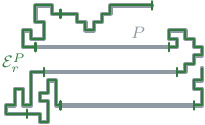}
    \caption{Green tricks along the path $P$ (grey) marks the points $P[kr]$ for integers $k$. $\mathcal E_r^P$ is marked in green.}
    \label{fig:Epr}
\end{figure}

\begin{lemma}\label{lem:app-l1-bins}
Consider a grid of side length $r$ in $\mathbb R^d$. Assign each
occupied $r$-bin of $P$, except for the exceptional bin, to the
grid cell containing its initial point, breaking ties arbitrarily.
For each grid cell $R$, let $b_R$ be the number of bins assigned
to $R$, and let $C$ be the cube concentric with $R$ of side
length $3r$. Define $m_C$ as the number of segments of $\mathcal E_r^P$
intersecting $C$, then $
    b_R=O\bigl(\lambda^{1/d}m_C^{1-1/d}\bigr)$.
The sum of $m_C$ over all grid cells $R$ is $O(n)$.
\end{lemma}

\begin{proof}
Fix a grid cell $R$ and let $C$ be the cube of side length $3r$ centred around $R$. Each bin assigned to $R$ starts in $R$ and
has length $r$, thus the entire bin is contained in $C$.
Each of these $b_R$ bins adds segments of total length $r$ to
$\mathcal E_r^P$, and all these segments intersect $C$. Thus,
by Lemma~\ref{lem:app-l1-length}, $
    rb_R \leq
    \sum\limits_{\substack{e\in\mathcal E_r^P\\e\cap C\neq\emptyset}} \ell(e)
    = O\bigl(r\lambda^{1/d}m_C^{1-1/d}\bigr)$.
Dividing by $r$ gives $b_R=O(\lambda^{1/d}m_C^{1-1/d})$.

Each segment of $\mathcal E_r^P$ has length at most $r$, so it
is counted by only $O(1)$ cubes. Since $\mathcal E_r^P$ has
at most $2n$ segments, we have $\sum_C m_C=O(n)$.
\end{proof}

\begin{lemma}\label{lem:app-l1-pairs}
Let $P$ and $Q$ be $\lambda$-low-density curves with at most $n$
vertices each.
Fix $0<r\leq D$. Consider pairs of occupied $r$-bins
$[a,a']$ of $P$ and $[b,b']$ of $Q$, excluding the exceptional
bin of either curve. At most $\widetilde O\bigl(\lambda^{2/d}n^{2-2/d}(D/r)^{2-4/d}\bigr)$
pairs satisfy:
\begin{itemize}
    \item $\bigl|\|P[a]-Q[b]\|_1-D\bigr|\leq 2r$, and
    \item There exists some coordinate, such that initial points $P[a]$ and $Q[b]$ differ by at most
    $2r$ in that coordinate.
\end{itemize}
\end{lemma}

\begin{proof}
Consider a grid of side length $r$, and assign each bin to the cell containing its initial point, as in Lemma~\ref{lem:app-l1-bins}. We call two cells $R$ and $S$ \emph{partners} if they contain the initial points of a pair of bins satisfying both conditions of the lemma.

Fix a grid cell $R$. We first bound the number of its partners. Let $S$ be a partner of $R$, and fix a pair of bins satisfying both conditions of the lemma, with $P[a]$ in $R$ and $Q[b]$ in $S$.
Moving the $P[a]$ and $Q[b]$ to the centres of their cells changes their $\ell_1$ distance and each coordinate difference by at most $O(r)$. Hence, the centres of $R$ and $S$ have $\ell_1$ distance $D\pm O(r)$ and differ by at most $O(r)$ in some coordinate.
Fix a coordinate $i$. There are $O(1)$ possible grid cells whose centre differ by at most $O(r)$ to the center of $R$, since each each cell have side length $r$.
Of the remaining $d-1$ coordinates, we choose $d-2$ freely, each with $O(D/r)$ possibilities, since the centres are at distance $O(D)$. Once these coordinates are fixed, the last coordinate must bring the $\ell_1$ distance to within $O(r)$ of $D$. This leaves at most two intervals of length $O(r)$ for the last coordinate, corresponding to $O(1)$ cell centres. 
Consequently, there are $O((D/r)^{d-2})$ partners for each choice of $i$. Summing over the $d$ coordinates gives $O((D/r)^{d-2})$ partners in total, since $d$ is constant.

For each grid cell $R$, let $C_R$ be the concentric cube of side
length $3r$, and let $m_{C_R}$ count the segments of
$\mathcal E_r^P\cup\mathcal E_r^Q$ that intersect $C_R$.
By Lemma~\ref{lem:app-l1-bins}, we have $\sum_R m_{C_R}=O(n)$.
The same lemma bounds the number of bins of either curve assigned
to $R$ by $O(\lambda^{1/d}m_{C_R}^{1-1/d})$. Thus, a pair of
partner cells $R,S$ accounts for at most
$O(\lambda^{2/d}(m_{C_R}m_{C_S})^{1-1/d})$ bin pairs. Summing over all pairs of cells $R,S$ and using Hölders's inequality gives the bound:
\begin{equation} \label{eq:pairsOfCells}
     \sum_{R,S} \lambda^{2/d}(m_{C_R}m_{C_S})^{1-1/d}
    \leq \lambda^{2/d}
    \left(\sum_{R,S} m_{C_R}m_{C_S}\right)^{1-2/d}
    \left(\sum_{R,S}\sqrt{m_{C_R}m_{C_S}}\right)^{2/d}
\end{equation}
   
We can bound the two sums as such: $\sum_{R,S} m_{C_R}m_{C_S} \leq \sum_R m_{C_R}\sum_S m_{C_S}  = O(n^2)$. Also, $\sum_{R,S}\sqrt{m_{C_R}m_{C_S}} \le \sum_{R,S}(m_{C_R}+m_{C_S}) \le \sum_{R}(m_{C_R}) + \sum_{S}(m_{C_S}) = O(Mn)$, since each cell has at most $M$ partners. Plugging this into equation (\ref{eq:pairsOfCells}) and substituting $M=O((D/r)^{d-2})$ bounds the total number of bin pairs:
\[
    \sum_{R,S}\lambda^{2/d}(m_{C_R}m_{C_S})^{1-1/d}
    = O\bigl(\lambda^{2/d}n^{2-2/d}M^{2/d}\bigr) = O\bigl(\lambda^{2/d}n^{2-2/d}(D/r)^{2-4/d}\bigr)
\]
\end{proof}

\begin{lemma}\label{lem:app-l1-refined}
For each $r\in\{L,L/2,\ldots,\mu\}$, the algorithm subdivides
at most $\widetilde O\bigl(\lambda^{2/d}n^{2-2/d}(D/r)^{2-4/d}\bigr)$ rectangles
defined by two $r$-bins.
\end{lemma}

\begin{proof}
If the algorithm subdivides a rectangle defined by the
$r$-bins $[a,a']$ of $P$ and $[b,b']$ of $Q$, then neither
of the first two stopping cases applies thus $
    \bigl|\|P[a]-Q[b]\|_1-D\bigr|\leq 2r$.
The two subcurves are not coordinate separated, so their ranges
overlap in some coordinate $j$. Each subcurve has length at most
$r$, so the $j$th coordinates of $P[a]$ and $Q[b]$ differ by at
most $2r$.

If either subcurve had no vertex between its endpoints, we
would solve the local problem without subdividing. Thus, each
subcurve contains a vertex of the simplified curve strictly
between its endpoints. These vertices are also vertices of
the original curves, so both $r$-bins are occupied. Thus, Lemma~\ref{lem:app-l1-pairs} bounds
the number of such pairs by $\widetilde O\bigl(\lambda^{2/d}n^{2-2/d}(D/r)^{2-4/d}\bigr)$
when neither bin is exceptional. There is at most one exceptional
bin per curve, and it can be paired with at most $O(n)$ occupied
$r$-bins of the other curve. Each bin pair defines one rectangle,
so the exceptional bins add at most $O(n)$ rectangles.
\end{proof}

The preceding lemma bounds how many rectangles we subdivide at
each scale $r$. Each subdivision creates at most four subrectangles.
The time spent on these subrectangles, excluding any further
recursive calls, is $\widetilde O(r/\mu)$. We charge this time
to the rectangle we subdivided and sum over all scales:

\begin{lemma}\label{lem:app-l1-decider}
Let $\Delta>0$ and $\eps \in (0, 1]$. Then
$\ApproxDecider(P,Q,\Delta,\eps)$ solves the (continuous or discrete) approximate decision
problem  under the 
$\ell_1$ norm in time $\widetilde O\left(\frac{\lambda^{2/d}n^{2-2/d}}{\eps^{2-4/d}}\right)$.
\end{lemma}
\begin{proof}
We construct the simplifications and the pieces in $O(n)$ time.
By Lemma~\ref{lem:app-l1-root-count}, the total time spent on
$L$-rectangles involving a long piece is
$\widetilde O(\lambda nL/\mu)$.

It remains to bound the time spent on regular--regular
$L$-rectangles. I.e., the time spent by our recursive procedure. Recall that a terminating
call on two $r$-bins takes $\widetilde O(1+r/\mu)$ time.
A call that subdivides the rectangle takes the same time
before its recursive calls, since it only checks the stopping
cases and passes the boundary intervals to at most four
subrectangles.

We first count the calls at scale $L$. If a regular--regular
$L$-rectangle is nonempty, the initial points of its two pieces
are at distance at most $D+2L$. Since $D\leq 2L$, for each grid
cell of side length $L$ there are only $O(1)$ other cells that
can contain the initial point of the other piece. The counting
argument in Lemma~\ref{lem:app-l1-pairs}, with $O(1)$ partner
cells per cell, gives $O(\lambda^{2/d}n^{2-2/d})$ pairs of bins. The exceptional bins
add $O(n)$ pairs.

At each smaller scale $r$, every call comes from a rectangle
that was subdivided at scale $2r$. Each such rectangle gives
at most four calls. Thus, by Lemma~\ref{lem:app-l1-refined},
there are $\widetilde O(\lambda^{2/d}n^{2-2/d}(D/r)^{2-4/d})$ calls at scale $r$.
This bound also holds at scale $L$ by the count above.

At scale $r$, the total time, excluding calls at smaller scales,
is therefore
\[
    \widetilde O\left(\lambda^{2/d}n^{2-2/d}(D/r)^{2-4/d}\frac r\mu\right).
\]
Since $2-4/d\geq 1$ and $\mu\leq r\leq D$, this is
at most $\widetilde O\bigl(\lambda^{2/d}n^{2-2/d}(D/\mu)^{2-4/d}\bigr)$ at each scale.
There are $O(\log(2/\eps))$ scales. The time spent on
$L$-rectangles involving a long piece is within the same bound
since $\lambda n\leq \lambda^{2/d}n^{2-2/d}$ and $L\leq D$. Finally, since
$D/\mu=\Theta(1/\eps)$, the total time is
$\widetilde O\left(\lambda^{2/d}n^{2-2/d}/\eps^{2-4/d}\right)$.

In the discrete case, a rectangle contains only pairs of vertices
of the simplified curves. If either bin contains no such vertex,
there are no pairs to process. Thus, whenever we subdivide a
rectangle, both bins contain a vertex of the original curve, and
we can apply the same bound on pairs of occupied bins.
\end{proof}

\begin{theorem}\label{thm:app-l1-improved}
Let $d\geq4$ be fixed, and let $P,Q\subset\mathbb R^d$ be two $\lambda$-low-density polygonal curves with at most $n$ vertices each, where $1\leq\lambda\leq n$. For every $0<\varepsilon\leq1$, a $(1+\varepsilon)$-approximation of the continuous or discrete Fr\'echet distance under the $\ell_1$ norm can be computed in time
\[
 \widetilde O\left(
 \frac{\lambda^{2/d}n^{2-2/d}}{\varepsilon^{2-4/d}}
 \right).
\]
\end{theorem}
\begin{proof}
We first find a constant-factor interval containing the
$\ell_1$ Fréchet distance. A Euclidean ball of radius $\rho$
can be covered by $O(1)$ $\ell_1$ balls of radius $\rho$, and
every edge of Euclidean length at least $\rho$ has $\ell_1$
length at least $\rho$. Thus, $P$ and $Q$ are also
$O(\lambda)$-low-density under $\ell_2$.

By \Cref{theorem:ApproxAlgDriemel,theorem:TwosidedLambdaLowDense},
we can compute a $2$-approximation $x$ of the Euclidean Fréchet
distance in $\widetilde O(\lambda^{2/d}n^{2-2/d})$ time. 
For the discrete case, the discrete Fréchet distance is attained
by a pair of vertices, one from each curve. Thus, by
\cite[Theorem~1]{AronovHKWW06}, we obtain a $2$-approximation $x$
of the discrete Euclidean Fréchet distance using $O(n\log n)$
preprocessing time and $O(\log n)$ calls to the discrete
\ApproxDecider of Lemma~\ref{lem:app-l2-decider} at constant
accuracy. This takes $\widetilde O(\lambda^{2/d}n^{2-2/d})$ time.

Since
$\|u\|_2\leq\|u\|_1\leq\sqrt d\,\|u\|_2$, the $\ell_1$ Fréchet
distance lies in $[x/2,\sqrt d\,x]$. If $x=0$, we return zero.

Otherwise, we apply the interval search from the proof of
Lemma~\ref{lemma:searchingConstantFactorInterval}, using the
\ApproxDecider from Lemma~\ref{lem:app-l1-decider}.
The search makes $O(\log(2/\eps))$ calls with approximation
parameters $\Omega(\eps)$. Thus, the total time is
$\widetilde O\left(\lambda^{2/d}n^{2-2/d}/\eps^{2-4/d}\right)$.
The same argument applies to the discrete Fréchet distance.
\end{proof}

\section{
\texorpdfstring{
Improved upper bound in $\eps$ under the $\ell_2$ norm}{Improved upper bound in epsilon under the l2 norm}}
\label{app:epsUnderl2norm}

We show that, for fixed constant $d\geq 2$, we can compute a $(1+\eps)$-approximation of the (discrete or continuous) Fréchet distance in time $\widetilde O\left(
 \frac{\lambda^{2/d}n^{2-2/d}}{\varepsilon^{3/2-1/d}}
 \right)$.
The tilde hides logarithmic factors in $n$ and $1/\eps$. The hidden constants may depend on $d$.
This improves our dependency on $\eps$ compared to the main body from $\eps^{-2}$ to $\varepsilon^{-(3/2-1/d)}$.

We use the same arclength simplification, decomposition into
$L$-rectangles, and diagonal sweepline as in the $\ell_1$
setting. The main difference is how we solve the local problem.
In the $\ell_1$ setting, coordinate-separated subcurves can
be reduced exactly to separated one-dimensional curves.
In the $\ell_2$ setting, we instead recursively subdivide
rectangles until their subcurves are sufficiently short.
We then project both subcurves onto the line through their
initial points, obtaining separated one-dimensional curves.
Although this projection does not preserve distances exactly,
the short lengths of the subcurves ensure that the resulting
error is small enough for our approximation.

This difference also affects the running time.
In the $\ell_1$ setting, rectangles that require further
subdivision must have subcurves overlapping in at least one
coordinate. This restricts the number of relevant pairs of
bins. In the $\ell_2$ setting, we do not have this additional
restriction, but we compensate by terminating the recursion
at a larger scale.

In this section, all distances, edge lengths, and balls are measured in the $\ell_2$ norm.
We consider two $\lambda$-low-density curves $P,Q\subset\mathbb R^d$ of at most $n$ vertices each, with $1\leq\lambda\leq n$.
We use $\mathrm{FD}$ and $\mathrm{FDD}$ for the continuous and discrete Fréchet distance, respectively.

\subsection{Algorithm}

Given $\Delta$ and $\eps$, we compute $\mu$, $D$, $L$, $\widehat P^\mu$ and $\widehat Q^\mu$ as in Appendix~\ref{app:epsUnderl1norm}.
The arclength simplification and the division into pieces are defined in the same way, using Euclidean arclength. The proof of Lemma~\ref{lem:app-l1-simplification} applies in the $\ell_2$ norm also.

We use the same diagonal sweepline over $L$-rectangles, as in Appendix~\ref{app:ld-l1-algorithm},  and
the same local problem recursively for each rectangle.

We use the same stopping cases for completely free rectangles, completely blocked rectangles, and rectangles consisting of a single row or column as in Appendix~\ref{app:ld-l1-algorithm}. 
In the setting of $\ell_1$, the remaning stopping case is at scale $\mu$.
Rather than recursing down to rectangles of scale $\mu$, we introduce a stopping scale $r_0$ and terminate the recursion once this scale is reached. We set $r_0$ as the largest scale in $\{L,L/2,\ldots,\mu\}$ that is at most $\sqrt{\eps}D/8$. Since $\mu=\eps\Delta/16$ and $L\geq D/2$, such a scale exists and $r_0=\Theta(\sqrt{\eps}D)$
In contrast to the $\ell_1$ setting, we do not check for coordinate
separation. 
If none of the stopping cases applies, we split the regular regular rectangle and solve the local problem for each of the resulting subrectangles, which we describe in the following section.

\subsubsection{\texorpdfstring{Solving the local problem for regular--regular $L$-rectangles}{Solving the local problem for regular--regular L-rectangles}}

We solve the local problem for a regular--regular $L$-rectangle recursively. 
We process them in an order that gives us the reachable intervals on the bottom and left facets of each subrectangle before we solve it.

In a recursive call, let $[a,a']$ and $[b,b']$ be the two $r$-bins
of $P$ and $Q$ that define the rectangle. In the first call, $r=L$.
We denote by $U$ the subcurve of $\widehat P^\mu$ from $P[a]$ to $P[a']$ and define $V$ analogously for the $r$-bin $[b,b']$ of $Q$.
By item~\ref{item:app-l1-length} of Lemma~\ref{lem:app-l1-simplification}, both subcurves have length at most $r$.

\begin{lemma}\label{lem:app-l2-projection}
Suppose that $r=r_0$ and that neither the completely free nor the completely blocked test applies. We can solve the local problem in $\widetilde O(k+\ell)$ time using a free space that contains the Euclidean $D$-free space and is contained in the Euclidean $(D+\eps D/32)$-free space. Here $k$ and $\ell$ are the numbers of vertices of the two subcurves.
\end{lemma}

\begin{proof}
We construct two separated one-dimensional curves by projecting $U$ and $V$ onto the line through $P[a]$ and $Q[b]$, as in~\cite{BK14}. We use the same parameters for the projected curves.
Since neither test applies, $\|P[a]-Q[b]\|_2>D-2r$. Each subcurve stays within distance $r$ of its initial point, and $r\leq D/8$. Thus, the projections are separated.

Let $u\in U$ and $v\in V$, and let $u'$ and $v'$ be their projections.
The perpendicular component of $u-v$ has length at most $2r$.
Thus, $\|u'-v'\|_2\leq\|u-v\|_2$, and if $\|u'-v'\|_2\leq D$ then
\[
    \|u-v\|_2\leq \sqrt{D^2+4r^2}
    \leq D+\frac{2r^2}{D}
    \leq D+\frac{\eps D}{32}.
\]
We solve the local problem on the projected curves at threshold $D$ using the algorithm for separated one-dimensional curves in the proof of Lemma~\ref{lem:app-l1-separated}.
\end{proof}

When passing reachable intervals to another rectangle, we intersect them with the Euclidean $D$-free space on the shared facet. Thus, every monotone traversal in the Euclidean $D$-free space remains feasible. Conversely, a monotone traversal found by the sweepline is a concatenation of monotone traversals in the rectangles, and each stays in the Euclidean $(D+\eps D/32)$-free space.

The boundary representation can include cuts at $r$-bin boundaries as well as at vertices of the simplified curves. Each part of a facet between consecutive vertices or cuts is on the boundary of one free-space cell of a terminating call. The reachable subset is empty or a suffix of its free interval, and intersecting it with the Euclidean $D$-free space leaves an interval. Thus, the cuts add at most $O(1+r/r_0)$ intervals to the boundary representation of a rectangle at scale $r$.
By item~\ref{item:app-l1-count} of Lemma~\ref{lem:app-l1-simplification}, a call on two $r$-bins takes $\widetilde O(1+r/\mu)$ time, excluding any further recursive calls.

Otherwise, we split the two $r$-bins into $r/2$-bins. This splits the rectangle into at most four subrectangles.
We start with the bottom-left subrectangle. Its output gives the remaining input for the local problems of the top-left and bottom-right subrectangles. After solving both, we have the input for the top-right subrectangle, which we solve last.

In the discrete case, we assign each vertex of a simplified curve to the bin containing its original parameter, breaking ties in favour of the earlier bin. We use the first and last rows and columns of vertex pairs as the boundary of a rectangle. We omit bins with no vertices and pass reachability between consecutive nonempty groups of vertices using the usual horizontal, vertical, and diagonal steps of the discrete Fréchet distance. When a step enters another rectangle, we test its destination at Euclidean threshold $D$. The number of such steps is linear in the number of boundary vertex pairs, and the same guarantees hold for the discrete traversal.

\subsection{Analysis of the algorithm}

The proof of Lemma~\ref{lem:app-l1-root-count} applies in the $\ell_2$ norm also. Thus, the total time spent on $L$-rectangles involving a long piece is $\widetilde O(\lambda nL/\mu)$.
It remains to bound the time spent on regular--regular $L$-rectangles.

\begin{lemma}\label{lem:app-l2-pairs}
Let $P$ and $Q$ be $\lambda$-low-density curves with at most $n$
vertices each.
Fix $0<r\leq D$. Consider pairs of occupied $r$-bins
$[a,a']$ of $P$ and $[b,b']$ of $Q$, excluding the exceptional
bin of either curve. At most
$\widetilde O\bigl(\lambda^{2/d}n^{2-2/d}(D/r)^{2-2/d}\bigr)$
pairs satisfy $\bigl|\|P[a]-Q[b]\|_2-D\bigr|\leq 2r.$
\end{lemma}

\begin{proof}
Consider a grid of side length $r$, and assign each bin to the
cell containing its initial point as in Lemma~\ref{lem:app-l1-bins}.
We call $R$ and $S$ partners if a bin of $P$ assigned to $R$
and a bin of $Q$ assigned to $S$ form a pair satisfying the
condition of the lemma.

Fix a grid cell $R$. We first bound the number of its partners.
Let $S$ be a partner of $R$, and fix a pair of bins satisfying
the condition of the lemma, with $P[a]$ in $R$ and $Q[b]$ in $S$.
Moving $P[a]$ to the centre of $R$ and $Q[b]$ to any point in $S$
changes their $\ell_2$ distance by at most $O(r)$.
By the condition of the lemma, their original distance is
within $2r$ of $D$.
Thus, $S$ lies in the region between two concentric Euclidean
balls centred at the centre of $R$, with radii
$\max\{0,D-cr\}$ and $D+cr$ for a constant $c>0$.

Since $r\leq D$, this region has volume $O(D^{d-1}r)$.
Each grid cell has volume $r^d$, so $R$ has at most
$M=O((D/r)^{d-1})$ partners.
The same bound holds when fixing $S$ and counting its partners $R$.

The remaining counting argument is identical to that in the proof
of Lemma~\ref{lem:app-l1-pairs}, since the bounds of
Lemma~\ref{lem:app-l1-bins} also apply in the $\ell_2$ norm.
In particular, the number of bin pairs is at most
$O(\lambda^{2/d}n^{2-2/d}M^{2/d})$.
Substituting $M=O((D/r)^{d-1})$ gives
$O\bigl(\lambda^{2/d}n^{2-2/d}(D/r)^{2-2/d}\bigr)$.
\end{proof}





\begin{lemma}\label{lem:app-l2-decider}
Let $\Delta>0$ and $\eps\in(0,1]$. Then \ApproxDecider$(P,Q,\Delta,\eps)$ solves the (continuous or discrete) approximate decision problem under the $\ell_2$ norm. Its running time is $\widetilde O\left(\lambda^{2/d}n^{2-2/d}/\eps^{3/2-1/d}\right)$.
\end{lemma}

\begin{proof}
By the proof of Lemma~\ref{lem:app-l1-approximation}, if the original distance is at most $\Delta$, the distance between the simplifications is at most $D$, so the algorithm returns \textsc{Yes}. Conversely, if the algorithm returns \textsc{Yes}, the distance between the simplifications is at most $D+\eps D/32$. By item~\ref{item:app-l1-discrete} of Lemma~\ref{lem:app-l1-simplification} and the triangle inequality, the original distance is at most $D+\eps D/32+2\mu\leq(1+\eps)\Delta$.

We construct the simplifications and the pieces in $O(n)$ time.
We count the regular--regular calls at scale $L$ as in the proof of Lemma~\ref{lem:app-l1-decider}. Every $L$-rectangle we process has an entrance in the Euclidean $D$-free space. Thus, the initial points of its pieces are at distance at most $D+2L$. Thus, each grid cell of side length $L$ has $O(1)$ partner cells, giving $O(\lambda^{2/d}n^{2-2/d})$ pairs of bins. The exceptional bins add at most $O(n)$ pairs.

At each smaller scale $r$, every call comes from a rectangle that was subdivided at scale $2r$. Each such rectangle gives at most four calls.
If we subdivide a rectangle, then neither the completely free nor the completely blocked test applies, so the initial points satisfy the condition of Lemma~\ref{lem:app-l2-pairs}.
In the continuous case, we would solve the local problem without subdividing if either subcurve were a single segment. Thus, both bins contain a vertex of the simplified curve strictly between their endpoints, and both bins are occupied.
In the discrete case, we subdivide only if both bins contain a vertex of the simplified curve, so the same count applies.
The exceptional bins add at most $O(n)$ rectangles at each scale. Since $\lambda^{2/d}n^{2-2/d}\geq n$, there are $O(\lambda^{2/d}n^{2-2/d}(D/r)^{2-2/d})$ calls at scale $r$.

At scale $r$, the total time, excluding calls at smaller scales, is therefore
\[
    \widetilde O\left(\lambda^{2/d}n^{2-2/d}(D/r)^{2-2/d}\frac r\mu\right).
\]
Since $2-2/d\geq 1$, the bound is largest at $r=r_0$.
Since $D/r_0=\Theta(1/\sqrt{\eps})$ and $r_0/\mu=\Theta(1/\sqrt{\eps})$, the time at each scale is at most $\widetilde O(\lambda^{2/d}n^{2-2/d}/\eps^{3/2-1/d})$.
There are $O(\log(2/\eps))$ scales. The time spent on $L$-rectangles involving a long piece is within the same bound since $\lambda n\leq \lambda^{2/d}n^{2-2/d}$ and $L/\mu=O(1/\eps)$.
\end{proof}

\begin{theorem}\label{thm:app-l2-improved}
Let $d\geq2$ be fixed, and let $P,Q\subset\mathbb R^d$ be two $\lambda$-low-density polygonal curves with at most $n$ vertices each, where $1\leq\lambda\leq n$. For every $0<\varepsilon\leq1$, a $(1+\varepsilon)$-approximation of the continuous or discrete Fréchet distance under the $\ell_2$ norm can be computed in time
\[
    \widetilde O\left(\min\left\{n^2,\,
        \frac{\lambda^{2/d}n^{2-2/d}}{\varepsilon^{3/2-1/d}}
    \right\}\right).
\]
\end{theorem}

\begin{proof}
We first find a constant-factor interval containing the $\ell_2$ Fréchet distance.
By \Cref{theorem:ApproxAlgDriemel,theorem:TwosidedLambdaLowDense}, we can compute a $2$-approximation $x$ of the Euclidean Fréchet distance in $\widetilde O(\lambda^{2/d}n^{2-2/d})$ time. 
For the discrete case, the discrete Fréchet distance is attained
by a pair of vertices, one from each curve. Thus, by
\cite[Theorem~1]{AronovHKWW06}, we obtain a $2$-approximation $x$
of the discrete Euclidean Fréchet distance using $O(n\log n)$
preprocessing time and $O(\log n)$ calls to the discrete
\ApproxDecider of Lemma~\ref{lem:app-l2-decider} at constant
accuracy. This takes $\widetilde O(\lambda^{2/d}n^{2-2/d})$ time.

The Fréchet distance lies in $[x/2,x]$. If $x=0$, we return zero.

Otherwise, we use the \ApproxDecider from Lemma~\ref{lem:app-l2-decider} in the interval search from the proof of Lemma~\ref{lemma:searchingConstantFactorInterval}.
The search makes $O(\log(2/\eps))$ calls with approximation parameters $\Omega(\eps)$.
Thus, the total time is $\widetilde O\left(\lambda^{2/d}n^{2-2/d}/\eps^{3/2-1/d}\right)$.
The same argument applies to the discrete Fréchet distance.
If this bound exceeds the time for the exact algorithm of Alt and Godau~\cite{alt1995computing} or the discrete dynamic program, respectively, we use that algorithm instead.
\end{proof}

\section{
\texorpdfstring{
Tight bounds of simplifying $\lambda$-low dense curves}{Tight bounds for simplifying lambda-low dense curves}}
\label{app:simple_lambdaPrime}

For any $c$-packed curve, one can show that its $\mu$-simplification, for any $\mu$, is $O(c)$-packed, as shown in \cite{DriemelHW12}. This is a nice property, and it is frequently used when designing efficient algorithms for $c$-packed curves. We show an analogous result for $\lambda$-low dense curves, as we show tight bounds for the $\lambda'$-low density of their $\mu$-simplifications. First our upper bound:

\simpleIsLowDense*

Note that this result arguably not as strong as the result for $c$-packed curves, as $\lambda'$ degrades as $n$ increases. However, we show at the end of this section that this bound is tight. I.e., that there exist families of curves where, for a suitable choice of $\mu$, their $\mu$-simplification is $\lambda'$-low dense where $\lambda' \in \Theta(n^{1-1/d} \lambda^{1/d})$.
Our proof follows a similar argument to the main body, where we relate the curve length in a region to the number of nearby vertices.

\begin{lemma} \label{lemma:VerticiesIn3uC}
Let $P$ be a $\lambda$-low-density polygonal curve in $\mathbb{R}^d$, and let $C$ be a cube of side length $r$. Suppose that the total length of $P$ inside $C$ is $\alpha$. Let $u$ be the longest edge  of $P$ intersecting $C$. If $m$ denotes the number of vertices of $P$ contained in the cube $C'$ centred at $C$ with side length $3\max(\|u\|,r)$, then $m=\Omega\!\left( \frac{(\alpha/r)^{d/(d-1)}}{\lambda^{1/(d-1)}} \right)$.
\end{lemma}

\begin{proof}
    This proof is almost entirely analogous that of \Cref{lemma:VerticesIn3C}, with the exception that the cube $C'$ will include all endpoints of edges interesting $C$. There is therefore no need to subtract the endpoints of edges in $E_{0}$.  Thus, we can drop the assumption that $\alpha\ge c\lambda r$.
\end{proof}

Let $P$ be an $\lambda$-low dense  curve and let $\Psimp$ be a $\mu$-simplification of $P$ for any $\mu \geq 0$. Let $\lambda'$ be the smallest value such that $\Psimp$ is $\lambda'$-low dense. We will relate $\lambda$ and $\lambda'$ by finding a ball that contains a substantial portion of the length of $P$ inside it, where 'substantial' is described as a function of $\lambda'$. We will then relate the length of $P$ inside that ball to the number of nearby vertices of $P$ as a function of $\lambda$ and $\lambda'$. As the number of vertices of $P$ is bounded by $n$ we can thus bound $\lambda'$ as a function $\lambda$ and $n$.

\begin{definition}\textit{(the defining ball)} 
    Let for a curve $P$, $\lambda \in \mathbb{R}$ be the minimal value such that $P$ is a $\lambda$-low dense curve. There exist at least one ball $B(p,r)$ centred at $p \in \mathbb{R}^d$ with radius $r$ such that $\lambda$ edges of $P$, each of length at least $r$, intersect $B(p,r)$. We denote any such ball $B(p,r)$ to be a \emph{defining} ball of $P$.
\end{definition}

\begin{lemma}\label{lemma:DefiningballRadius}
    Let $P$ be an $\lambda$-low dense  curve and let $\Psimp$ be a $\mu$-simplification of $P$ for any $\mu \geq 0$.  Let $\lambda' \in \mathbb{R}$ be the minimal value such that $\Psimp$ is $\lambda'$-low dense. If $\lambda' > \lambda$ then there exists a defining ball of $\Psimp$ with radius at least $\mu$.
\end{lemma}

\begin{proof}
    Suppose the defining ball of $\Psimp$ is $B(p,r)$ where $r< \mu$. The ball $B(p,r)$ is intersected by $\lambda'$ edges of $\Psimp$ with length $\ge r$. All these edges have length $\ge \mu$, since all edges of $\Psimp$ have length $\ge \mu$. The ball $B(p,\mu)$ must therefore also intersect $\lambda'$ edges of length $\ge \mu$ and is thus a defining ball of $\Psimp$. 
\end{proof}

\begin{lemma} \label{lemma:PinsideDefiningBall}
    Let P be a $\lambda$-low dense curve of complexity $n$ and let $\Psimp$ be a $\mu$-simplification of $P$ for any $\mu \geq 0$.  Let $\lambda' \in \mathbb{R}$ be the minimal value such that $\Psimp$ is $\lambda'$-low dense. There exist a ball $B(p,3r)$, with $r \geq \mu$, such that  the total length of $P$ inside $B(p, 3r)$ is at least $\lambda'r$.
\end{lemma}

\begin{proof}
    Lemma \ref{lemma:DefiningballRadius} tells that there exists a defining ball $B(p,r)$ of $\Psimp$ such that $r \ge \mu$. The ball $B(p,r)$ intersects $\lambda'$ edges of $\Psimp$ of length $\ge r$ . Consider an edge $e$ of $\Psimp$ that intersects $B(p,r)$. Let $P_e$ be the subpath of $P$, which has been shortcut into the edge $e$. Let $v$ be a point on $e$ that is in $B(p,r)$. There exists a point $u$ in the subpath $P_e$ within distance $\mu$ to $v$. 
    Consider the ball $B(u, r)$ around $u$ of radius $r$.
    Since the length of $e$ is at least $r$, and $e$ defines the start and end vertices of $P_e$, the length of $P_e$ is also at least $r$. It follows that   $\|P_e\cap B(u,r)\| \ge r$ (See \Cref{fig:PinsideDefiningBall}).
    As the ball $B(u,r)$ is contained in the ball $B(p,3r)$ we have $\|P_e \cap B(p,3r)\| \ge r$. Thus, $\|P \cap B(p,3r)\| \ge \lambda'r$ in total.

    \begin{figure}
        \centering
        \includegraphics[]{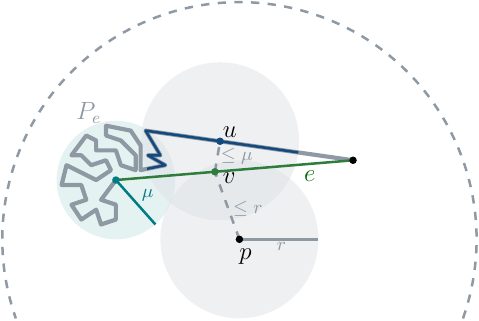}
        \caption{Illustrates the ball $B(p,r)$, the edge $e$ intersecting $B(p,r)$,  the vertex $v$ on $e$ within distance $r$ of $p$, the vertex $u$ on $P_e$ within distance $\mu$ of $v$ and the ball $B(u,r)$ containing at least $r$ edge length of $P_e$. The ball $B(u,r)$ is in the ball $B(p,3r)$.}
        \label{fig:PinsideDefiningBall}
    \end{figure}
\end{proof}

\begin{proof}[Proof of Theorem~\ref{theorem:SimplIsLowDense}]
We invoke \Cref{lemma:PinsideDefiningBall} to obtain a ball $B(p, 3r)$ where the total length of $P$ inside $B(p, 3r)$ is at least $\lambda' r$, for $r \geq \mu$. 
Now let $C$ be the cube of side length $6r$ that has $B(p, 3r)$ as its inscribed circle, then  $\alpha = \|C\cap P\| \ge \lambda'r$. 
We will now invoke \Cref{lemma:VerticiesIn3uC} with $C$, $\alpha$ and $u$: the longest edge of $P$ intersecting $C$.
The lemma yields that there must be at least $\Omega((\alpha/r)^{1+1/(d-1)}/\lambda^{1/(d-1)})$ vertices of $P$ inside the cube $C'$ with side length $3\max(6r,\|u\|)$.  Thus, there are $\Omega(\lambda'^{1+1/(d-1)}/\lambda^{1/(d-1)})$ vertices of $P$ inside $C'$. The total number of vertices of $P$ in $C'$ is less than $n$, and thus $\lambda' \in O(n^{1-1/d} \lambda^{1/d})$.
\end{proof}

\begin{lemma}
    For any $\mu>0$ and $\lambda \geq 1$, we can construct a size-$n$ $\lambda$-low-density
    curve $P$ whose $\mu$-simplification $\Psimp$ is only
    $\lambda'$-low dense for
    $\lambda' \in
        \Omega\left(n^{1-1/d}\lambda^{1/d}\right).$
\end{lemma}

\begin{proof}
    Let
    $
        m=\left(\frac{n}{\lambda}\right)^{1/d}
    $. Set $\gamma$ such that $(m\gamma)^2+\gamma^2=\mu^2$. Consider the $d$-dimensional integer grid scaled by a factor $\gamma$: $\mathcal{G}_{\gamma}=\left\{\gamma z\;\middle|\;z \in \mathbb{N}^d,\;\|z\|_\infty \leq m \right\}.$ The grid $\mathcal{G}_{\gamma}$ has $m^{d-1}$ rows parallel to the first coordinate axis of length 
    $m\gamma$. Along each row, consecutive grid points have distance
    $\gamma$. In the remaining coordinates, consecutive
    rows have distance 
    $\gamma$.

    We construct a curve $Q$ by traversing the first row
    from left to right, moving one step along the $(d-1)$-dimensional grid to an adjacent row, traversing that row from right
    to left, and so on, alternating the direction that we move to between consecutive rows.
    The resulting curve $Q$ visits every grid point exactly once and has
    $\Theta(m^d)$ vertices.
    Since $Q$ is a path on a $d$-dimensional grid, $Q$ is $2^d$-low dense. However, if we subdivide every edge of $Q$ into three edges of equal length then we obtain a path of $\Theta(m^d)$ vertices that is $O(1)$-low dense.

    To obtain the curve $P$, we simply traverse $Q$ $\lambda$ times, alternating between the
    forward and reverse traversal. The resulting curve $P$ has
    complexity
    $
        \Theta(\lambda m^d)=\Theta(n)
    $
    and is by construction $O(\lambda)$-low dense.

    It remains to consider its $\mu$-simplification.
    Let $a$ be the point at which $P$ starts one row, and let $b$
    be the point at which it starts the next row.
    The subcurve $P[a,b]$ first traverses that row over a distance $m\gamma$ and then moves a distance $\gamma$ to the next row. Hence, for every point $x$
    on the sub curve $P[a,b]$, $\lVert x-a \rVert \leq \sqrt{(m\gamma)^2+\gamma^2}\le \mu$. Thus $P[a,b]$ is contained in the ball of radius $\mu$ around $a$, and $b$ is the first point of $P$ after
    $a$ at distance $\ge\mu$ from $a$.
    The $\mu$-simplification of $P$ thus replaces each of the $m^{d-1}$ sub curves $P[a,b]$ by an edge of length $\mu$ (\Cref{fig:simplTight} Left). 
    Since the ball of radius $\mu$ centred at the centre
    of the grid, covers the whole of $P$, every simplification edge described
    above intersects this ball. Each such edge has length $\mu$.
    There are $m^{d-1}$ rows, and $P$ traverses each row of the grid $\lambda$ times. Hence this ball intersects
    $
        \Omega(\lambda m^{d-1})
    $
    edges of $\Psimp$ of length greater or equal than its radius, and:
    \[
        \lambda'
        \in
        \Omega(\lambda m^{d-1})
        =
        \Omega\left(
            \lambda
            \left(\frac n\lambda\right)^{(d-1)/d}
        \right)
        =
        \Omega\left(
            n^{1-1/d}\lambda^{1/d}
        \right). \qedhere
    \]
\end{proof}

    \begin{figure}
        \centering
        \includegraphics[]{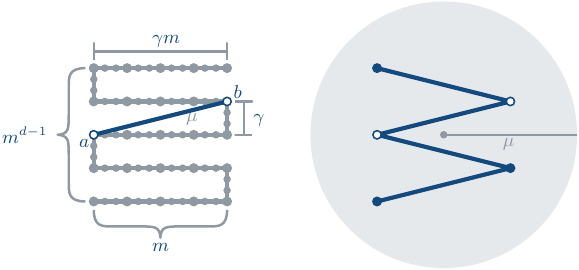}
        \caption{Left: The path $P$ for $d=2$. Each each edge is traversed $\lambda$ time and the path is $O(\lambda)$-low dense. Rigth: The path $\Psimp$.}
        \label{fig:simplTight}
    \end{figure}

\section{An imbalanced lower bound}
\label{app:imbalanced}

Let $P$ and $Q$ be two $d$-dimensional $\lambda$-low dense curves of complexity $n$ and $m$. 
An \emph{imbalanced} lower bound for computing the Fréchet distance between $P$ and $Q$ assumes that $n > m$.
Our imbalanced lower bound is slightly weaker than our balanced lower bound but proven in the exact same manner:

\fourUpHardness*
\begin{proof}
The construction is very similar to the one in
Theorem~\ref{thm:3up-approximation-hardness}.
We start with a balanced instance of the Orthogonal Vector Problem
with $K$ vectors in each set. Split the second set into groups of
$\lceil K^{3/4}\rceil$ vectors, except possibly for the last group.
The instance is a YES instance if and only if some vector
in the first set is orthogonal to a vector in one group.

Fix one group. We apply the planar construction from the proof of
Lemma~\ref{lem:2Dhardness} to the first set and this group.
Let $N$ and $M$ be the numbers of paths in $P$ and $Q$, respectively.
Since the vectors have $O(\log K)$ coordinates, we have
$N=O(K\log K)$ and $M=O(K^{3/4}\log K)$. Set
$
    a=\left\lceil d\left(\frac{N}{\lambda\eps}\right)^{1/(d-2)}\right\rceil$  and $w=\lceil\eps a\rceil$ and $C=20w$. 
At scale $C$, the curves $P$ and $Q$ have complexity $O(CN)$ and
$O(CM)$, respectively. By Lemma~\ref{lem:2Dhardness}, their Fréchet
distance is at most $1.1C$ in a YES instance and at least $2.9C$ in
a NO instance.

The construction in \Cref{lem:origin-path}, applied in $d-2$
dimensions, gives $\lambda$-paths $P^*$ in $\band^{d-2}(a,w)$ and
$Q^*$ in $-\band^{d-2}(a,w)$, each of length at least
\[
    \lambda w\left(\frac{a}{d-2}\right)^{d-3}
    \geq \frac{\lambda\eps a^{d-2}}{(d-2)^{d-3}}
    \geq N.
\]
We construct two $\lambda$-paths $\hat P$ and $\hat Q$ of complexity
$n=O(CN)$ and $m=O(CM)$ by the same procedure as in the proof of
Theorem~\ref{thm:3up-approximation-hardness}. By
Observation~\ref{obs:surface-distance}, the distance between a point
of $P^*$ and a point of $Q^*$ lies between $2a$ and $2a+4w-2$.
Thus, in a YES instance, the Fréchet distance between $\hat P$ and
$\hat Q$ is at most $2a+4w-2+1.1C\leq 2a+26w$. In a NO instance,
their Fréchet distance is at least $2a+2.9C=2a+58w$. Since
$\eps a\leq w$ and $\eps\leq1$, we have
\[
    (1+\eps)(2a+26w)\leq 2a+54w<2a+58w.
\]
Thus, a $(1+\eps)$-approximation distinguishes the two cases.

Suppose there exists a $(1+\eps)$-approximation algorithm with
running time
\[
    O\left(\left(
    \lambda^{2/(d-1)}
    \frac{(m/d)(n/d)^{1-2/(d-1)}}{\eps^{2(1-2/(d-1))}}
    \right)^{1-\delta}\right)
\]
for some $\delta>0$. If $w\geq2$, then $\eps a>1$, so
$w=\lceil\eps a\rceil\leq2\eps a$. Thus,
$C=O\bigl(\eps(N/(\lambda\eps))^{1/(d-2)}\bigr)$. Substituting
$n=O(CN)$ and $m=O(CM)$ gives
\[
    \lambda^{2/(d-1)}
    \frac{(m/d)(n/d)^{1-2/(d-1)}}{\eps^{2(1-2/(d-1))}}
    =O(NM).
\]
If $w=1$, then $C=20$, and the lower bound on $\eps$ in the
theorem bounds the same expression by $nm=O(NM)$. Thus, in both
cases the algorithm takes $O((NM)^{1-\delta})$ time per
group.

It remains to account for all $O(K^{1/4})$ groups. Since $d\geq4$,
$\lambda\geq1$, and $\eps\leq1$, we have $C=O(\sqrt N)$. For each
group, we construct only the first $N$ vertices of $P^*$ and the
first $M$ vertices of $Q^*$. Thus, constructing $\hat P$ and
$\hat Q$ takes $O(C(N+M))$ time per group and
$\widetilde O(K^{1/4}K^{3/2})=\widetilde O(K^{7/4})$ time in total.
Running the assumed algorithm on all groups takes
\[
    \widetilde O\left(K^{1/4}(K^{7/4})^{1-\delta}\right)
    =\widetilde O\left(K^{2-7\delta/4}\right)
\]
time, contradicting OVH.
\end{proof}

\section{Lower bound under the $\ell_2$ norm}

We now consider the Fréchet distance under the Euclidean norm. We show that, for any fixed dimension $d\geq 3$, the same dependence on $n$ and $\lambda$ as in our lower bound under the $\ell_1$ norm can be obtained under the $\ell_2$ norm.

Our previous construction relies on the fact that points on opposite $r$-diagonals are equidistant under the $\ell_1$ norm. We cannot easily establish an equivalent construction under the $\ell_2$ norm. Instead, we construct the origin paths inside a $(d-1)$-dimensional cube and scale the one-dimensional hardness construction sufficiently to preserve a constant approximation gap.

\begin{lemma}\label{lem:euclidean-origin-path}
    For any $d\geq 3$, $N\geq 1$ and $1\leq\lambda\leq N$, there exists a $(d-1)$-dimensional $\lambda$-path of complexity $N$ contained in the cube $[0,L]^{d-1}$, where $L=O((N/\lambda)^{1/(d-1)})$.
\end{lemma}

\begin{proof}
    Consider the subgraph of the $(d-1)$-dimensional unit grid graph induced by the cube $[0,L]^{d-1}$. This graph contains $(L+1)^{d-1}$ vertices. We can construct a  path through all these vertices by traversing the cube in a snake-like fashion. By traversing this path back and forth at most $\lambda$ times, we obtain a walk of complexity $\Theta(\lambda L^{d-1})$, where each vertex is visited at most $\lambda$ times and each edge is traversed at most $\lambda$ times.
    We choose $L=\Theta((N/\lambda)^{1/(d-1)})$ with a sufficiently large constant such that the resulting walk has complexity $N$.
\end{proof}

\begin{theorem}\label{thm:euclidean-hardness}
    Unless the Orthogonal Vector Hypothesis fails, for any fixed dimension $d\geq 3$ and any $\delta>0$, there is no algorithm that $2$-approximates the continuous or discrete Fréchet distance under the $\ell_2$ norm between two $d$-dimensional $\lambda$-paths of complexity $n$ in time
    \[
        O\left(\left(\lambda^{2/d}n^{2-2/d}\right)^{1-\delta}\right).
    \]
\end{theorem}

\begin{proof}
    Given a fixed dimension $d\geq 3$ and $1\leq\lambda$, we apply Lemma~\ref{lem:1Dhardness} to obtain two one-dimensional curves $P$ and $Q$, each consisting of $N$ $1$-paths of complexity $O(C)$, where we choose $L=\Theta((N/\lambda)^{1/(d-1)})$ as in Lemma~\ref{lem:euclidean-origin-path}, and set $C=\lceil 10\sqrt{d-1}\,L\rceil$. Their Fréchet distance is at most $1.1C$ in the YES instance and at least $2.9C$ in the NO instance.\\

    Next, we apply Lemma~\ref{lem:euclidean-origin-path} to obtain two $(d-1)$-dimensional $\lambda$-paths $P_*$ and $Q_*$, each of complexity $N$, contained in the cube $[0,L]^{d-1}$. We combine $P$ and $P_*$ to construct a $d$-dimensional $\lambda$-path $\hat P$ in the same manner as in the proof of Theorem~\ref{thm:exact-hardness}. We start at $(P_*[0],P[0])$ and traverse the first path of $P$ while keeping the position on $P_*$ fixed. Once the path is completed, we advance one vertex along $P_*$ while keeping the position on $P$ fixed, and continue with the next path of $P$. Repeating this procedure yields $\hat P$. We construct $\hat Q$ analogously.
    Since $P$ and $Q$ consist of $N$ one dimensional $1$-paths of complexity $O(C)$, the resulting curves $\hat P$ and $\hat Q$ have complexity $n=O(NC)$. Moreover, $\hat P$ and $\hat Q$ are $O(\lambda)$-paths: each of the $N$ $1$-paths of $P$ and $Q$ traverses any grid edge at most once, each origin vertex is visited at most $\lambda$ times, and $P_*$ and $Q_*$ traversed each vertex in $[0,L]^{d-1}$ at most $\lambda$ times.\\

    We now bound the Euclidean Fréchet distance between $\hat P$ and $\hat Q$.
Consider any pair of points $\hat p=(u,p)\in\hat P$ and $\hat q=(v,q)\in\hat Q$, where $u,v\in[0,L]^{d-1}$ and $p,q\in\mathbb{R}$. Since $u$ and $v$ lie in the same cube, we have $\|u-v\|_2\leq\sqrt{d-1}\,L$. Thus,
\[
    |p-q|
    \leq \|\hat p-\hat q\|_2
    = \sqrt{\|u-v\|_2^2+|p-q|^2}
    \leq \sqrt{(d-1)L^2+|p-q|^2}.
\]
We use these pointwise bounds to distinguish the YES and NO instances.

In the YES instance, $\fd(P,Q)\leq 1.1C$, so there exists a traversal of $P$ and $Q$ in which the distance between the matched points is at most $1.1C$ throughout. Since the projections of $\hat P$ and $\hat Q$ onto the last coordinate are precisely $P$ and $Q$, respectively, we can lift this traversal to $\hat P$ and $\hat Q$ by using the same parametrizations. At every pair of matched points $\hat p=(u,p)$ and $\hat q=(v,q)$ in this lifted traversal, we have $|p-q|\leq 1.1C$. By the pointwise bound above and our choice of $C$, which ensures $(d-1)L^2\leq C^2/100$, we obtain
$
    \|\hat p-\hat q\|_2
    \leq \sqrt{(1.1C)^2+(d-1)L^2}
    \leq \sqrt{1.22}\,C
    < 1.105C.
$
Hence, this lifted traversal witnesses that $\fd(\hat P,\hat Q)<1.105C$.

In the NO instance, $\fd(P,Q)\geq 2.9C$. Consider any traversal of $\hat P$ and $\hat Q$. Projecting the matched points onto their last coordinates gives a valid traversal of $P$ and $Q$, using the same parametrizations. Since $|p-q|\leq\|\hat p-\hat q\|_2$ for every pair of matched points, the maximum distance along the projected traversal is no greater than that along the original traversal. As every traversal of $P$ and $Q$ has maximum distance at least $2.9C$, the same must hold for every traversal of $\hat P$ and $\hat Q$. Therefore,
$
    \fd(\hat P,\hat Q)\geq\fd(P,Q)\geq 2.9C.
$

Consequently, we have $\fd(\hat P,\hat Q)<1.105C$ in the YES instance and $\fd(\hat P,\hat Q)\geq 2.9C$ in the NO instance. Since $2\cdot 1.105C<2.9C$, a $2$-approximation algorithm can distinguish between the two cases.

We now express the complexity of the construction in terms of $N$ and $\lambda$. Recall that $C=O((N/\lambda)^{1/(d-1)})$. Hence,
    $
            n
            =O(NC)
            =O\left(N\left(\frac{N}{\lambda}\right)^{1/(d-1)}\right)
            =O\left(N^{d/(d-1)}\lambda^{-1/(d-1)}\right).
    $
    Suppose, for the sake of contradiction, that there exists an algorithm that $2$-approximates the Euclidean Fréchet distance in time $O((\lambda^{2/d}n^{2-2/d})^{1-\delta})$ for some $\delta>0$. This algorithm would distinguish the two cases in time $O(N^{2-2\delta})$. This contradicts OVH.
\end{proof}

\section{Supplementary Arithmetics}
In this appendix, we supply the additional arithmetic needed to support the claim that $a \eps > 1$ from the proof of Theorem~\ref{thm:3up-approximation-hardness}.
The math for this claim in Theorem~\ref{thm:4up-approximation-hardness} is analogous. 

\subsection{\texorpdfstring{Details for the $\eps$-bound in the proof of Theorem~\ref{thm:3up-approximation-hardness}}{Details for the epsilon-bound in the proof of Theorem \ref{thm:3up-approximation-hardness}}}

\begin{lemma} \label{lem:lowerboundOnEpsilon}
    Consider variables $\eps$, $\lambda$, $N$, $d$ and define
    \begin{align}
        a& := d \left(\frac{N}{\lambda \eps}\right)^{1/(d-1)}\label{eq:a}\\
        n& :=N^{1+1/(d-1)}\eps^{1-1/(d-1)}\lambda^{-1/(d-1)}d\label{eq:n}
    \end{align}
    Moreover, assume that 
    \begin{align}    
        \varepsilon &> d^{-\frac{d-1}{d-2}} \, \left(\frac{\lambda}{n}\right)^{\frac{1}{d-2}}.\label{eq:eps}
    \end{align}   

    Then it follows that $a \varepsilon>1$.Note that equations \ref{eq:a},\ref{eq:n}, and \ref{eq:eps} define the same setting as in the proof of Theorem~\ref{thm:3up-approximation-hardness}.
\end{lemma}

\begin{proof}
    
    Combining Equation~\eqref{eq:a} with the inequality $a\varepsilon>1$ gives

    \begin{align*}
        d\left(\frac{N}{\lambda \varepsilon}\right)^{\frac{1}{d-1}} \varepsilon &> 1\\
        d \, N^{\frac{1}{d-1}} \lambda^{-\frac{1}{d-1}} \varepsilon^{\,1 - \frac{1}{d-1}} &> 1\\
        d \, N^{\frac{1}{d-1}} \lambda^{-\frac{1}{d-1}} \varepsilon^{\frac{d-2}{d-1}} &> 1 \\
        \varepsilon^{\frac{d-2}{d-1}} &> \frac{1}{d}\,\lambda^{\frac{1}{d-1}} N^{-\frac{1}{d-1}}
    \end{align*}
    which is equivalent to
    \begin{equation}
                \varepsilon > \frac{1}{d}^{\frac{d-1}{d-2}}\,\left(\frac{\lambda}{N}\right)^{\frac{1}{d-2}} \label{eq:epsboundN}.
    \end{equation}

We now solve Equation~\eqref{eq:n} for
$\left(\frac{\lambda}{N}\right)^{\frac{1}{d-2}}$:

\begin{align*}
    n =& d \, N^{1+\frac{1}{d-1}} \, \varepsilon^{1-\frac{1}{d-1}} \, \lambda^{-\frac{1}{d-1}}\\
    n =& d \, N^{\frac{d}{d-1}} \, \varepsilon^{\frac{d-2}{d-1}} \, \lambda^{-\frac{1}{d-1}}\\
    N^{\frac{d}{d-1}} =& \frac{n}{d} \, \varepsilon^{-\frac{d-2}{d-1}} \, \lambda^{\frac{1}{d-1}} \\
N =& \left(\frac{n}{d}\right)^{\frac{d-1}{d}} \, \varepsilon^{-\frac{d-2}{d}} \, \lambda^{\frac{1}{d}}\\
\frac{\lambda}{N} =& \lambda^{\frac{d-1}{d}} \left(\frac{n}{d}\right)^{-\frac{d-1}{d}} \varepsilon^{\frac{d-2}{d}}\\
\left(\frac{\lambda}{N}\right)^{\frac{1}{d-2}}
=& \lambda^{\frac{d-1}{d(d-2)}} 
\left(\frac{n}{d}\right)^{-\frac{d-1}{d(d-2)}} 
\varepsilon^{\frac{1}{d}}.
\end{align*}

Substituting this expression into Equation~\eqref{eq:epsboundN} gives

\begin{align*}
\varepsilon & > d^{-\frac{d-1}{d-2}} \,
\lambda^{\frac{d-1}{d(d-2)}} \,
\left(\frac{n}{d}\right)^{-\frac{d-1}{d(d-2)}} \,
\varepsilon^{\frac{1}{d}}\\
\varepsilon^{\frac{d-1}{d}} &>
d^{-\frac{d-1}{d-2}} \,
\lambda^{\frac{d-1}{d(d-2)}} \,
\left(\frac{n}{d}\right)^{-\frac{d-1}{d(d-2)}}\\
\varepsilon &>
d^{-\frac{d-1}{d-2}} \,
\left(\frac{\lambda}{n}\right)^{\frac{1}{d-2}}.
\end{align*}

which is precisely the desired assumption.

\end{proof}

\section{\texorpdfstring{Revisiting previous result for constant low dense curves}{Revisiting previous result for constant low dense curves}} 
\label{app:driemel_details}
{
\newcommand{\Cube}{C}
\newcommand{\curveA}{\pi}
\newcommand{\curveB}{\sigma}
\newcommand{\lenX}[1]{\left\|{#1} \right\|}
\newcommand{\cDim}{c_{d}}
\newcommand{\Holder}{H\"older\xspace}
\newcommand{\VolumeX}[1]{\mathrm{vol}\pth{#1}}
\newcommand{\constA}{c_1}
\newcommand{\constB}{c_2}
\newcommand{\constC}{c_3}
\providecommand{\pth}[2][\!]{#1\left({#2}\right)}
\providecommand{\cardin}[1]{\lvert#1\rvert}
\newcommand{\curveAs}{\pi'}
\newcommand{\curveBs}{\sigma'}
\newcommand{\Cell}{\mathsf{R}}
\newcommand{\SimpComplexityC}{\mathsf{N}}
\newcommand{\SimpComplexity}[3]{\SimpComplexityC\pth{#1, #2, #3}}
\newcommand{\simpX}[1]{\mathrm{s{i}m{p}l}\pth{#1}}
\newcommand{\sRadius}{\mu}
\newcommand{\Nleq}[1]{\NleqC_{\leq #1}}
\newcommand{\NleqC}{N}
\newcommand{\segA}{\mathsf{u}}
\newcommand{\segB}{\mathsf{v}}
\newcommand{\pntA}{p}
\newcommand{\FDleq}[1]{\FDleqC_{\leq #1}}
\newcommand{\Links}{\mathsf{d}}
\newcommand{\nVCell}[2]{\cardin{{#1} \sqcap{#2}}\,}
\providecommand{\nfrac}[2]{#1/#2}
\newcommand{\FDleqC}{\EuScript{R}}
\newcommand{\Grid}{\EuScript{G}}

The original Lemma 4.9 presented in \cite{DriemelHW12} only works for $\phi$-low dense paths where $\phi$ is constant. The following write up is a copy of Lemma 4.9 in \cite{DriemelHW12}, where two difference sets of edits have been injected. The reader is only meant to read one version at a time. 
 \colorbox{BurntOrange}{The orange enhances the original, specifying which steps assume or use that $\phi$ is constant.} \colorbox{Orchid}{The purple introduces an assumption which allows the statements to work for arbitrary $\phi$.}

\begin{lemma}\label{lemma:revisit4.9}[Lemma 4.9 in \cite{DriemelHW12}]
    Let $\curveA$ be a $\phi$-low density curve in $\bR^d$ , \colorbox{BurntOrange}{where $\phi$ is a constant}, and let
    $\Cube$ be a cube in $\bR^d$ with side length $r$. Let $\alpha =
    \lenX{\curveA \cap \Cube}$.  
    \colorbox{Orchid}{Assume $\alpha \ge c\phi r$, for a } \colorbox{Orchid}{sufficiently large constant $c$}.
    There must be at least
    $\Omega((\alpha/ (r\colorbox{Orchid}{$\phi$}))^{1+1/(d-1)})$ vertices of $\curveA$ contained
    in $3\Cube$, where $3 \Cube$ is the scaling of $\Cube$ by a factor
    of $3$ around its center.
    
\end{lemma}

\begin{proof}
    We will first give a lower bound on the number $n$ of edges
    intersecting $\Cube$ (i.e., the edges that contribute to
    $\alpha$). Then we will account for the edges that have endpoints
    outside $3 \Cube$. So, take the $n$ edges of $\curveA$ that
    intersect $\Cube$ and construct the cover of $\Cube$ resulting
    from Lemma 4.8 (in \cite{DriemelHW12})  with respect to these edges.
    
    Let $\Cube_1, \ldots, \Cube_N$ denote the cubes in this cover,
    where $r_1 \leq r_2 \leq \dots \leq r_{N}$ are the side lengths of
    the cubes used by the cover,
    respectively. Lemma 4.8 (in \cite{DriemelHW12}) implies that $N \leq
    2^{d+1} d n$, and, therefore, a lower bound on $N$ would provide a
    lower bound on $n$.
    
    So, the sum of the diameters of those $N$ cubes bounds the length
    of the intersection $\alpha \leq \sum_{i=1}^{N} \cDim\phi \sqrt{d}
    r_i$, since every cube in this cover can intersect at most $\cDim
    \phi$ edges of $\curveA$.  Setting $p=d$ and $q = d/(d-1)$, we
    observe that $1/p + 1/q = 1/d + (d-1)/d = 1$, and by \Holder's
    inequality\footnote{\Holder's inequality states that $\sum_{i=1}^n
       \cardin{a_i b_i} \leq \pth{\sum_{i=1}^n
          \cardin{a_i}^{q}}^{1/q}\pth{\sum_{i=1}^n
          \cardin{b_i}^{p}}^{1/p}$ if $1/p+1/q = 1$.}, we have that
    \begin{align*}
        \sum_{i=1}^{N} r_i = \sum_{i=1}^{N} r_i \cdot 1%
        \leq%
        \pth[]{\sum_{i=1}^N r_i^d }^{1/d} \pth{ \sum_{i=1}^N
           1^q}^{1/q} %
        = %
        \pth[]{\sum_{i=1}^N r_i^d }^{1/d} N^{(d-1)/d}.
    \end{align*}
    
    Lemma 4.8 also implies that the sum of the
    volumes of the cubes is at most $2^d \VolumeX{\Cube}$, since every
    point in $\Cube$ is covered at most $2^d$ times by this
    cover. Therefore we have that $\sum_{i=1}^{N} r_i^d = \sum_{i=1}^N
    \VolumeX{\Cube_i} \leq 2^d \VolumeX{\Cube} = 2^d r^d$. Hence
    \begin{align*}
        \alpha %
        \leq%
        \sum_{i=1}^{N} \cDim \phi \sqrt{d} r_i%
        \leq%
        \cDim\phi \sqrt{d} \pth[]{\sum_{i=1}^N r_i^d }^{1/d}
        N^{(d-1)/d}%
        \leq%
        \cDim\phi \sqrt{d} \pth{2^d r^d }^{1/d} N^{(d-1)/d}.
    \end{align*}
    \colorbox{Orchid}{\parbox{\textwidth}{
    First, consider edges that intersect $\Cube$ whose endpoints lie outside of $3 \Cube$: these edges must be longer than the side-length of $\Cube$ and by Claim 4.7 (in~\cite{DriemelHW12}), there are at most $\cDim\phi$ such edges. Thus, their total contribution to $\alpha$ is at most $\cDim\phi\sqrt{d}r$. As we assume $\alpha \ge c\phi r$ for a sufficiently large constant $c$, we can simply ignore the contribution of these long edges to $\alpha$ and assume henceforth that all edges that intersect $\Cube$ have an endpoint in $3 \Cube$.
    }}
    
    This implies that $\constB\pth{ {\alpha}/{r}}^{d/(d-1)} \leq
    N$, where $\constB = \pth{{ 2 \cDim \phi \sqrt{d}}
    }^{-d/(d-1)}$. Since $N \leq 2^{2d+1}  n$, we have that $
    \constC\pth{ {\alpha}/{r}}^{d/(d-1)} \leq n$, where $ \constC =
    \frac{1}{2^{2d+1} }\pth{{ 2 \cDim\phi \sqrt{d}} }^{-d/(d-1)}$.
    
    \noindent
    \colorbox{Orchid}{\parbox{\textwidth}{Since we may assume that every edge counted above has a vertex inside $3C$, the total number of verticies inside $3\Cube$ is at least $\Omega((\alpha/ (r\phi)))^{1+1/(d-1)})$. Note that the $\phi^{-(1+1/(d-1))}$ simply is carried over from $c_3$. This concludes  (the purple version of) the proof.}}
   
     \colorbox{BurntOrange}{We want to use the above to  bound the number of endpoints that lie in $3\Cube$, however,}
    some of these $n$ edges intersecting $\Cube$ can have both
    endpoints outside $3\Cube$. Such edges are longer than the
    sidelength of $\Cube$ and by Claim 4.7 their number
    is bounded by $\cDim\phi$.
    
    Hence, the number of vertices of $\curveA$ inside $3\Cube$ is at
    least $n - \cDim\phi \geq \constC \pth{\alpha / r}^{d/(d-1)} -
    \cDim\phi$.
    
    \noindent
    \colorbox{BurntOrange}{\parbox{\textwidth}{
    As the function $\cDim\phi$ grows when $\phi$ increases and $\constC \pth{\alpha / r}^{d/(d-1)}$ decrees when $\phi$ increases, we can only bound the number of verticies inside $3\Cube$ for constant $\phi$. Thus, for constant $\phi$ there are at least $\Omega\left((\alpha/r)^{1+1/(d-1)}\right)$ vertices of $\pi$ contained in $3C$.}}
\end{proof}

The purple edits can then be injected into the original version of Lemma 4.12:

\begin{lemma}
    For any two low density curves $\curveA$ and $\curveB$ in $\Re^d$,
    and $0 < \eps < 1$, we have that
    $\SimpComplexity{\eps}{\curveA}{\curveB} \in O \pth{
       \colorbox{Orchid}{$\phi^2$} \frac{ n^{2(d-1)/d}}{\eps^2}}$.    
\end{lemma}

\begin{proof}
    Let $\delta \geq 0$ be an arbitrary radius, and let $\curveAs=
    \simpX{\curveA,\sRadius}$ and $\curveBs =
    \simpX{\curveB,\sRadius}$ be their simplifications, where
    $\sRadius = \eps \delta$. Then, we need to prove that
    $\Nleq{\delta}(\curveAs, \curveBs) \in O \pth{
       \colorbox{Orchid}{$\phi^2$} \frac{n^{2(d-1)/d}}{\eps^2}}$.
    
    To this end, it suffices to bound the number of vertex-edge pairs
    $(\pntA,\segA)$, where $\pntA$ is a vertex of $\curveAs$, $\segA$
    is an edge of $\curveBs$, and the distance between $\pntA$ and
    $\segA$ is at most $\delta$ (naturally, we need to apply the same
    argument to pairs with vertices in $\curveBs$ and edges in
    $\curveAs$). The total number of such pairs bounds the total
    complexity of $\FDleq{\delta} = \FDleq{\delta}(\curveAs,
    \curveBs)$.
    
    Set $M \in O \pth{ \nfrac{\colorbox{Orchid}{$\phi^2$}n^{1-2/d}}{\eps^2}}$, and associate every
    vertex-edge pair $(\pntA,\segA)$ that appears in the free space
    diagram $\FDleq{\delta}$ with the vertex $\pntA$.
    
    
    Consider the grid $\Grid$ of side length $\delta$.  For a grid
    cell $\Cell$, consider the vertex of $\curveAs$ in $\Cell$ that is
    associated with the largest number of such vertex-edge pairs, and
    say it is being associated with $\Links_\Cell$ such vertex-edge
    pairs, and let $v_\Cell$ denote this ``popular'' vertex of
    $\curveAs$. The total number of vertex-edge pairs associated with
    vertices of $\curveAs$ inside $\Cell$ is bounded by $U_{\Cell} =
    \nVCell{\curveAs}{\Cell} \Links_\Cell$, where
    $\nVCell{\curveAs}{\Cell}$ denotes the number of vertices of
    $\curveAs$ that lie inside $\Cell$.
    
    If $\Links_\Cell \leq M$ then $U_\Cell \leq
    \nVCell{\curveAs}{\Cell} M$, and we charge $M$ units to each
    vertex of $\curveA$ inside $\Cell$.
    
    If $\Links_\Cell > M$ then the length of $\curveBs$ inside
    $\Cube/3$ is at least $\Links_\Cell \sRadius$, where $\Cube$ is a
    cube centered at $\Cell$ with side length $O(\delta)$.  Indeed,
    all the charges $\Links_\Cell$ rise from different segments of
    $\curveBs$ that are in distance at most $\delta$ from $v_\Cell$,
    and each such segment has length at least $\sRadius$.
    
    \noindent
    \colorbox{Orchid}{\parbox{\textwidth}{
    Thus, the length of $\curveB$ inside of $C/3$ is at least $\Links_\Cell \sRadius> M\mu = \left(\frac{\phi^{2/d} n^{1-2/d}}{\varepsilon^2}\right) \delta \varepsilon \ge \phi \delta \varepsilon^{-1}$. This means that for any cube $C/3$ of radius $O(\delta)$ and for a small enough $\varepsilon$ the length of $\curveB$ inside of $C/3$ is at least $\phi \delta \varepsilon^{-1}$ which is the assumption needed to invoke \Cref{lemma:revisit4.9}}.}
    
    By Lemma 4.9 (in~\cite{DriemelHW12}, here \Cref{lemma:revisit4.9}),  we have that $\curveB$ must have
    at least $\Omega\pth{ ( \Links_\Cell \sRadius / (\delta\colorbox{Orchid}{$\phi$}))^{d/(d-1)}
    }= \Omega\pth{ ( \Links_\Cell \eps/\colorbox{Orchid}{$\phi$})^{d/(d-1)} }$ vertices inside
    $\Cube$. There is some constant $c$ such that
    \begin{align*}
        &c \pth{ \eps \Links_\Cell/\colorbox{Orchid}{$\phi$} }^{d/(d-1)}%
        \leq %
        \nVCell{\curveB}{\Cube}
        \;\implies\;%
        \Links_\Cell%
        \leq%
        \frac{\colorbox{Orchid}{$\phi$}}{\eps} \pth{\frac{\nVCell{\curveB}{\Cube}}{c}
        }^{(d-1)/d} %
        \;\implies\;%
        \Links_\Cell^2 \leq \frac{\colorbox{Orchid}{$\phi^2$}}{\eps^2}
        \pth{\frac{\nVCell{\curveB}{\Cube}}{c} }^{2-2/d}%
        \\
        &\;\implies\;%
        \Links_\Cell^2 \leq%
        \nVCell{\curveB}{\Cube} {\frac{\colorbox{Orchid}{$\phi^2$}}{c\eps^2}
           \pth{\frac{\nVCell{\curveB}{\Cube}}{c} }^{1-2/d}}
        \leq%
        {\frac{\colorbox{Orchid}{$\phi^2$}}{c \eps^2} \pth{\frac{n}{c} }^{1-2/d}}
        \nVCell{\curveB}{\Cube} %
        \leq%
        M \nVCell{\curveB}{\Cube},
    \end{align*}
    by picking $M$ to be sufficiently large.  In particular, if
    $\nVCell{\curveAs}{\Cell} \leq \Links_\Cell$, then $U_\Cell =
    \nVCell{\curveAs}{\Cell} \Links_\Cell\leq \Links_\Cell^2 \leq M
    \nVCell{\curveB}{\Cube}$. Hence, we charge $M$ units to each
    vertex of $\curveB$ inside the cube $\Cube$.
    
    Otherwise, $\nVCell{\curveAs}{\Cell} > \Links_\Cell > M$. But
    then, the length of $\curveAs$ inside $\Cube$ is at least
    $\nVCell{\curveAs}{\Cell} \sRadius$, and by
    Lemma 4.9 (in~\cite{DriemelHW12}, here \Cref{lemma:revisit4.9}), we have that $\curveA$ must have at
    least $\Omega\pth{ ( \nVCell{\curveAs}{\Cell} \eps/\colorbox{Orchid}{$\phi$})^{d/(d-1)} }$
    vertices inside $\Cube$.  Arguing as above, this implies that
    $\nVCell{\curveAs}{\Cell}^2 \leq M \nVCell{\curveA}{\Cube}$. As
    such, we have that $U_\Cell = \nVCell{\curveAs}{\Cell}
    \Links_\Cell\leq \nVCell{\curveAs}{\Cell}^2 \leq M
    \nVCell{\curveA}{\Cube}$.  Again, we charge $M$ units to each
    vertex of $\curveA$ inside the cube $\Cube$.
    
    Since $\Cube$ intersects a constant number of cells of the grid,
    no vertex would get charged more than a constant number of times
    by the above scheme.  Thus, every vertex, of either curve, gets
    charged $O(M)$ units overall, and the total number of vertex-edge
    pairs present in $\FDleq{\delta}$ is $O( n M)$, as claimed.
\end{proof}

By plugging the above into \Cref{theorem:ApproxAlgDriemel}, we get the following result.

\begin{theorem}
    Given two low-density curves $\curveA$ and $\curveB$ with a total
    of $n$ vertices in $\bR^d$, and a parameter $\eps > 0$, there
    exists an algorithm which $(1+\eps)$-approximates the Fréchet
    distance between $\curveA$ and $\curveB$ in $O \pth{
       \frac{\colorbox{Orchid}{$\phi^2$}n^{2(d-1)/d}}{\eps^2} + \colorbox{Orchid}{$\phi^2$}n^{2(d-1)/d} \log n }$ time.    
\end{theorem}

}

\section{Map Matching} \label{app:mapmatching}
Map matching is the problem of determining whether a trajectory matches a traversal in an embedded graph. The trajectory is typically obtained from noisy measurements, such as GPS recordings, while the graph represents an underlying transportation network, such as a road or railway network. The objective is to identify the walk in the graph that best matches the observed trajectory.
Crucially, we define a \emph{walk} as any sequence of vertices of a graph $G$ where incident vertices in the sequence are connected. 

\begin{definition}
    Formally, let $G=(V,E)$ be a connected embedded graph with $|E|$ edges in $\mathbb{R}^d$ and let $P$ be a curve with $n$ vertices in $\mathbb{R}^d$. \emph{The map-matching problem} asks whether there exists a walk $Q$ in $G$ whose Fréchet distance to $P$ is at most a given threshold $\Delta$. Here $\walk$ denotes the set of all walks in $G$.\par
In the optimization variant, we $(1+\eps)$-approximate the value $\mapFD GPQ$, by returning a value that is at most a factor $(1+\eps)$ away from $\mapFD GPQ$.
\end{definition}

\subparagraph{Previous work} The Map Matching problem has been previously studied under the assumption that the graph $G$ is $\lambda$-low-density, and the trajectory $P$ is $c$-packed~\cite{Chen2011approximate}. They show for the decision variant a $(1 + \eps)$-approximate decider which takes near-linear time. This yields for the optimization variant a near-linear time algorithm by invoking the techniques from~\cite{DriemelHW12}.
We present an 
$\tilde{O}\left( \frac{\lambda^{2/d}(n+|E|)^{2-2/d}}{\eps^2} \right)$
time $(1+\varepsilon)$-approximate decider, under the assumption that both the embedded graph and the query curve are $\lambda$-low dense. Furthermore, we present a near-tight lower bound \ref{thm:MapMatching3up-approximation-hardness} in the same setting, by using the same general framework for embedding orthogonal vector gadgets into low-density grid graphs.

\subsection{Concurrent or competing work}
\label{app:concurrent}

The first version of this paper that claims these $\lambda$-low density map matching bounds appeared on ArXiV on 23 September 2026. 
Very recently (Tuesday, 24 November 2026) Deryckere, Gudmundsson, Renssen and Wong \cite{DeryckereGRW26}
presented an ArXiV preprint with the following results: 
first, they consider the map matching problem where the graph $G$ is arbitrary but the input curve $P$ is $c$-packed. 
They present a $(1+\eps)$-approximation for the map matching problem which takes $\widetilde{O}(\eps n +$ time. 
Secondly, they consider (just as in this paper) the case where both $G$ and $P$ are $\lambda$-low density. They present an $(1+\eps)$-approximation algorithm for map matching with \emph{expected} runtime $\tilde{O}\left( \frac{\lambda^2(n|V|)^{1-1/d}}{\eps^2}
        +\frac{\lambda^2(n+|V|)}{\eps^d}\right)$.

\subparagraph{Comparison between the results.}
We compare the two approximation algorithms for fixed $d$ and $\eps$,
as in Theorem~32 of~\cite{DeryckereGRW26}.
When $n$, $|V|$, and $|E|$ are of the same order, their expected
running time is $\widetilde O(\lambda^2 n^{2-2/d})$, whereas our
running time is $\widetilde O(\lambda^{2/d}n^{2-2/d})$.
Thus, in this case, our algorithm has the same dependence on $n$
and improves the dependence on $\lambda$ from $\lambda^2$ to
$\lambda^{2/d}$. Our algorithm is also deterministic and works
for any $\ell_p$ norm.

We consider our approximate decider to be simpler as for our approach, it suffices to simplify the curve and the graph and explore
the reachable free space in a brute-force manner. We find our analysis simpler too, as we bound the number of
relevant free-space cells via a direct charging argument.

\subsection{Preliminaries for map matching}

\subparagraph{$\lambda$-low dense input} The input consists of a connected undirected embedded graph $G=(V,E)$ with $|E|$ edges and a curve $P$ with $n$ vertices in $\mathbb R^d$, both $\lambda$-low dense under the chosen $\ell_p$ norm. Throughout the analysis, $d\geq 2$ is fixed. Since $G$ is connected, $|V|\leq |E|+1$, so the combined input size is $\Theta(n+|E|)$. We assume that the intersection of a line segment with an $\ell_p$ ball
can be computed in constant time.

We interpret both the edge set of the embedded graph and the set of segments of the polygonal curve as collections of line segments. Consequently, an embedded graph is $\lambda$-low dense if its edge set is $\lambda$-low dense, and a polygonal curve is $\lambda$-low dense if its constituent segments form a $\lambda$-low-density collection of edges.

\subparagraph{Free-Space Surface} To solve the decision variant of the map-matching problem, Alt, Efrat, Rote and Wenk~\cite{ALT2003262} introduce the \emph{free-space surface} between a curve $P$ and a graph $G$. As we will not be using it directly, we only give a high-level overview of this mathematical object. Let $P=(p_1,\ldots,p_n)$ be a polygonal curve and let $G=(V,E)$ be a connected graph with $|E|$ edges embedded in $\mathbb{R}^d$. For each edge $e\in E$, we construct the free-space diagram between $P$ and $e$ $\FSD(P,e)$. By connecting the diagrams of adjacent graph edges according to the connectivity of $G$, these diagrams together form the free-space surface $\FSS(P,G)$.

Crucially for this paper, each cell of the free-space surface corresponds to a segment of $P$ and an edge of $G$. A point in this cell is free if the corresponding points on the segment and edge are within distance $\Delta$ in the $\ell_p$ norm.

\subsection{An approximate decider}

The approximate decider from~\cite{Chen2011approximate} follows the same blueprint as approximate deciders for the Fréchet distance.
First, it shows that there is an exact decider whose running time can be upper bounded by the complexity of the relevant free space in the $\Delta$-FSS.
Then, they invoke this exact decider with a simplified query curve to bound the complexity of the reachable free-space between the graph and the simplified curve. The fact that they invoke the exact decider with a simplified curve then yields a $(1+\eps)$-approximate decider. 

\begin{definition}
Following~\cite[Section~3.1]{Chen2011approximate}, a point in the free space is \emph{relevant} if it can be reached through the free space from the initial point of $P$ paired with a vertex of $G$. This path need not be monotone in $P$.
\end{definition}

\begin{definition}
For a curve $P$ with $n$ vertices, a connected graph $G$ with $|E|$ edges and $\Delta\geq 0$, let $\graphFSDelta(P,G)$ denote the number of cells intersecting the relevant free space of $\FSS(P,G)$.
\end{definition}

\begin{theorem}[The exact decider, implied by Theorem 3.1 in \cite{Chen2011approximate}]\label{thm:MapmatchingDecider}
Let $P$ be a polygonal curve of complexity $n$ and $G$ be a connected graph with $|E|$ edges. Then there is an algorithm $\Decider(P,G,\Delta)$ that decides $\mapFD GPQ \le \Delta$ in $O(n+|E|+\graphFSDelta(P,G)\log(n+|E|))$ time and $O(n+|E|+\graphFSDelta(P,G))$ space.
\end{theorem}

\subparagraph{Our approach.}
The solution in~\cite{Chen2011approximate} simplified only the query curve, and bounded the complexity of the reachable free-space between $G$ and the simplified curve. 
On a moral level, this is why they \emph{must} require the query curve to be $c$-packed as they essentially make use of the one-sided Fréchet distance free-space complexity bounds.
Our approach, rather simplifies both $G$ and the query curve $P$ and then bounds the reachable free-space complexity. To this end, we must explain how we simplify a graph:

\subsubsection{\texorpdfstring{$\mu$-simplifications of $\lambda$-low-density graphs}{Mu-Simplifications of lambda-low-density graphs}}
Let $G=(V,E)$ be a connected graph with $|E|$ edges and a straight-line embedding in $\mathbb R^d$. Fix $\mu>0$.
We define a $\mu$-simplification through the following iterative algorithm (see Figure~\ref{fig:simpG}):

\begin{definition}
Initially, all vertices are unprocessed. Visit the vertices in input
order, skipping those already processed. For each remaining vertex $u$,
let $B_r$ be the cube centred at $u$ with half-side length $r$.
Consider the original edges of length less than $\mu$ whose endpoints
are both unprocessed and lie in $B_{3\mu}$. Let $L$ be their total
length. Starting at $r=\mu$, grow $B_r$ until at most $L/\mu$ of these
edges have exactly one endpoint in $B_r$.

Contract each connected component induced by the unprocessed vertices
in $B_r$ to a separate vertex at $u$. Mark these original vertices as
processed and repeat until all vertices are processed. Two vertices of
$\Gsimp$ are adjacent if an original edge joins their components.
\end{definition}

\subparagraph{Constructing $\Gsimp$.}
In the above definition, the cube centred at $u$ stops growing before $r$ exceeds $2\mu$. Indeed, for each edge $vz$
contributing to $L$, the radii for which $B_r$ contains exactly one
endpoint form an interval of length
$
    \bigl|\|v-u\|_\infty-\|z-u\|_\infty\bigr|\leq \|v-z\|_p$.
Summing over these edges, the average number with exactly one endpoint
in $B_r$ over $r\in[\mu,2\mu]$ is at most $L/\mu$. Thus, there exists an $r\in[\mu,2\mu]$ for which at most $L/\mu$
of the edges contributing to $L$ have exactly one endpoint in $B_r$. We use this fact to construct $\Gsimp$ as follows:

\begin{itemize}
    \item First, we the vertices of $G$ by disjoint cubes of side length
    $\mu$.
    \item For each centre $u$, find the unprocessed vertices in
    $B_{3\mu}$ by finding the $O(1)$ disjoint cubes that intersect $B_{3\mu}$.
    \item Compute $L$ from the incident edges of all unprocessed vertices in $B_{3\mu}$. From $r=\mu$,
    grow $B_r$ until at most $L/\mu$ of the edges contributing to $L$
    have exactly one endpoint in $B_r$.
    \item Contract each connected component induced by the unprocessed
    vertices in $B_r$ to a separate vertex at $u$. Mark these
    vertices as processed and repeat with the next centre.
    \item Once all vertices have been processed, construct the edges
    of $\Gsimp$ from the original edges.
\end{itemize}

Our first step is to make disjoint squares of size-length $\mu$ that cover all vertices of $G$.
Under most models of computation, a grid of side-length $\mu$ would suffice. For Real-RAM purists, the following solution will yield the same:
In each coordinate, sort the vertex coordinates and cover them from left
to right by intervals of length $\mu$, starting each interval at the
first uncovered value. The products of these intervals form cells of
side length $\mu$. Since $d$ is fixed, we find the cell containing each
vertex in $O(\log |E|)$ time by binary search in each coordinate.

Once we have chosen $u$, we find the unprocessed vertices in $B_{3\mu}$
as follows. By packing, only constantly many of our disjoint cubes of
side length $\mu$ meet $B_{3\mu}$. We find them in $O(\log |E|)$ time by
binary search on the intervals in each coordinate and then look up the
resulting cubes. We scan their vertices and keep the unprocessed ones
in $B_{3\mu}$.

We inspect the incident edges of these vertices to compute $L$ and
sort the vertices by the radius $r$ at which they enter $B_r$ (i.e., by their $\ell_\infty$ distance from $u$). Starting
at $r=\mu$, we count how many edges contributing to $L$ have exactly
one endpoint in $B_r$. As $r$ increases, we update this count whenever
vertices enter $B_r$, processing vertices at the same distance together.
Once all vertices have been processed, each original edge whose
endpoints were contracted to different vertices gives an edge of
$\Gsimp$ between those vertices. 

When we choose a centre $u$, we process every unprocessed vertex
within $\ell_\infty$ distance $\mu$ of $u$. Thus every later centre
has $\ell_\infty$ distance greater than $\mu$ from $u$.
Fix a cell $C$ of side length $\mu$. We inspect $C$ only if the cube
$B_{3\mu}$ centred at $u$ meets $C$. In that case, $u$ lies in the
cube obtained by extending $C$ by $3\mu$ on each side. This cube has
side length $7\mu$ and can be covered by $7^d$ cubes of side length
$\mu$. Each contains at most one centre, since any two points in it
are at $\ell_\infty$ distance at most $\mu$.
Thus we inspect each cell at most $7^d$ times. Since $d$ is fixed,
we scan each vertex and its incident edges only constantly many times,
taking $O(|V|+|E|)$ time in total. Each graph vertex also appears in only
constantly many lists to be sorted, so sorting the vertices and edges takes $O(|E| \log |E|)$
time in total. Finally, we spend $O(\log(|E|))$ time on binary searches
per centre, and there are at most $|V|\leq |E|+1$ centres. Thus, we may over-estimate the running time as follows:

\begin{figure}[h]
    \centering
    \includegraphics{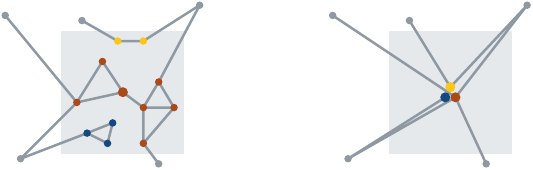}
    \caption{The contraction used during our simplification. The contracted red, blue and yellow vertices are all placed in the centre of the square, but are for illustration purposes placed apart.}
    \label{fig:simpG}
\end{figure}

\begin{lemma}\label{lemma:computingGsimp}
    Computing a $\mu$-simplification $\Gsimp$ takes  $O((n+|E|)\log(n+|E|))$ time.
\end{lemma}

\begin{lemma}
    \label{lemm:triangle}
Fix the $\ell_p$ norm with $1\leq p\leq\infty$. Let $G$ be a connected graph with $|E|$ edges. For any walk $Q$ in $G$, there exists a walk $Q'$ in $\Gsimp$ such that $\fd(Q,Q')\leq 2d^{1/p}\mu$. Conversely, for every walk $Q'$ in $\Gsimp$, there exists a walk $Q$
in $G$ such that $\fd(Q,Q')\leq 2d^{1/p}\mu$. 
\end{lemma}

\begin{proof}
Consider a walk $Q=(v_1,v_2,\ldots,v_k)$ in $G$. 
For each vertex $v\in V(G)$, let $simpl(v)$ denote the vertex in $\Gsimp$ into which $v$ was contracted. 
Define the sequence

\[Q' =(simpl(v_1), simpl(v_2), \ldots, simpl(v_k)). \]

We first show that $Q'$ is a walk in $\Gsimp$ (after removing identical subsequent vertices).
Observe that throughout the simplification procedure, every edge $(w,w')$ in the simplified graph corresponds to an original edge $(v,v')$ in $G$,
 where $v$ is contracted into $w$ and $v'$ is contracted into $w'$. This invariant clearly holds initially, and is preserved whenever an endpoint of an edge is replaced by the vertex to which it is contracted.

Now consider an edge $(v_i,v_{i+1})$ of $Q$. If $simpl(v_i)=simpl(v_{i+1})$, then the two consecutive vertices in $Q'$ are identical.
Otherwise, suppose that $v_i$ was contracted into a cube $B_i$ and $v_{i+1}$ was contracted into a different cube $B_{i+1}$. 
Note that $B_{i+1} \neq B_i$, as this would contradict $simpl(v_i)\neq simpl(v_{i+1})$: since $v_i$ and $v_{i+1}$ are connected by an edge in $G$, they would belong to the same connected component inside $B_i$. 
By the invariant above, the original edge $(v_i,v_{i+1})$ therefore induces the edge $(simpl(v_i),simpl(v_{i+1}))$ in $\Gsimp$. Thus, consecutive vertices in $Q'$ are either the same vertex or connected by an edge of $\Gsimp$, 
and $Q'$ is a walk.

Consider the traversal of an edge $(v_i,v_{i+1})$ in $Q$ simultaneously with the traversal of the corresponding edge $(simpl(v_i),simpl(v_{i+1}))$ in $Q'$.
The Fréchet distance between two line segments is realised by the maximum distance between their corresponding endpoints~\cite[Section 2]{alt1995computing}. 
Since $|v-simpl(v)|\leq 2d^{1/p}\mu$, each pair of corresponding line segments admits a monotone traversal where the distance between the entities is at most $2d^{1/p}\mu$.
These traversals agree at their common endpoints and can therefore be concatenated into a monotone traversal of $Q$ and $Q'$, yielding $\fd(Q,Q')\leq 2d^{1/p}\mu$.

For the converse, consider a walk $Q'=(w_1,w_2,\ldots,w_k)$ in $\Gsimp$. 
In the special case that $Q'$ consists of a single vertex $w_1$, we
choose any vertex $v$ of $G$ that was contracted to $w_1$ and take
$Q=(v)$. The contraction moves $v$ by at most $2d^{1/p}\mu$, so
$\fd(Q,Q')\leq 2d^{1/p}\mu$.

Now suppose that $Q'$ has at least one edge. For each vertex $w_i$
of $Q'$, let $C_i$ be the component contracted to $w_i$ and let $B_i$
be the cube used in that contraction.
By the invariant above, for every edge $(w_i,w_{i+1})$ of $Q'$, there exists an original edge $(u_i,z_i)$ in $G$ with $u_i\in C_i$ and $z_i\in C_{i+1}$.

We construct a walk $Q$ in $G$ corresponding to $Q'$ as follows.
Start at $u_1$ and traverse the edge $(u_1,z_1)$.
For every edge $(u_i,z_i)$ corresponding to $(w_i,w_{i+1})\in Q'$, with $i\geq 2$, assume inductively that the walk $Q$ constructed thus far ends at some vertex $z_{i-1}\in C_i$, where $C_i$ is the maximal contracted connected component corresponding to $w_i$.
We choose any walk inside $C_i$ from $z_{i-1}\in C_i$ to $u_i\in C_i$, and then traverse the edge $(u_i,z_i)$, ending at $z_i$ and thereby restoring the induction invariant.
By induction, this yields a walk $Q$ in $G$ with the following properties:
(a) for every edge $(w_i,w_{i+1})$ of $Q'$, the walk $Q$ traverses a corresponding original edge $(u_i,z_i)$ with $u_i\in C_i$ and $z_i\in C_{i+1}$;
(b) between two such edges, $Q$ follows a walk entirely inside $C_i$, connecting $z_{i-1}$ to $u_i$.

It remains to bound the Fréchet distance by constructing a traversal between $Q'$ and $Q$.
By property~(b), between the original edges $(u_{i-1},z_{i-1})$ and $(u_i,z_i)$, the walk $Q$ follows a subwalk from $z_{i-1}$ to $u_i$ entirely inside $C_i\subset B_i$.
While one entity remains stationary at $w_i$ and the other traverses this subwalk, their distance is at most  $2d^{1/p}\mu$, since every point of the subwalk lies inside $B_i$.

Whenever $Q$ traverses an original edge $(u_i,z_i)$ corresponding to $(w_i,w_{i+1})\in Q'$, we simultaneously traverse the edge $(w_i,w_{i+1})$.
By property~(a), $u_i\in C_i$ and $z_i\in C_{i+1}$, and hence $|u_i-w_i|,|z_i-w_{i+1}|\leq 2d^{1/p}\mu$.
Since the Fréchet distance between two line segments is realised by the maximum distance between their corresponding endpoints~\cite[Section 2]{alt1995computing}, these two edges admit a monotone traversal whose distance never exceeds $2d^{1/p}\mu$

Concatenating the traversals given by properties~(a) and~(b) yields a traversal of $Q$ and $Q'$ of width at most $2d^{1/p}\mu$. Thus, $\fd(Q,Q')\leq 2d^{1/p}\mu$.
\end{proof}

\subsubsection{An approximate decider}

Now that we have our simplified graph, we can create our (1+$\eps$)-approximate decider.
To this end, we first make our analogous definition to Definition~\ref{def:worstcase_complexity}.
We then upper bound the complexity of this worst-case relative free-space complexity.

\begin{definition}
      For a curve $P$ and a connected graph $G$  and any $\eps > 0$, we define the worst-case relative free-space surface complexity as
     $\graphFS(\eps,P,G):=\sup_{\Delta>0}\graphFSDelta(P^{(\eps\Delta)},G^{(\eps\Delta)})$.
\end{definition}

\begin{lemma}\label{lemma:MapmatchingFSSComplexity}
Let $P$ be a $\lambda$-low dense $n$-vertex polygonal curve and $G$ a connected $\lambda$-low dense graph with $|E|$ edges. For $0<\eps\leq1$, we have $\graphFS(\eps,P,G)\in O(\eps^{-2}\lambda^{2/d}(n+|E|)^{2-2/d})$.
\end{lemma}

\begin{proof}
Fix $\Delta>0$ and set $\mu=\eps\Delta$. Consider a grid of side length $\Delta$ and fix a grid cell $R$.
Let $3R$ be the concentric cube of side length $3\Delta$, and let
$q_R$ count the edges of $\Psimp$ and $\Gsimp$ intersecting $3R$. 
We first show that the original edges have total length $\Omega((q_R-1)\mu)$ in a constant-factor enlargement of $R$:

\begin{claim}
The original edges of $P$ and $G$ have total length
$\Omega((q_R-1)\mu)$ inside the cube $C_R$ concentric with $R$ of side
length $19\Delta$.
\end{claim}

Each simplified curve edge intersecting $3R$ represents a subcurve of $P$ that comes within distance $\mu$ of $3R$. Except possibly for the final edge, this subcurve has length at least $\mu$. Thus it contributes length $\mu$ nearby, and the subcurves have disjoint interiors in the parameter space of $P$. 
For each edge $uv$ of $\Gsimp$ intersecting $3R$, choose an edge $xy$
of $G$ such that $x$ was contracted to $u$ and $y$ to $v$. Such an
edge exists by the definition of $\Gsimp$. Since each original vertex
is contracted to a single vertex, distinct edges of $\Gsimp$ yield
distinct choices of $xy$. If its length is at least $\mu$, it too contributes length $\mu$ nearby, since each endpoint moved by at most $2\mu$ in each coordinate.

Otherwise, $xy$ has length less than $\mu$. Let $u$ be the centre of
the first contraction that processes either endpoint, and suppose it
processes $x$. Then $y$ lies outside $B_r$, since otherwise the edge $xy$ would
place $x$ and $y$ in the same component and they would be contracted
to the same vertex.
Both endpoints lie in $B_{3\mu}(u)$, since $x$ is within $2\mu$ of $u$
in $\ell_\infty$ and $y$ is within $\mu$ of $x$.
Fix a centre $u$. Let $k$ be the number of chosen original edges of
length less than $\mu$ for which both endpoints are still unprocessed
when we choose $u$ and exactly one endpoint lies in $B_r$. 
Each has exactly one
endpoint in $B_r$, so the stopping rule gives $k\leq L/\mu$.
Thus the original edges counted by $L$ have total length at least
$k\mu$ in $B_{3\mu}(u)$.
Each of the $k$ edges of $\Gsimp$ has an endpoint at $u$ and meets
$3R$. Its $\ell_\infty$ length is at most $5\mu$, since $xy$ has
length less than $\mu$ and each endpoint moves by at most $2\mu$ in
$\ell_\infty$. Thus $u$ is within $5\mu$ of $3R$, and $B_{3\mu}(u)$
lies in a constant enlargement of $R$. The centres are more than
$\mu$ apart in $\ell_\infty$, so each point lies in only constantly
many such cubes. Summing over the centres proves the claim.

Let $C_R$ be the enlarged cube from the claim, and let $m_R$ count
the original edges of $P$ and $G$ of length at most $\Delta$ that
intersect $C_R$. Together, the edges of $P$ and $G$ are
$2\lambda$-low dense. By \Cref{lem:app-l1-length}, whose packing
proof applies in any fixed $\ell_p$ norm, the full lengths of the
$m_R$ edges sum to $O(\Delta\lambda^{1/d}m_R^{1-1/d})$.
Covering $C_R$ by constantly many balls of radius $\Delta$ shows
that only $O(\lambda)$ edges of length greater than $\Delta$ intersect
$C_R$. Each has length $O(\Delta)$ inside $C_R$, since $C_R$ is convex
and has diameter $O(\Delta)$. Comparing these bounds with the claim
and substituting $\mu=\eps\Delta$ gives
$
    q_R=O\bigl((\lambda+\lambda^{1/d}m_R^{1-1/d})/\eps\bigr)$.
Moreover, $\sum_R m_R=O(n+|E|)$, since each original edge of length at
most $\Delta$ intersects only constantly many cubes $C_R$.

Every relevant cell has a free point on one of its sides. Charge the cell to the corresponding endpoint of one of its two edges, counting endpoints once per incident edge. If the endpoint lies in $R$, the other edge intersects $3R$, so the endpoint receives at most $q_R$ charges. Cells $R$ with $q_R\leq K\lambda/\eps$ therefore receive $O((n+|E|)\lambda/\eps)$ charges in total. For $K$ sufficiently large, every remaining $R$ satisfies $q_R=O(\lambda^{1/d}m_R^{1-1/d}/\eps)$. At most $2q_R$ endpoints lie in $R$, so it receives at most $2q_R^2$ charges. Since $d\geq2$ and we may take $\lambda\leq n+|E|$, the total is
\[
 O\left(\frac{(n+|E|)\lambda}{\eps}+
 \frac{\lambda^{2/d}}{\eps^2}\sum_R m_R^{2-2/d}\right)
 =O\left(\frac{\lambda^{2/d}(n+|E|)^{2-2/d}}{\eps^2}\right). \qedhere
\]
\end{proof}

We will, likewise \cite{DriemelHW12}, show how to invoke the exact $\Decider$ with a simplified curve and graph to yield a $(1+\eps)$-$\ApproxDecider$.

\begin{lemma}
\label{lemma:fuzzydeciderMapMatching}
Fix any $\ell_p$ norm. Let $P$ be a $\lambda$-low dense polygonal curve of complexity $n$ and $G$ be a connected $\lambda$-low dense graph with $|E|$ edges in $\mathbb{R}^d$.
Let $\Delta>0$ and $0<\eps\leq 1$. Then we can construct an algorithm $\ApproxDecider(P,G,\Delta,\eps)$ that in $O(\eps^{-2} \lambda^{2/d} (n+|E|)^{2-2/d} \log(n+|E|))$ time outputs one of the following:
\begin{enumerate}
\item $\mapFD GPQ \le \Delta(1+\eps)$.
\item $\mapFD GPQ> \Delta$.
\end{enumerate}
\end{lemma}

\begin{proof}
Set $\mu=\eps\Delta/(2(1+2d^{1/p}))$. Compute $\Psimp$ in $O(n)$
time and $\Gsimp$ in $O((n+|E|)\log(n+|E|))$ time by \Cref{lemma:computingGsimp}.
By Observation~\ref{obs:simplified_distance}, $\fd(P,\Psimp)\leq\mu$.
By \Cref{lemm:triangle}, each walk in either graph has a walk in the
other graph at Fréchet distance at most $2d^{1/p}\mu$. Thus, simplifying $P$ and replacing a walk $Q$ in $G$ by its walk
$Q'$ in $\Gsimp$ changes the Fréchet distance by at most
$
    \bigl|\fd(P,Q)-\fd(\Psimp,Q')\bigr|
    \leq \fd(P,\Psimp)+\fd(Q,Q')
    \leq (1+2d^{1/p})\mu
    = \frac{\eps\Delta}{2}$. 
Set $\Delta'=\Delta+\eps\Delta/2$ and call
$\Decider(\Psimp,\Gsimp,\Delta')$.

If the decider returns \textsc{Yes}, there is a walk $Q'$ in $\Gsimp$
with $\fd(\Psimp,Q')\leq\Delta'$. Choose the walk $Q$ in $G$ given
by \Cref{lemm:triangle}. The triangle inequality gives
\[
    \fd(P,Q)\leq\fd(P,\Psimp)+\fd(\Psimp,Q')+\fd(Q',Q)
    \leq\Delta'+\eps\Delta/2=(1+\eps)\Delta.
\]
Thus we return $\mapFD GPQ\leq(1+\eps)\Delta$.

If the decider returns \textsc{No}, suppose that some walk $Q$ in $G$
satisfies $\fd(P,Q)\leq\Delta$. By \Cref{lemm:triangle}, there is a
walk $Q'$ in $\Gsimp$ with
\[
    \fd(\Psimp,Q')\leq\fd(\Psimp,P)+\fd(P,Q)+\fd(Q,Q')
    \leq\Delta+\eps\Delta/2=\Delta'.
\]
The decider would have returned \textsc{Yes}. Thus we return
$\mapFD GPQ>\Delta$.

If $\Gsimp$ consists of a single vertex $v$, we test whether every
vertex of $\Psimp$ is within distance $\Delta'$ of $v$. If $\Psimp$ consists of a single point, we test whether
some vertex of $\Gsimp$ is within distance $\Delta'$ of that point.
Either test takes $O(n+|E|)$ time.
For the remaining  input instances, apply \Cref{lemma:MapmatchingFSSComplexity}
at threshold $\Delta'$ with simplification parameter $\mu$. Since
$
    \frac{\mu}{\Delta'}=\frac{\eps}{(2+\eps)(1+2d^{1/p})}
    =\Theta(\eps)$,
this lemma implies that there are $O(\eps^{-2}\lambda^{2/d}(n+|E|)^{2-2/d})$ relevant cells.
By \Cref{thm:MapmatchingDecider}, the exact decision takes $O(n+|E|)$ time
plus $O(\log(n+|E|))$ time per relevant cell. Including the construction
of the simplifications, the total time is
$O(\eps^{-2}\lambda^{2/d}(n+|E|)^{2-2/d}\log(n+|E|))$.
\end{proof}

\paragraph{A small tangent: TwoCallsApproxDecider}
If we want to use our approximate decider for binary search purposes, we make a slight modification. We use two calls to \ApproxDecider to decide how to proceed when
a binary search tests a value $\Delta$. If the calls establish that
$\mapFD GPQ\leq\Delta$, we keep $\Delta$ as an upper bound. If they
establish that $\mapFD GPQ>\Delta$, we keep $\Delta$ as a lower bound.
Otherwise, we return a $(1+\eps)$-approximation to $\mapFD GPQ$ and
stop the search. We call this procedure \TwoApproxDecider.

\begin{lemma}
\label{lemma:ApproxdeciderMapmatching}
    Fix the $\ell_p$ norm with $1\leq p\leq\infty$. Let $P$ be a $\lambda$-low dense polygonal curve of complexity $n$ and $G$ be a connected $\lambda$-low dense graph with $|E|$ edges in $\mathbb{R}^d$. Let $\varepsilon \in (0, 1]$ and $\Delta>0$ be parameters. There is an algorithm $\TwoApproxDecider(P, G, \Delta, \varepsilon)$ that in $O(\ApproxDeciderRuntime{\eps})$ time, returns either:
    \begin{enumerate}
        \item  A $(1+\varepsilon)$-approximation to $\min _{Q \in G} \fd(P, Q)$
        \item $\min _{Q \in G} \fd(P, Q)\le\Delta$.
        \item $\min _{Q \in G} \fd(P, Q)>\Delta$.
    \end{enumerate}
\end{lemma}

\begin{proof}
 Set $\gamma = \eps/3$ and $\Delta' = \Delta/(1+\gamma)$.
    If \textsc{ApproxDecider}$(P,G,\Delta',\gamma)$ returns
    $\min_{Q\in W(G)} FD(P,Q) \leq (1+\gamma)\Delta' = \Delta$,
    we return this result.
    If \textsc{ApproxDecider}$(P,G,\Delta,\gamma)$ returns
    $\min_{Q\in W(G)} FD(P,Q) > \Delta$, we return this result.
    Otherwise
    \[
        \Delta' < \min_{Q\in W(G)} FD(P,Q)
        \leq (1+\gamma)\Delta.
    \]
    We return $(1+\gamma)\Delta$, which is at most
    $(1+\gamma)^2$ times the optimum and hence within
    a factor $1+\eps$.
\end{proof}

\subsection{Computing a \texorpdfstring{$(1+\eps)$}{(1+epsilon)}-approximation}

To compute a $(1+\eps)$-approximation of $\mapFD GPQ$,
we follow a slightly modified version of the framework
of~\cite{Chen2011approximate}. The algorithm proceeds in
two phases. First, we compute a constant-factor interval
$[\alpha,\beta]$ containing the optimal map-matching distance.
We then refine this interval using \textsc{ApproxDecider}
until $\beta\leq(1+\eps)\alpha$, yielding the desired
$(1+\eps)$-approximation.

\subparagraph{Finding constant-factor interval} Let $\delta_2$ be the minimum Euclidean Fréchet distance between $P$
and a walk in $G$. We follow Algorithm~3.1 of~\cite{Chen2011approximate}.
Using Lemma~3.8 of~\cite{Chen2011approximate} and $O(\log(n+|E|))$ calls
to \textsc{TwoCallsApproxDecider} with $\eps=1$, we either obtain a
$2$-approximation to $\delta_2$ or an interval $[a,b]$ containing
$\delta_2$. If $b\leq 4a$, we already have a constant-factor interval.
Otherwise, the guarantee of Lemma~3.8 implies that no distance between
two input vertices lies in $(a,b)$.

In this case, compute the greedy $(b/2)$-simplification of $P$ and
leave $G$ unchanged. Every vertex skipped by the simplification is
within distance $a$ of the last vertex chosen, since its distance is
less than $b/2$. Thus $\fd(P,P^{b/2})\le a$. Every edge of the simplification, except possibly its final
edge, has length at least $b$.

We bound the number of non-empty free-space cells between
$P^{b/2}$ and $G$ at threshold $b+a$.
The final edge of $P^{b/2}$ accounts for at most $|E|$ cells.
Every other simplified edge $f$ corresponds to a distinct
original edge $f'$ ending at the same vertex, where
$FD(f,f')\leq a$ and $\ell(f')\geq b-a$.
Thus, each remaining non-empty cell corresponds to a distinct
pair of original edges at distance at most
$b+2a\leq 2(b-a)$, of which at least one has length at
least $b-a$.
Charge each such pair to its shorter edge $e$ and set
$s=\max\{b-a,\ell(e)\}$.
The longer edge has length at least $s$ and intersects
a ball of radius $O(s)$ centred at an endpoint of $e$.
By low density, each edge receives $O(\lambda)$ charges.
Hence, the number of non-empty cells is
$O(\lambda(n+|E|))$.

Since each relevant free-space cell contributes at most four vertex--edge critical values, there are $O(\lambda(n+|E|))$ such values. Following Algorithm~3.1 of~\cite{Chen2011approximate}, we binary search these critical values using \textsc{Decider}. By Theorem~56, each call takes $O(\lambda(n+|E|)\log(n+|E|))$ time.
As the binary search requires $O(\log(n+|E|))$ calls, the total running time is $O(\lambda(n+|E|)\log^2(n+|E|))$.

By Lemma~3.7 and the analysis of Lines~5--11 of
Algorithm~3.1 of~\cite{Chen2011approximate}, the search
yields a $3$-approximation $y$ of the optimal map-matching
distance between $P^{b/2}$ and $G$.
Since $\fd(P,P^{b/2})\leq a$ and $a\leq\delta_2$, we have 
    $\delta_2 \leq y+a
    \leq 3(\delta_2+a)+a
    \leq 7\delta_2.$
Thus, we obtain a constant-factor interval containing
$\delta_2$, namely $[(y+a)/7,y+a]$.

Combining the running times of the two searches, we obtain
a constant-factor interval $[L,U]$ containing $\delta_2$ in
$O(\lambda^{2/d}(n+|E|)^{2-2/d}\log^2(n+|E|))$ time,
since $\lambda(n+|E|)\leq
\lambda^{2/d}(n+|E|)^{2-2/d}$.

Under the Euclidean
norm we set $[\alpha,\beta]=[L,U]$.
For any other $\ell_p$ norm, we set
$[\alpha,\beta]=[L/\sqrt{d},\sqrt{d}U]$, since the two
norms differ by at most a factor $\sqrt{d}$.
In fixed dimension, this remains a constant-factor interval,
and changing between these norms changes the low-density
parameter by at most a constant factor.

\subparagraph{Searching through constant-factor interval} 
We now cover how we search for a $(1+\eps)$ approximation, using \ApproxDecider in the constant-factor interval $[\alpha,\beta]$.

\begin{lemma} \label{lemma:searchingConstantFactorInterval}
    Fix the $\ell_p$ norm with $1\leq p\leq\infty$. Let $P$ be a $\lambda$-low dense polygonal curve of complexity $n$ and $G$ be a connected $\lambda$-low dense graph with $|E|$ edges in $\mathbb{R}^d$. Given $0<\eps\leq 1$ and a constant-factor interval $[\alpha,\beta]$ with $0<\alpha\leq\mapFD GPQ\leq\beta$, one can obtain a $(1+\eps)$-approximation to $\mapFD GPQ$ in $O(\ApproxDeciderRuntime{\eps})$ time.
\end{lemma}

\begin{proof}
We show how to halve the interval $[\alpha,\beta]$ using 2 calls to $\ApproxDecider$. Fix $c\geq 2$ with $\beta\leq c\alpha$ for the initial interval. Set $\delta = (\beta-\alpha)/4$ and $\eps' =\delta/(\alpha c)$. Set $x_i = \alpha+i\delta$ for $i=1,2,3$ to divide the interval $[\alpha,\beta]$ into 4 equal sections. Note that as $\beta\le c\alpha$, we have $\delta = (\beta-\alpha)/4 \le c\alpha/4$ and thus $\frac{\delta^2}{\alpha c}\le \frac{\delta}{4}$.

If $\ApproxDecider(P,G,x_1, \eps')$ returns 
$\mapFD GPQ \le (1+\eps') x_1 
=\left(1+\frac{\delta}{c\alpha}\right)(\alpha+\delta)
=\alpha+\delta + \frac{\delta}{c}+\frac{\delta^2}{c\alpha}
\leq \alpha+\delta +\frac{\delta}{c}+\frac{\delta}{4}
\leq \alpha+2\delta
=x_2$,  we update $[\alpha,\beta]$ to $[\alpha,x_2]$.

Else if $\ApproxDecider(P,G,x_2, \eps')$ returns 
$\mapFD GPQ \le (1+\eps') x_2
=\left(1+\frac{\delta}{c\alpha}\right)(\alpha+2\delta)
=\alpha+2\delta + \frac{\delta}{c}+2\frac{\delta^2}{c\alpha}
\leq \alpha+2\delta +\frac{\delta}{c}+\frac{\delta}{2}
\leq \alpha+3\delta
=x_3$, update $[\alpha,\beta]$ to $[x_1,x_3]$

Otherwise, if $\ApproxDecider(P,G,x_2, \eps')$ returns 
$\mapFD GPQ > x_2$ we simply update $[\alpha,\beta]$ to $[x_2,\beta]$.
Thus, we can half the interval $[\alpha,\beta]$ using 2 calls to $\ApproxDecider$.

Let $[\alpha_i,\beta_i]$ be the interval of the $i$th iteration. We continue to half the interval until $\beta_i\le (1+{\eps}) \alpha_i$, meaning we have found a $(1+\eps)$-approximation.

As $\beta-\alpha \in O(\alpha)$ it is only necessary to half the interval $[\alpha,\beta]$ at most $M \in O( \log(1/\eps) )$ times to find a $(1+\eps)$-approximation.
Let $\eps'_i$ be the accuracy used in round $i$. The total running time is
\[
    O\left(\lambda^{2/d}(n+|E|)^{2-2/d}\log(n+|E|)
        \sum_{i=1}^{M}\frac{1}{(\eps'_i)^2}\right).
\]
Note that $\eps'_{i+1}=\frac{\delta_{i+1}}{c\alpha_{i+1}}\le \frac{\delta_i/2}{c\alpha_i}=\frac{\eps'_i}{2}$.
Thus, $1/(\eps'_i)^2$ grows by at least a factor four per round. If $j$ is the last round, then $\beta_j-\alpha_j>\eps\alpha_j$, so $\eps'_j=(\beta_j-\alpha_j)/(4c\alpha_j)>\eps/(4c)$. Summing the running times of the decider therefore takes $O(\ApproxDeciderRuntime{\eps})$ time. We return the final upper endpoint $\beta$, which lies between the optimum and $(1+\eps)$ times the optimum.
\end{proof}

This yields the main theorem:

\begin{theorem}\label{thm:MapmatchingUpperbound}
    Let $P$ be a $\lambda$-low dense polygonal curve of complexity $n$ and $G$ be a connected $\lambda$-low dense graph with $|E|$ edges in $\mathbb{R}^d$.
    A $(1+\eps)$ approximation of $\mapFD GPQ$ can be calculated in time $O\left(\lambda^{2/d}(n+|E|)^{2-2/d}\log(n+|E|)\bigl(\eps^{-2}+\log(n+|E|)\bigr)\right)$.
\end{theorem}

\begin{proof}
Finding the constant-factor interval takes $O(\lambda^{2/d}(n+|E|)^{2-2/d}\log^2(n+|E|))$ time. By \Cref{lemma:searchingConstantFactorInterval}, searching this interval in the chosen norm takes $O(\ApproxDeciderRuntime{\eps})$ time. Adding these two bounds proves the theorem.
\end{proof}

\subsection{Lower bound}

Lastly we will show that the algorithm is tight up to a factor $\eps^{4/d}$ and logarithmic factors.

\begin{theorem}\label{thm:MapMatching3up-approximation-hardness}
    The following lower bound holds for directed graphs, or, for when the mapped path cannot make U-turns at a vertex.
    Unless the Orthogonal Vector Hypothesis fails, for any $d\geq 3$ and any $\eps\in \left(d^{-\frac{d-1}{d-2}} \, \left(\frac{\lambda}{n}\right)^{\frac{1}{d-2}},1\right]$, there is no algorithm to $(1+\eps)$-approximate the Fréchet distance between size-$n$ $d$-dimensional $\lambda$-low dense curve and a connected $\lambda$-low dense graph with $|E|$ edges in time $O\left(\left( \,  \frac{\lambda^{2/d}(n+|E|)^{2(1-1/d)}}{\eps^{2(1-2/d)}}\right)^{1-\delta}\right)$ with $\delta > 0$ even when $|E|=n$.
\end{theorem}

\begin{proof}
   We reduce from the lower bound established in \Cref{thm:3up-approximation-hardness} for computing the Fréchet distance between two $\lambda$-paths in a $d$-dimensional grid graph. 

   We construct a graph $G_Q$ from $Q$ by viewing the vertices and edges of $Q$ as a graph. Map matching only asks for the best walk in $G_Q$, whereas the Fréchet distance in \Cref{thm:3up-approximation-hardness} compares against the specific traversal $Q$.
   
   To force the entire graph traversal we extend $G_Q$ and $P$ with a new start vertex and ending vertex placed sufficiently far way from the construction of the curves from \Cref{thm:3up-approximation-hardness}.

   Thus deciding if $\mapFD {G_Q} P {Q'} \le \Delta$ is the same as deciding $\fd(P,Q) \le \Delta$. 
\end{proof}
\end{document}